\documentclass[a4paper,UKenglish,thm-restate]{lipics-v2021}
\pdfoutput=1 %
\hideLIPIcs  %

\title{Specification and Automatic Verification of Computational Reductions} %

\author{Julien Grange}{LACL, Université Paris-Est Créteil, France}{julien.grange@lacl.fr}{https://orcid.org/0009-0005-0470-1781}{}
\author{Fabian Vehlken}{Ruhr University Bochum, Germany}{fabian.vehlken@rub.de}{https://orcid.org/0009-0002-1434-3672}{Supported by the Deutsche Forschungsgemeinschaft (DFG, German Research Foundation), grant 448468041.}
\author{Nils Vortmeier}{Ruhr University Bochum, Germany}{nils.vortmeier@rub.de}{https://orcid.org/0009-0000-2821-7365}{}
\author{Thomas Zeume}{Ruhr University Bochum, Germany}{thomas.zeume@rub.de}{https://orcid.org/0000-0002-5186-7507}{Supported by the Deutsche Forschungsgemeinschaft (DFG, German Research Foundation), grant 448468041.}

\authorrunning{J. Grange,  F. Vehlken, N. Vortmeier, and T. Zeume} %

\Copyright{Julien Grange, Fabian Vehlken, Nils Vortmeier, and Thomas Zeume} %



\nolinenumbers %

\EventEditors{Rastislav Kr\'{a}lovi\v{c} and Anton\'{i}n Ku\v{c}era}
\EventNoEds{2}
\EventLongTitle{49th International Symposium on Mathematical Foundations of Computer Science (MFCS 2024)}
\EventShortTitle{MFCS 2024}
\EventAcronym{MFCS}
\EventYear{2024}
\EventDate{August 26--30, 2024}
\EventLocation{Bratislava, Slovakia}
\EventLogo{}
\SeriesVolume{306}
\ArticleNo{48}

\usepackage{graphicx}
\usepackage{xspace}
\usepackage{array}
\usepackage{tikz}
\usetikzlibrary{arrows,decorations.pathmorphing,backgrounds,calc,positioning,fit,petri,matrix,backgrounds, decorations.pathreplacing, shapes.geometric}
\usepackage{tcolorbox}
\tcbuselibrary{skins,breakable,xparse,raster}
\usepackage{adjustbox}

\usepackage{macros_julien}
\usepackage{placeins}
\usepackage{float}

\colorlet{newNodeEdge}{blue!85!black}
\colorlet{oldNodeEdge}{gray}

\tikzstyle{originnode} = [draw, circle, fill=black, inner sep=1.5pt]
\tikzstyle{originedge} = [draw, thick, shorten <=1pt, shorten >=1pt]

\tikzstyle{targetnode-new} = [draw, circle, fill=newNodeEdge, inner sep=1.5pt]
\tikzstyle{targetnode-old} = [draw, circle, fill=oldNodeEdge, inner sep=1.5pt]
\tikzstyle{targetedge-new} = [draw = newNodeEdge, fill=newNodeEdge, thick, shorten <=1pt, shorten >=1pt]
\tikzstyle{targetedge-old} = [draw = oldNodeEdge, fill=oldNodeEdge, thick, shorten <=1pt, shorten >=1pt]

\tikzstyle{globalnode-new} = [draw, circle, fill=newNodeEdge, inner sep=1.5pt]
\tikzstyle{globalnode-old} = [draw, circle, fill=oldNodeEdge, inner sep=1.5pt]
\tikzstyle{globaledge-new} = [draw = newNodeEdge, fill=newNodeEdge, thick, shorten <=1pt, shorten >=1pt]
\tikzstyle{globaledge-old} = [draw = oldNodeEdge, fill=oldNodeEdge, thick, shorten <=1pt, shorten >=1pt]

\colorlet{cbglobalcol}{green!50!black!70}
\tikzstyle{cbglobalnode} = [draw, circle, fill=cbglobalcol, inner sep=1.5pt]
\tikzstyle{cbglobaledge} = [draw = cbglobalcol, fill=cbglobalcol, thick, shorten <=1pt, shorten >=1pt]

\colorlet{cbnodetypecol}{oldNodeEdge}
\tikzstyle{cbnodetypenode} = [draw, circle, fill=cbnodetypecol, inner sep=1.5pt]
\tikzstyle{cbnodetypeedge} = [draw = cbnodetypecol, fill=cbnodetypecol, thick, shorten <=1pt, shorten >=1pt]

\colorlet{cbedgetypecol}{newNodeEdge}
\tikzstyle{cbedgetypenode} = [draw, circle, fill=cbedgetypecol, inner sep=1.5pt]
\tikzstyle{cbedgetypeedge} = [draw = cbedgetypecol, fill=cbedgetypecol, thick, shorten <=1pt, shorten >=1pt]
\tikzstyle{freshnode} = [draw, circle, fill=cbedgetypecol, inner sep=1.5pt]
\tikzstyle{freshedge} = [draw = cbedgetypecol, fill=cbedgetypecol, thick, shorten <=1pt, shorten >=1pt]

\tikzstyle{originnode-small} = [draw, circle, fill=black, inner sep=1pt]
\tikzstyle{originedge-small} = [draw, thick, shorten <=0.5pt, shorten >=0.5pt]

\tikzstyle{targetnode-old-small} = [draw, circle, fill=oldNodeEdge, inner sep=1pt]
\tikzstyle{targetnode-new-small} = [draw, circle, fill=newNodeEdge, inner sep=1pt]
\tikzstyle{targetedge-old-small} = [draw = oldNodeEdge, fill=oldNodeEdge, thick, shorten <=0.5pt, shorten >=0.5pt]
\tikzstyle{targetedge-new-small} = [draw = newNodeEdge, fill=newNodeEdge, thick, shorten <=0.5pt, shorten >=0.5pt]
\tikzstyle{globalnode-new-small} = [draw, circle, fill=newNodeEdge, inner sep=1pt]
\tikzstyle{globalnode-old-small} = [draw, circle, fill=oldNodeEdge, inner sep=1pt]
\tikzstyle{globaledge-new-small} = [draw = newNodeEdge, fill=newNodeEdge, thick, shorten <=0.5pt, shorten >=0.5pt]
\tikzstyle{globaledge-old-small} = [draw = oldNodeEdge, fill=oldNodeEdge, thick, shorten <=0.5pt, shorten >=0.5pt]

\tikzstyle{targetnode-old-big} = [draw, circle, fill=oldNodeEdge, inner sep=3pt]
\tikzstyle{originnode-bigger} = [draw, circle, fill=black, inner sep=2.5pt]

\tikzstyle{vgnode} = [draw, circle, inner sep=1pt]
\tikzstyle{vgnode-big} = [draw, circle, inner sep=3pt]
\tikzstyle{vgedge} = [draw, shorten <=1pt, shorten >=1pt]

\tikzstyle{hcdnode} = [draw, circle, inner sep=1.5pt] %
\tikzstyle{hcdnode-big} = [draw, circle, inner sep=3pt]
\tikzstyle{hcdedge} = [draw, -stealth, shorten <=1pt, shorten >=1pt]
\tikzstyle{hcunode} = [draw, circle, fill=oldNodeEdge, inner sep=1.5pt] %
\tikzstyle{hcunode-big} = [draw, circle, fill=oldNodeEdge, inner sep=3pt] 
\tikzstyle{hcunode-narrow} = [draw, circle, fill=oldNodeEdge, inner sep=0.5pt] %
\tikzstyle{hcuedge} = [draw = oldNodeEdge, fill=oldNodeEdge, thick, shorten <=1pt, shorten >=1pt] %

\tikzstyle{crossedge} = [draw = newNodeEdge, fill=newNodeEdge, shorten <=0.5pt, shorten >=0.5pt] %
\tikzstyle{witnessedge} = [draw = orange!70, line width=3pt, shorten <=0.5pt, shorten >=0.5pt]

\newcommand{\postoneg}{\ensuremath{\oplus \mapsto \ominus}\xspace}
\newcommand{\negtopos}{\ensuremath{\ominus \mapsto \oplus}\xspace}

\newcommand{\qdFO}[1][k]{\ensuremath{\FO_{#1}}}
\newcommand{\qdMSO}[1][k]{\ensuremath{\MSO_{#1}}}
\newcommand{\univ}[1]{\ensuremath{\mathrm{dom}(#1)}}
\newcommand{\univcard}[1]{\ensuremath{|\univ{#1}|}}
\newcommand{\fresh}[1]{\ensuremath{\#_\mathrm{fresh}(#1)}}

\newcommand{\SAT}{\algorithmicProblem{SAT}}
\newcommand{\Clique}{\algorithmicProblem{Clique}}
\newcommand{\IS}{\algorithmicProblem{IndependentSet}}
\newcommand{\VC}{\algorithmicProblem{VertexCover}}
\newcommand{\FVS}{\algorithmicProblem{FeedbackVertexSet}}
\newcommand{\HCd}{\algorithmicProblem{HamCycle$_d$}}
\newcommand{\HCu}{\algorithmicProblem{HamCycle$_u$}}
\newcommand{\ThreeClique}{\algorithmicProblem{$3$-Clique}}
\newcommand{\FourClique}{\algorithmicProblem{$4$-Clique}}

\newcommand{\isotypeGlobal}{\isotype_{\emptyset}}
\newcommand{\isotypeVertex}{\isotype_{\subscriptVertexType}}
\newcommand{\isotypeEdge}{\isotype_{\subscriptEdgeType}}
\newcommand{\isotypeNonEdge}{\isotype_{\subscriptNonEdgeType}}
\newcommand{\isotypeDirEdge}{\isotype_{\subscriptDirEdgeType}}

\newcommand{\AllTypes}{\mathfrak{T}}
\newcommand{\isotype}{\ensuremath{\mathfrak{t}}\xspace}
\newcommand{\subtype}{\ensuremath{\mathfrak{tp}}}
\newcommand{\redS}{\ensuremath{\mathfrak{S}}}
\newcommand{\Aut}{\text{Aut}}

\newcommand{\tpl}{\bar}
\newcommand{\arity}{\ensuremath{\text{Ar}}}

\newcommand{\schema}{\ensuremath{\sigma}\xspace}

\newcommand{\restrict}[2]{#1[#2]}

\newcommand{\dom}{\ensuremath{U}}
\newcommand{\struc}{\ensuremath{\mathcal{S}}}

\newcommand{\N}{\ensuremath{\mathbb{N}}}

\newcommand{\df}{\ensuremath{\mathrel{\smash{\stackrel{\scriptscriptstyle{
    \text{def}}}{=}}}} \;}

\newcommand  {\myclass} [1]  {\ensuremath{\textsf{\upshape #1}}}

\newcommand{\StaClass}[1]{\myclass{#1}\xspace}

\newcommand  {\algorithmicProblem} [1] {\textnormal{\textsc{#1}}\xspace}

\newcommand{\algorithmicProblemDescription}[4][10cm]{
    \def\Name{#2}
    \def\Input{#3}
    \def\Question{#4}
      \setlength{\tabcolsep}{1mm}
      \begin{tabular}{rp{#1}r}%
      \textit{Problem:}&\algorithmicProblem{\Name} \\
     \textit{Input:}&\Input \\
     \textit{Question:}&\Question
     \end{tabular}%
    }

\newcommand     {\LOGSPACE}     {\StaClass{LogSpace}}
\newcommand     {\PTIME}    {\myclass{PTime}}
\newcommand     {\NP}   {\StaClass{NP}}

\newcommand     {\AC}   {\StaClass{AC}}
\newcommand     {\TC}   {\myclass{TC}}

\newcommand{\FO}{\StaClass{FO}}

\newcommand{\MSO}{\StaClass{MSO}}

\newcommand{\QFO}[1][\quant]{\StaClass{\ensuremath{#1}FO}}
\newcommand{\EFO}{\QFO[\exists^*]}

\newcommand{\CQ}[1][]{\StaClass{CQ}}
\newcommand{\UCQ}[1][]{\StaClass{UCQ}}
\newcommand{\CQneg}[1][]{\StaClass{CQ\ensuremath{^{\mneg}}}}
\newcommand{\UCQneg}[1][]{\StaClass{UCQ\ensuremath{^{\mneg}}}}

\theoremstyle{plain}

\theoremstyle{definition}
\newtheorem*{question*}{Question}
\newtheorem*{openquestion*}{Open question}

\newenvironment{proofsketch}{\begin{proof}[Proof sketch.]}{\end{proof}}
\newenvironment{proofidea}{\begin{proof}[Proof idea.]}{\end{proof}}
\newenvironment{proofof}[1]{\begin{proof}[Proof (of #1).]}{\end{proof}}
\newenvironment{proofsketchof}[1]{\begin{proof}[Proof sketch (of #1).]}{\end{proof}}

\providecommand {\calA}      {{\mathcal A}\xspace}

\providecommand {\calC}      {{\mathcal C}\xspace}
\providecommand {\calD}      {{\mathcal D}\xspace}

\providecommand {\calF}      {{\mathcal F}\xspace}

\providecommand {\calI}      {{\mathcal I}\xspace}
\providecommand {\calL}      {{\mathcal L}\xspace}

\providecommand {\calR}      {{\mathcal R}\xspace}
\providecommand {\calS}      {{\mathcal S}\xspace}

\usepackage{tcolorbox}
\usepackage{xifthen}
\definecolor{applegreen}{rgb}{0.55, 0.71, 0.0} %
\definecolor{problemtitle}{HTML}{DDDDDD} %
\definecolor{problemtitletext}{HTML}{000000} %
\definecolor{problembody}{HTML}{EEEEEE} %
\definecolor{problemcomplexity}{HTML}{444444} %

\newcommand{\prb}[2][]{%
    \ifthenelse{\isempty{#1}}%
        {\hyperref[prb:#2]{\textsc{#2}}%
        }
        {\hyperref[prb:#2]{\textsc{#1}}%
        }}

\newenvironment{notopspacecolorbox} {
   \setlength{\topsep}{0pt}
   \setlength{\partopsep}{0pt}
   \vspace*{-1.5em}
   \tcolorbox
}{\endtcolorbox}

\newcommand\auxilTit{}
\newcommand\auxilLab{}
\newcommand{\prob}[4][]{
    \ifthenelse{\isempty{#1}}%
    {%
        \renewcommand\auxilTit{\textsc{\large\textcolor{black}{#2}}}
        \renewcommand\auxilLab{\label{prb:#2}}
    }%
    {%
        \renewcommand\auxilTit{\textsc{\large\textcolor{black}{#2 \hfill #1}}}
        \renewcommand\auxilLab{\label{prb:#2}}
    }%
    \par\medskip
    \begin{tcolorbox}[colback=problembody, colframe=problemtitle, right=1mm, left=1mm,
        title=\auxilTit]\auxilLab
        \begin{tabularx}{\textwidth}{l X}
            \textbf{Input:} & #3\\
            \textbf{Question:} & #4
        \end{tabularx}
    \end{tcolorbox}
    \par\medskip
}

\newcommand{\rprob}[4][]{
    \ifthenelse{\isempty{#1}}%
    {%
        \renewcommand\auxilTit{\textsc{\large\textcolor{black}{#2}}}
    }%
    {%
        \renewcommand\auxilTit{\textsc{\large\textcolor{black}{#2}}}
    }%
    \pms
    \begin{tcolorbox}[colback=problembody, colframe=problemtitle, right=1mm, left=1mm,
        title=\auxilTit]
        \begin{tabularx}{\textwidth}{l X}
            \textbf{Given:} & #3\\
            \textbf{Question:} & #4
        \end{tabularx}
    \end{tcolorbox}
    \pms
}

\newcommand{\manuallabelprob}[5][]{
    \ifthenelse{\isempty{#1}}%
    {%
        \renewcommand\auxilTit{\textsc{\large\textcolor{black}{#2}}}
    }%
    {%
        \renewcommand\auxilTit{\textsc{\large\textcolor{black}{#2}}}
    }%
    \pms
    \begin{tcolorbox}[colback=problembody, colframe=problemtitle, right=1mm, left=1mm,
        title=\auxilTit]\label{prb:#5}
        \begin{tabularx}{\textwidth}{l X}
            \textbf{Given:} & #3\\
            \textbf{Question:} & #4
        \end{tabularx}
    \end{tcolorbox}
    \pms
}

\newcommand{\cprob}[5][]{
    \ifthenelse{\isempty{#1}}%
    {%
        \renewcommand\auxilTit{\textsc{\large\textcolor{black}{#2}}}
        \renewcommand\auxilLab{\label{prb:#2}}
    }%
    {%
        \renewcommand\auxilTit{\textsc{\large\textcolor{black}{#2 \hfill #1}}}
        \renewcommand\auxilLab{\label{prb:#2}}
    }%
    \pms
    \begin{tcolorbox}[colback=problembody, colframe=problemtitle, right=1mm, left=1mm, bottom=0mm,
        title=\auxilTit]\auxilLab
        \begin{tabularx}{\textwidth}{l X}
            \textbf{Given:} & #3\\
            \textbf{Question:} & #4
        \end{tabularx}~
        \raggedleft
        \begin{tcolorbox}[colback=problemcomplexity,colframe=problemcomplexity,height=18pt,
            width=0.4\textwidth, right=0mm, left=0mm, bottom=0mm, top=0mm, box align=center,
            halign=center, valign=center]
            \textcolor{white}{#5}
        \end{tcolorbox}
    \end{tcolorbox}
    \pms
}

\newcommand{\cmanuallabelprob}[6][]{
    \ifthenelse{\isempty{#1}}%
    {%
        \renewcommand\auxilTit{\textsc{\large\textcolor{black}{#2}}}
        \renewcommand\auxilLab{\label{prb:#6}}
    }%
    {%
        \renewcommand\auxilTit{\textsc{\large\textcolor{black}{#2 \hfill #1}}}
        \renewcommand\auxilLab{\label{prb:#6}}
    }%
    \pms
    \begin{tcolorbox}[colback=problembody, colframe=problemtitle, right=1mm, left=1mm, bottom=0mm,
        title=\auxilTit]\auxilLab
        \begin{tabularx}{\textwidth}{l X}
            \textbf{Given:} & #3\\
            \textbf{Question:} & #4
        \end{tabularx}~
        \raggedleft
        \begin{tcolorbox}[colback=problemcomplexity,colframe=problemcomplexity,height=18pt,
            width=0.4\textwidth, right=0mm, left=0mm, bottom=0mm, top=0mm, box align=center,
            halign=center, valign=center]
            \textcolor{white}{#5}
        \end{tcolorbox}
    \end{tcolorbox}
    \pms
}

\newcommand{\Prob}[4][]{
    \ifthenelse{\isempty{#1}}%
    {%
        \renewcommand\auxilTit{\textsc{\large\textcolor{black}{#2}}}
    }%
    {%
        \renewcommand\auxilTit{\textsc{\large\textcolor{black}{#2 (#1)}}}
    }%
    \begin{tcolorbox}[colback=problembody, colframe=problemtitle, right=1mm, left=1mm,
        bottom=-5mm, title=\auxilTit]
        \begin{tabularx}{\textwidth}{l X}
            \textbf{Given:} & #3\\
            \textbf{Find:} & #4
        \end{tabularx}
    \end{tcolorbox}
}

\newcommand{\subscriptVertexType}{%
    \begin{tikzpicture}[]%
        \node[originnode, inner sep=1pt] (v1) at (0,0) {};
    \end{tikzpicture}
}

\newcommand{\inlineEdgeType}{%
    \begin{tikzpicture}[scale=0.55, inner sep=0pt]%
        \node[originnode] (v1) at (0,0) {};
        \node[originnode] (v2) at (0.8,0) {};
        \draw[originedge] (v1) edge (v2);
    \end{tikzpicture}
}

\newcommand{\subscriptEdgeType}{%
    \begin{tikzpicture}[]%
        \node[originnode, inner sep=1pt] (v1) at (0,0) {};
        \node[originnode, inner sep=1pt] (v2) at (0.3,0) {};
        \draw[originedge] (v1) edge (v2);
    \end{tikzpicture}
}

\newcommand{\subscriptNonEdgeType}{%
    \begin{tikzpicture}[]%
        \node[originnode, inner sep=1pt] (v1) at (0,0) {};
        \node[originnode, inner sep=1pt] (v2) at (0.2,0) {};
    \end{tikzpicture}
}

\newcommand{\subscriptDirEdgeType}{%
    \begin{tikzpicture}[]%
        \node[originnode, inner sep=1pt] (v1) at (0,0) {};
        \node[originnode, inner sep=1pt] (v2) at (0.35,0) {};
        \draw[originedge, shorten <=0.5pt, shorten >=0pt, -{stealth}] (v1) edge (v2);
    \end{tikzpicture}
}

\begin{document}

\maketitle

\begin{abstract}

We are interested in the following validation problem for computational reductions: for algorithmic problems $P$ and $P^\star$, is a given candidate reduction indeed a reduction from $P$ to $P^\star$? Unsurprisingly, this problem is undecidable even for very restricted classes of reductions. This leads to the question: Is there a natural, expressive class of reductions for which the validation problem can be attacked algorithmically? We answer this question positively by introducing an easy-to-use graphical specification mechanism for computational reductions, called cookbook reductions. We show that cookbook reductions are sufficiently expressive to cover many classical graph reductions and expressive enough so that SAT remains NP-complete (in the presence of a linear order). Surprisingly, the validation problem is decidable for natural and expressive subclasses of cookbook reductions.

\end{abstract}

\section{Introduction}\label{section:introduction}
	Computational reductions are one of the most powerful concepts in theoretical computer science. They are used, among others, to  establish undecidability in computability theory and hardness of algorithmic problems in computational complexity theory. In practical applications, reductions help to harness the power of modern SAT solvers for other problems.
	
        Teaching reductions in introductory courses is usually a difficult task. To teach reductions in introductory courses, instructors often design learning tasks for (i) understanding the computational problems involved, (ii) exploring existing reductions via examples, and (iii) designing reductions between computational problems. In particular, tasks for (iii) are challenging for many students. Although learning reductions is perceived as difficult by students, technological teaching support has so far only been provided for (i) and (ii), likely because these tasks are typically easy to illustrate and checking student solutions is algorithmically straightforward. 
	
	Providing teaching support for (iii) requires to address the foundational question: Is there a suitable language for specifying reductions that can express a variety of reductions, but is also algorithmically accessible? In particular, it should be possible to test whether a candidate for a reduction provided  by a student is indeed a valid reduction, preferably also providing a counterexample in case a submitted answer is incorrect.

	In this paper, we propose such a specification language for reductions and study variants of the following algorithmic problem, parameterized by a class $\calR$ of reductions and complexity classes $\calC$ and $\calC^*$:
	
	\vspace{2mm}
\algorithmicProblemDescription{\redgen[\calR][\calC][\calC^\star]}{Algorithmic problems $P \in \calC$, $P^\star \in \calC^\star$, and a reduction $\rho \in \calR$.}{Is $\rho$ a reduction from $P$ to $P^\star$?}
\vspace{2mm}

	More precisely, our contributions are twofold:
\begin{itemize}
 \item We propose a graphical and modular specification language for reductions, which we call \emph{cookbook reductions} (Section~\ref{section:specification-language}). Its design is inspired by ``building blocks'' such as local replacement of nodes, edges, \dots\cite{GareyJ1979} that are used in the context of many standard reductions. Cookbook reductions allow these building blocks to be combined in a simple, stepwise fashion. 
 We compare the expressive power of cookbook reductions with standard methods of specifying reductions. Specifically, we relate cookbook reductions to quantifier-free first-order interpretations (Section~\ref{subsec:cookbook-qf}) and observe that \SAT remains \NP-hard under cookbook reductions, assuming the presence of a linear order (Corollary~\ref{cor:nphard_cookbook}).
 \item We study variants of the decision problem \textsc{Reduction?}, %
 obtained by choosing different classes of reduction candidates and by either fixing the algorithmic problems $P, P^\star$ or by fixing complexity classes $\calC, \calC^\star$ and letting $P \in \calC, P^\star \in \calC^\star$ be part of the input (Section~\ref{section:algorithms}). Not surprisingly, \textsc{Reduction?} is undecidable for many restricted variants (Theorem~\ref{theorem:algorithmic-hardness-of-reductions}). %
 
 To our surprise, several interesting variants remain decidable: for example, \textsc{Reduction?} is decidable for an arbitrary fixed problem $P$ and fixed $P^\star$ expressible in monadic second-order logic\footnote{This logic extends first-order logic with quantification over sets and can express for example the \NP-complete problem \textsc{3-Colorability}.}, if reduction candidates are from the subclass of cookbook reductions that allows local replacements of edges by a gadget graph (Theorem~\ref{th:fSO_fMSO}). 
 Also, for some concrete choices of problems $P, P^\star$, we characterize valid reductions; the characterizations can be used to generate counterexamples for invalid candidates, which is particularly relevant in teaching contexts. 
 
\end{itemize}

\subparagraph*{Related work}

        Restricted specification languages have also been used in \cite{CrouchIM10, JordanK13} in the context of learning reductions algorithmically. Reductions that are similar in spirit to cookbook reductions due to their stepwise fashion are pp-constructions and gadget reductions in the realm of (finite) constraint satisfaction problems \cite{BOP2018, DO23, Bodirsky2021}. %

        \subparagraph*{Outline} We introduce cookbook reductions as a specification language for reductions in Section \ref{section:specification-language}. In Section \ref{section:expressive-power}, we study how the expressive power of the language compares to reductions definable in quantifier-free first-order logic. We then study the algorithmic problem of deciding whether a given candidate reduction correctly reduces a source to a target problem in Section \ref{section:algorithms}. We conclude 
        by discussing a preliminary implementation of the presented framework in the teaching support system \emph{Iltis} in Section \ref{section:summary}.

\section{Preliminaries}\label{section:preliminaries}
We assume familiarity with basic notions from finite model theory \cite{Libkin2004}.

A \emph{(purely relational) schema} $\schema = \{R_1, \ldots, R_m\}$ is a set of relation symbols $R_i$ with associated \emph{arities} $\arity(R_i)$. A (finite) $\schema$-structure $\struc = (\dom, R_1^\struc, \ldots, R_m^\struc)$ consists of a finite set~$\dom$, called the \emph{universe} or the \emph{domain} of $\struc$, and relations $R_i^\struc \subseteq \dom^{\arity(R_i)}$. If clear from the context, we sometimes omit the superscript $\calS$. We also refer to the domain of $ \struc $ as $ \univ{\struc} $. We write $\qdFO[k]$ for the set of all first-order formulas with quantifier depth at most $k$. The \emph{$\qdFO[k]$-type} of a $\schema$-structure $\struc$ is the set of all $\qdFO[k]$ formulas over schema $\schema$ that $\struc$ satisfies. Two structures $\struc_1, \struc_2$ are \emph{\FO-similar} up to quantifier depth $k$, written $\struc_1 \foeq{k} \struc_2$, if they have the same $\qdFO[k]$-type.

An \emph{isomorphism type} of $\schema$-structures is an equivalence class of the equivalence relation ``is isomorphic to''. We represent an isomorphism type by an arbitrarily fixed $\schema$-structure $\isotype$ with universe $\{1, \ldots, k\}$, for the appropriate number $k$, from that equivalence class. The \emph{arity} of an isomorphism type is the universe size of its representative. Often, we identify an isomorphism type with its representative $\isotype$. 
Given a structure $\struc$ and a subset $A$ of its universe, we write $\subtype_\struc(A)$ for the isomorphism type of $\restrict{\struc}{A}$, so, the isomorphism type of the substructure of $\struc$ that is induced by $A$. We write $\subtype(A)$ if $\struc$ is clear from the context and call $\subtype(A)$ the isomorphism type of $A$.

An \emph{embedding} $\pi$ of a structure $\struc$ into a structure $\struc^\star$ is an injective mapping from the domain of $\struc$ into the domain of $\struc^\star$ that is an isomorphism between $\struc$ and the substructure of $\struc^\star$ that is induced by the image of $\pi$. So, an embedding $\pi$ witnesses that $\struc^\star$ contains an isomorphic copy of $\struc$ as an induced substructure.

An (\emph{algorithmic}) \emph{problem} $P$ is an isomorphism-closed set of $\schema$-structures, for some schema~$\schema$. A \emph{reduction} $\rho$ from a problem $P$ over schema $\schema$ to a problem $P^\star$ over schema $\schema^\star$ is a mapping from $\schema$-structures to $\schema^\star$-structures such that $\struc \in P \Leftrightarrow \rho(\struc) \in P^\star$, for every $\schema$-structure $\struc$. A \emph{$d$-dimensional first-order interpretation} from $\schema$-structures to $\schema^\star$-structures is a tuple $\Psi = (\varphi_\dom(\tpl x), \varphi_\sim(\tpl x_1, \tpl x_2), (\varphi_R(\tpl x_1, \ldots, \tpl x_{\arity(R)}))_{R \in \sigma^\star})$ of first-order formulas over schema $\schema$, where each tuple $\tpl x = (x_1, \ldots, x_d), \tpl x_i = (x_{i,1}, \ldots, x_{i,d})$ consists of $d$ variables. For a given $\sigma$-structure $\struc$ with universe $\dom$, let $\hat{\Psi}(\struc)$ be the $\schema^\star$-structure with universe $\hat{\dom} = \{\tpl a \in \dom^d \mid \struc \models \varphi_\dom(\tpl a)\}$ and relations $R^{\hat{\Psi}(\struc)} = \{(\tpl a_1, \ldots, \tpl a_{\arity(R)}) \in \hat{\dom}^{\arity(R)} \mid \struc \models \varphi_{R}(\tpl a_1, \ldots, \tpl a_{\arity(R)})\}$ for each $R \in \schema^\star$.
We demand that for every $\schema$-structure $\struc$, the binary relation $\sim^{\hat{\Psi}(\struc)} = \{(\tpl a_1, \tpl a_2) \in \hat{\dom}^2 \mid \struc \models \varphi_{\sim}(\tpl a_1, \tpl a_2)\}$ is a congruence relation on $\hat{\Psi}(\struc)$, that is, an equivalence relation on the universe that is compatible with the relations of the structure.
For a given $\sigma$-structure~$\struc$, the interpretation $\Psi$ defines the $\schema^\star$-structure $\Psi(\struc)$ that is the quotient structure of $\hat{\Psi}(\struc)$ with respect to $\sim^{\hat{\Psi}(\struc)}$, that is, the structure that results from $\hat{\Psi}(\struc)$ by restricting the universe to only one element for every equivalence class of $\sim^{\hat{\Psi}(\struc)}$.

Most of our examples will be drawn from the algorithmic problems from Figure \ref{figure:algorithmic-problems}. We also consider variants of some of these problems where $k$ is a fixed parameter, e.g. \algorithmicProblem{$k$-Clique} asks, given a graph $G$, whether there is a $k$-clique in $G$.

For a natural number $n$, we sometimes write $[n]$ for the set $\{1, \ldots, n\}$.

\begin{figure}[t]
\newtcolorbox{problembox}{left=0pt,right=0pt,top=0pt,bottom=0pt,colback=white}

\begin{tcbraster}[raster columns=2,raster rows=3, raster row skip=1mm,  raster column skip=1mm, size=small,left=0pt,right=0pt,top=0pt,bottom=0pt,raster equal height=rows]

	\begin{problembox}
		\small
		\algorithmicProblemDescription[4.9cm]
			{\footnotesize Clique}
			{\footnotesize Undirected graph $G = (V, E)$, $k \in \N$}
			{\footnotesize Is there a set $ U \subseteq V $ of size $k$ with $(u,v) \in E$ for all $ u,v \in U $?}%
	\end{problembox}
	\begin{problembox}
			\small
		\algorithmicProblemDescription[4.9cm]
			{\footnotesize IndependentSet}
			{\footnotesize Undirected graph $G = (V, E)$, $k \in \N$}
			{\footnotesize Is there a set $ U \subseteq V $ of size $k$ with $(u,v) \notin E$ for all $ u,v \in U $?}%
	\end{problembox}

	\begin{problembox}
			\small
		\algorithmicProblemDescription[4.9cm]
			{\footnotesize VertexCover}
			{\footnotesize Undirected graph $G = (V, E)$, $k \in \N$}
			{\footnotesize Is there a set $ U \subseteq V $ of size  at most $k$ such that $u \in U$ or $v \in U$ for all  $(u,v) \in E$?}%
	\end{problembox}
	\begin{problembox}
			\small
		\algorithmicProblemDescription[4.9cm]
			{\footnotesize FeedbackVertexSet}
			{\footnotesize Undirected graph $G = (V, E)$, $k \in \N$}
			{\footnotesize Is there a set $ U \subseteq V $ of size at most $k$ such that removing $U$ from $G$ yields a cycle-free graph?}%
	\end{problembox}

	\begin{problembox}
			\small
		\algorithmicProblemDescription[4.9cm]
			{\footnotesize HamCycle$_\text{u}$}
			{\footnotesize Undirected graph $G = (V, E)$}
			{\footnotesize Is there an undirected cycle in $G$ that passes each node exactly once?}
	\end{problembox}
	\begin{problembox}
			\small
		\algorithmicProblemDescription[4.9cm]
			{\footnotesize HamCycle$_\text{d}$}
			{\footnotesize Directed graph $G = (V, E)$}
			{\footnotesize Is there a directed cycle in $G$ that passes each node exactly once?}
	\end{problembox}
\end{tcbraster}
\caption{Collection of algorithmic problems considered in the paper.\label{figure:algorithmic-problems}}
\end{figure}

\section{Cookbook reductions: A specification language for reductions}\label{section:specification-language}
\newcommand{\targetlabel}[2][]{%
    \ifthenelse{\isempty{#1}}%
        {\scriptsize\ensuremath{#2^\star}}%
        {\scriptsize\ensuremath{(#2^\star\!\!,\!#1)}}%
}
\newcommand{\targetu}[1][]{\targetlabel[#1]{u}}
\newcommand{\targetv}[1][]{\targetlabel[#1]{v}}
\newcommand{\tupletargetlabel}[2][]{%
    \ifthenelse{\isempty{#1}}%
        {\scriptsize\ensuremath{(\!\{#2\},\!1)}}%
        {\scriptsize\ensuremath{(#2,\!#1)}}%
}

When looking for a reduction, one approach by typical experts is to subsequently try building blocks that they have encountered in the context of other reductions before. For example, Garey and Johnson \cite[Section~3.2]{GareyJ1979} discuss common proof techniques like local replacements that occur in many standard reductions. 
An example is the standard reduction from the problem of finding a directed Hamiltonian cycle to finding an undirected Hamiltonian cycle that transforms a directed graph into an undirected graph by mapping each node \scalebox{0.7}{\hspace{-2mm}\begin{tikzpicture}[baseline=(1b)]
             \node[] (1a) at (-0.8,0.2){};
             \node[] (1b) at (-0.8,-0.2){};
             \node[] (3a) at (.8,0.2){};
             \node[] (3b) at (.8,0){};
             \node[] (3c) at (.8,-0.2){};

             \node[originnode, minimum size=6pt, label=below:{\scriptsize $v$}] (2) at (0.0,0){};
             \draw[hcdedge, thick] (1a) -- (2);
             \draw[hcdedge, thick] (1b) -- (2);
             \draw[hcdedge, thick] (2) -- (3a);
             \draw[hcdedge, thick] (2) -- (3b);
             \draw[hcdedge, thick] (2) -- (3c);
             \end{tikzpicture}} to a small gadget \scalebox{0.8}{\hspace{-2mm}\begin{tikzpicture}[baseline=(tmp)]
             \node[] (tmp) at (-0.5,-.150){};
             \node[] (1a) at (-0.7,0.2){};
             \node[] (1b) at (-0.7,-0.2){};
             \node[] (3a) at (1.7,0.2){};
             \node[] (3b) at (1.7,0){};
             \node[] (3c) at (1.7,-0.2){};
             \node[originnode, minimum size=6pt, label=below:{\scriptsize $v_\text{in}$}] (1) at (0,0){};
             \node[originnode, minimum size=6pt, label=below:{\scriptsize $v$}] (2) at (0.5,0){};
             \node[originnode, minimum size=6pt, label=below:{\scriptsize $v_\text{out}$}] (3) at (1,0){};
             \draw[originedge] (1) -- (2);
             \draw[originedge] (2) -- (3);
             \draw[originedge] (1a) -- (1);
             \draw[originedge] (1b) -- (1);
             \draw[originedge] (3) -- (3a);
             \draw[originedge] (3) -- (3b);
             \draw[originedge] (3) -- (3c);
             \end{tikzpicture}}~\hspace{-2.5mm}~. Constructing such node gadgets is one of the typical building blocks when designing reductions.

Our approach towards constructing a specification language for reductions is to (1) identify common building blocks used in computational reductions between graph problems, and to (2) abstract these building blocks into a more general specification language. The resulting language is reasonably broad and, due to its modular and graphical nature, easy to use.

\subsection{Building blocks and recipes}\label{section:building-blocks}
Many computational reductions can be crafted from a small set of common building blocks. For reductions between graph problems, some such building blocks are the following:
\begin{itemize}
 \item \emph{Edge gadgets} replace each edge $(u, v)$ of the source instance uniformly by a graph. For example, in the standard reduction from \VC to \FVS, every edge \begin{tikzpicture}[inner sep=0pt, node distance=5mm, baseline=-1mm]%
        \node[originnode, label={[below=2mm]:{\scriptsize $u$}}] (u){};
        \node[originnode, label={[below=2mm]:{\scriptsize $v$}}, right of = u] (v){};
        \draw[originedge] (u) edge (v);
    \end{tikzpicture} in the source instance is replaced by a triangle \begin{tikzpicture}[inner sep=0pt, node distance=5mm, baseline=0mm]%
        \node[originnode, label={[below=2mm]:{\scriptsize $u$}}] (u){};
        \node[originnode, above right of= u] (uv){};
        \node[originnode, label={[below=2mm]:{\scriptsize $v$}}, below right of= uv] (v){};
        \draw[originedge] (u) edge (v);
        \draw[originedge] (u) edge (uv);
        \draw[originedge] (uv) edge (v);
    \end{tikzpicture}.
    
	 \item \emph{Node gadgets} replace each node of the source instance uniformly by a graph and specify how these graphs are connected. For example, in the standard reduction from  \HCd to \HCu, every node \begin{tikzpicture}[inner sep=0pt, node distance=5mm, baseline=-1mm]%
        \node[originnode, label={[below=2mm]:{\scriptsize $v$}}] (u){};
    \end{tikzpicture} in the source instance is replaced by a path \begin{tikzpicture}[inner sep=0pt, node distance=5mm, baseline=-1mm]%
        \node[originnode, label={[below=2mm]:{\scriptsize $v_\text{in}$}}] (vin){};
        \node[originnode, label={[below=2mm]:{\scriptsize $v$}}, right of= vin] (v){};
        \node[originnode, label={[below=2mm]:{\scriptsize $v_\text{out}$}}, right of= v] (vout){};
        \draw[originedge] (vin) edge (v);
        \draw[originedge] (v) edge (vout);
    \end{tikzpicture}  and if there is an edge $(u,v)$ in the source instance, then the paths for $u$ and $v$ are connected via
    \begin{tikzpicture}[inner sep=0pt, node distance=5mm, baseline=-2mm]%
        \node[originnode, label={[above=0.25mm]:{\scriptsize $u_\text{in}$}}] (uin){};
        \node[originnode, label={[above=0.5mm]:{\scriptsize $u$}}, right of= uin] (u){};
        \node[originnode, label={[above=0.25mm]:{\scriptsize $u_\text{out}$}}, right of= u] (uout){};
        \draw[originedge] (uin) edge (u);
        \draw[originedge] (u) edge (uout);
        \node[originnode, label={[below=2mm]:{\scriptsize $v_\text{in}$}}, below of=uin] (vin){};
        \node[originnode, label={[below=2mm]:{\scriptsize $v$}}, right of= vin] (v){};
        \node[originnode, label={[below=2mm]:{\scriptsize $v_\text{out}$}}, right of= v] (vout){};
        \draw[originedge] (vin) edge (v);
        \draw[originedge] (v) edge (vout);
        \draw[originedge] (uout) edge (vin);
    \end{tikzpicture}.
	 \item \emph{Global gadgets} introduce a (global) graph and specify how each node of this graph is connected to the nodes of the source instance. For example, in the simple reduction from  \ThreeClique to \FourClique, a single node \begin{tikzpicture}[inner sep=0pt, node distance=5mm, baseline=-1mm]%
        \node[originnode, label={[below=2mm]:{\scriptsize $g$}}] (g){};
    \end{tikzpicture} is introduced as global graph and each node~$v$ of the source instance is connected to $g$ via an edge  \begin{tikzpicture}[inner sep=0pt, node distance=5mm, baseline=-1mm]%
        \node[originnode, label={[below=2mm]:{\scriptsize $g$}}] (g){};
        \node[originnode, label={[below=2mm]:{\scriptsize $v$}}, right of= g] (v){};
        \draw[originedge] (g) edge (v);
    \end{tikzpicture}.

\end{itemize}

These building blocks have in common that target instances of reductions are obtained\footnote{Contrary to the formulation above, a reduction does not alter a source instance to form the target instance, but creates a new structure.} 
from source instances by following simple, recipe-like steps of the form ``for every occurrence of a substructure $\isotype$ in the source instance, create a copy of the substructure $\isotype^\star$ in the target structure''. For example, the recipes for the above reductions are as follows:
\begin{itemize}
 \item Reducing $k$-\VC to $k$-\FVS: For every node $v$ in the source instance, create a node $v^\star$ in the target instance. For every edge $(u,v)$ in the source instance, create a node $ w^\star_{uv} $ and edges $(u^\star,v^\star), (v^\star, w^\star_{uv}), (w^\star_{uv}, u^\star)$ in the target instance.
 \item Reducing \HCd to \HCu: For every node $v$ in the source instance, create nodes $v^\star_\text{in}, v^\star, v^\star_\text{out}$ in the target instance and connect them as a path. For every directed edge $(u,v)$ in the source instance, create the undirected edge $(u^\star_\text{out}, v^\star_\text{in})$ in the target instance. 
 \item Reducing \ThreeClique to \FourClique: Create a node $g^\star$ in the target instance. For every node $v$ of the source instance, create a node $v^\star$ in the target instance and add the edge $(v^\star,g^\star)$. Copy all  edges $(u,v)$ of the source instance as edges $(u^\star,v^\star)$ to the target instance.
\end{itemize}

Other reductions can also be phrased in this form, for instance:
\begin{itemize}
 \item Reducing $k$-\Clique to $k$-\IS: First, for every node $v$ of the source instance, create a node $v^\star$ in the target instance. Then, for every pair $u,v$ of nodes that are not connected by an edge in the source instance, create an edge $(u^\star,v^\star)$ in the target instance.
\end{itemize}

Reductions specified this way capture building blocks such as the ones from \cite{GareyJ1979} and are usually easy to understand, often much more than their presentation as algorithms or as logical interpretations. 
Such reductions can also easily be specified graphically, see Figure~\ref{figure:example:informal-cookbook}%
, facilitating the implementation in educational support systems (see Section \ref{section:summary}).  

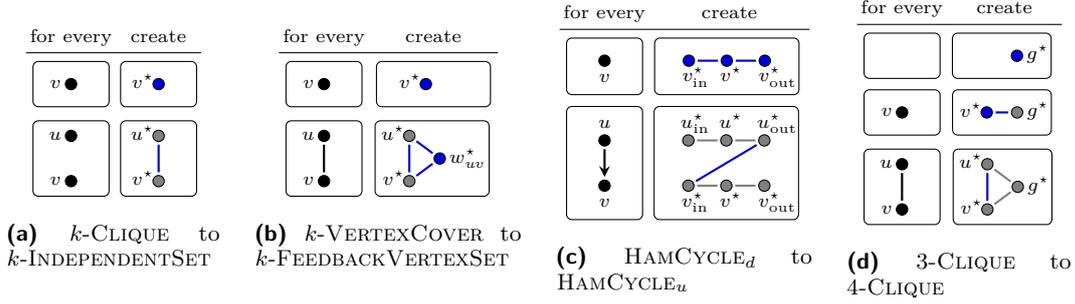
\begin{figure}[t]
\setlength{\tabcolsep}{0.15em}
\begin{subfigure}[t]{0.20\textwidth}
\centering
\begin{tabular}{c c}
\scriptsize{for every} &\scriptsize{create}  \\
\hline
\begin{tikzpicture}
\node at (0.5,.8) {}; %
 \draw[rounded corners = 2pt] (0,.2) rectangle (1,.8);
 \node[originnode, label={[shift={(0,0)}, label distance=-2pt]left:{\scriptsize $ v $}}] (v) at (0.5,0.5) {};
\end{tikzpicture} &
\begin{tikzpicture}
 \draw[rounded corners = 2pt] (0,.2) rectangle (1,.8);
 \node[targetnode-new, label={[shift={(0,0.04)}, label distance=-4pt]left:\targetv}] (v) at (0.5,0.5) {};
\end{tikzpicture}\\
\begin{tikzpicture}
 \draw[rounded corners = 2pt] (0,0) rectangle (1,1);
 \node[originnode, label={[shift={(0,0)}, label distance=-2pt]left:{\scriptsize $ u $}}] (v2) at (0.5,0.8) {};
 \node[originnode, label={[shift={(0,0)}, label distance=-2pt]left:{\scriptsize $ v $}}] (v1) at (0.5,0.2) {};
\end{tikzpicture} &
\begin{tikzpicture}
 \draw[rounded corners = 2pt] (0,0) rectangle (1,1);
 \node[targetnode-old, label={[shift={(0,0.04)}, label distance=-4pt]left:\targetu}] (v2) at (0.5,0.8) {};
 \node[targetnode-old, label={[shift={(0,0.04)}, label distance=-4pt]left:\targetv}] (v1) at (0.5,0.2) {};
 \draw[targetedge-new] (v1) edge (v2);
\end{tikzpicture}
\end{tabular}

 \caption{$k$-\Clique to \mbox{$k$-\IS}}
\end{subfigure} \quad
\begin{subfigure}[t]{0.25\textwidth}
\centering
\begin{tabular}{c c}
\scriptsize{for every} &\scriptsize{create}  \\
\hline
\begin{tikzpicture}
\node at (0.5,.8) {}; %
 \draw[rounded corners = 2pt] (0,.2) rectangle (1,.8);
 \node[originnode, label={[shift={(0,0)}, label distance=-2pt]left:{\scriptsize $ v $}}] (v) at (0.5,0.5) {};
\end{tikzpicture} &
\begin{tikzpicture}
 \draw[rounded corners = 2pt] (-0.15,.2) rectangle (1.35,.8);
 \node[targetnode-new, label={[shift={(0,0.04)}, label distance=-4pt]left:{\scriptsize $ v^\star $}}] (v) at (0.5,0.5) {};
\end{tikzpicture}\\
\begin{tikzpicture}
 \draw[rounded corners = 2pt] (0,0) rectangle (1,1);
 \node[originnode, label={[shift={(0,0)}, label distance=-2pt]left:{\scriptsize $ u $}}] (v2) at (0.5,0.8) {};
 \node[originnode, label={[shift={(0,0)}, label distance=-2pt]left:{\scriptsize $ v $}}] (v1) at (0.5,0.2) {};
 \draw[originedge] (v1) edge (v2);
\end{tikzpicture} &
\begin{tikzpicture}
 \draw[rounded corners = 2pt] (-0.15,0) rectangle (1.35,1);
 \node[targetnode-old, label={[shift={(0,0.04)}, label distance=-4pt]left:\targetu}] (v2) at (0.3,0.8) {};
 \node[targetnode-old, label={[shift={(0,0.04)}, label distance=-4pt]left:\targetv}] (v1) at (0.3,0.2) {};
 \node[targetnode-new, label={[shift={(0,0.0)}, label distance=-2pt]right:\targetlabel{w_{u\hspace{-0.7pt}v}}}] (v12) at (0.7,0.5) {};
 \draw[targetedge-new] (v1) edge (v2) edge (v12) 
 	(v2) edge (v12);
\end{tikzpicture}
\end{tabular}
 \caption{$k$-\VC to \mbox{$k$-\FVS}}
\end{subfigure} \quad
\begin{subfigure}[t]{0.24\textwidth}
\centering
\begin{tabular}{c c}
\scriptsize{for every} &\scriptsize{create}  \\
\hline
\begin{tikzpicture}
\node at (0.5,0.8) {}; %
 \draw[rounded corners = 2pt] (0,0.05) rectangle (1,.8);
 \node[originnode, label={[shift={(0,-0.045)}, label distance=-2pt]below:{\scriptsize $ v $}}] (v) at (0.5,0.5) {};
\end{tikzpicture} &
\begin{tikzpicture}
 \draw[rounded corners = 2pt] (-.45,0.05) rectangle (1.45,0.8);
 \node[targetnode-new, label={[shift={(0.09,0.04)}, label distance=-2pt]below:\targetlabel{v_{\mathrm{in}}}}] (v1) at (0,0.5) {};
 \node[targetnode-new, label={[shift={(0.06,0.04)}, label distance=-2pt]below:\targetv}] (v2) at (0.5,0.5) {};
 \node[targetnode-new, label={[shift={(0.17,0.04)}, label distance=-2pt]below:\targetlabel{v_{\mathrm{out}}}}] (v3) at (1,0.5) {};
 \draw[targetedge-new] (v1) edge (v2)
  (v2) edge (v3);
\end{tikzpicture}\\
\begin{tikzpicture}
 \draw[rounded corners = 2pt] (0,-0.25) rectangle (1,1.25);
 \node[originnode, label={[shift={(0,0.045)}, label distance=-2pt]above:{\scriptsize $ u $}}] (v2) at (0.5,0.8) {};
 \node[originnode, label={[shift={(0,-0.045)}, label distance=-2pt]below:{\scriptsize $ v $}}] (v1) at (0.5,0.2) {};
 \draw[originedge, -stealth] (v2) edge (v1);
\end{tikzpicture} &
\begin{tikzpicture}
 \draw[rounded corners = 2pt] (-.45,-0.25) rectangle (1.45,1.25);
 \node[targetnode-old, label={[shift={(0.09,0.04)}, label distance=-2pt]below:\targetlabel{v_\mathrm{in}}}] (v11) at (0,0.2) {};
 \node[targetnode-old, label={[shift={(0.06,0.04)}, label distance=-2pt]below:\targetv}] (v12) at (0.5,0.2) {};
 \node[targetnode-old, label={[shift={(0.17,0.04)}, label distance=-2pt]below:\targetlabel{v_\mathrm{out}}}] (v13) at (1,0.2) {};
 \node[targetnode-old, label={[shift={(0.09,0.04)}, label distance=-4pt]above:\targetlabel{u_\mathrm{in}}}] (v21) at (0,0.8) {};
 \node[targetnode-old, label={[shift={(0.06,0.04)}, label distance=-2pt]above:\targetu}] (v22) at (0.5,0.8) {};
 \node[targetnode-old, label={[shift={(0.17,0.04)}, label distance=-4pt]above:\targetlabel{u_\mathrm{out}}}] (v23) at (1,0.8) {};
 \draw[targetedge-old] (v11) edge (v12)
  (v12) edge (v13)
  (v21) edge (v22)
  (v22) edge (v23);
 \draw[targetedge-new] (v23) edge (v11);
\end{tikzpicture}
\end{tabular}
 \caption{\HCd to \HCu}
\end{subfigure} \quad
\begin{subfigure}[t]{0.21\textwidth}
\centering
\begin{tabular}{c c}
\scriptsize{for every} &\scriptsize{create}  \\
\hline
\begin{tikzpicture}
\node at (0.5,.8) {}; %
 \draw[rounded corners = 2pt] (0,.2) rectangle (1,.8);
\end{tikzpicture} &
\begin{tikzpicture}
 \draw[rounded corners = 2pt] (-0.15,.2) rectangle (1.15,.8);
 \node[globalnode-new, label={[shift={(0,0.02)}, label distance=-2pt]right:\targetlabel{g}}] (v) at (0.7,0.5) {};
\end{tikzpicture}\\
\begin{tikzpicture}
 \draw[rounded corners = 2pt] (0,.2) rectangle (1,.8);
 \node[originnode, label={[shift={(0,0)}, label distance=-2pt]left:{\scriptsize $ v $}}] (v) at (0.5,0.5) {};
\end{tikzpicture} &
\begin{tikzpicture}
 \draw[rounded corners = 2pt] (-0.15,.2) rectangle (1.15,.8);
 \node[targetnode-new, label={[shift={(0,0.04)}, label distance=-4pt]left:\targetv}] (v1) at (0.3,0.5) {};
 \node[globalnode-old, label={[shift={(0,0.02)}, label distance=-2pt]right:\targetlabel{g}}] (v2) at (0.7,0.5) {};
 \draw[targetedge-new] (v1) edge (v2);
\end{tikzpicture}\\
\begin{tikzpicture}
 \draw[rounded corners = 2pt] (0,0) rectangle (1,1);
 \node[originnode, label={[shift={(0,0)}, label distance=-2pt]left:{\scriptsize $ v $}}] (v1) at (0.5,0.2) {};
 \node[originnode, label={[shift={(0,0)}, label distance=-2pt]left:{\scriptsize $ u $}}] (v2) at (0.5,0.8) {};
 \draw[originedge] (v1) edge (v2);
\end{tikzpicture} &
\begin{tikzpicture}
 \draw[rounded corners = 2pt] (-0.15,0) rectangle (1.15,1);
 \node[targetnode-old, label={[shift={(0,0.04)}, label distance=-4pt]left:\targetv}] (v1) at (0.3,0.2) {};
 \node[targetnode-old, label={[shift={(0,0.04)}, label distance=-4pt]left:\targetu}] (v2) at (0.3,0.8) {};
 \node[globalnode-old, label={[shift={(0,0.02)}, label distance=-2pt]right:\targetlabel{g}}] (v12) at (0.7,0.5) {};
 \draw[targetedge-old] (v1) edge (v12)  
 	(v2) edge (v12);
 \draw[targetedge-new] (v1) edge (v2); 	
\end{tikzpicture}
\end{tabular}
 \caption{\ThreeClique to \mbox{\FourClique}}
\end{subfigure}
\caption{Graphical representations of four reductions. The reductions are applied stepwise, from the top-most step to the bottom-most step. Nodes and edges coloured blue are created in this step, grey nodes and edges were created in a previous step.}\label{figure:example:informal-cookbook} 
\end{figure}

\subsection{Cookbook reductions: Formalization}
We now formalize cookbook reductions as such recipe-style descriptions of computational reductions. In general, graphical representations as in Figure~\ref{figure:example:informal-cookbook} can be used to specify a cookbook reduction. In this section, we discuss the formal syntax and semantics.%

Intuitively, a reduction specified in our formalism builds, based on a source structure, the target structure in a sequence of stages, starting from an empty structure. At first, independent of the source structure, some global elements and tuples over these elements may be introduced to the target structure. Then, for every element of the source structure, a set of elements may be added, together with tuples that may also incorporate the elements that were introduced in the step before. The added elements and tuples depend on the (atomic) type of the respective element of the source structure.
In further stages, elements are analogously introduced for every set of two, three, \dots, elements of the source structure, depending on the type of these sets. %

Syntactically, a \emph{cookbook reduction} $\rho$ from $\schema$-structures to $\schema^\star$-structures is a finite set $\rho = \{(\isotype_1, \redS_1), \ldots, (\isotype_m, \redS_m)\}$ of pairs which we call \emph{instructions}. The structures $\isotype_i$ are $\schema$-structures with universe $\{1, \ldots, k_i\}$, for some natural number $k_i \geq 0$, that represent pairwise distinct isomorphism types of $\schema$-structures. The set $\{\isotype_1, \ldots, \isotype_m\}$ is the \emph{support} of $\rho$. The \emph{arity} of $\rho$ is the maximal arity of an isomorphism type in the support of $\rho$.
The structures $\redS_i$ are over the schema $\schema^\star$. 
For $(\isotype_i, \redS_i) \in \rho$, we also refer to $ \redS_i $ as $ \redS(\isotype_i) $. 
Each instruction $(\isotype, \redS)$, where $\isotype$ has the universe $[k] = \{1, \ldots, k\}$, satisfies the following properties:
\begin{enumerate}[(P1)]
    \item The universe $\univ{\redS}$ of $\redS$ consists of elements $(A, j)$, where $A \subseteq [k]$ and $j \geq 1$. If $(A, j) \in \univ{\redS}$ with $j > 1$, then also $(A, 1), \ldots, (A, j-1)$ are in $\univ{\redS}$.
    \item For any $(A, j) \in \univ{\redS} $ with $A \subsetneq [k]$, the isomorphism type $\isotype' = \subtype_\isotype(A)$ is in the support of $\rho$ and $(\{1, \ldots, |A|\}, j)$ is in $\univ{\redS(\isotype')}$.%
 \item For any tuple $((A_1, j_1), \ldots, (A_\ell, j_\ell))$ in any relation of $\redS$ with $\bigcup_{i \leq \ell} A_i  \subsetneq [k]$, the isomorphism type $\subtype_\isotype(\bigcup_{i \leq \ell} A_i)$ is in the support of $\rho$.
 \item For any $(\isotype', \redS') \in \rho$ and any $A \subsetneq [k]$ with $\subtype_\isotype(A) = \isotype'$, there is an isomorphism $\pi$ from $\isotype'$ to $\restrict{\isotype}{A}$ such that the injective mapping $\hat\pi$ with $\hat\pi((A',j')) = (\pi(A'),j')$, for all $(A',j')$ in $\univ{\redS'}$, is an embedding from $\redS'$ into $\redS$.
\end{enumerate}

A cookbook reduction has to satisfy a further, semantic property, which we state after defining the semantics.

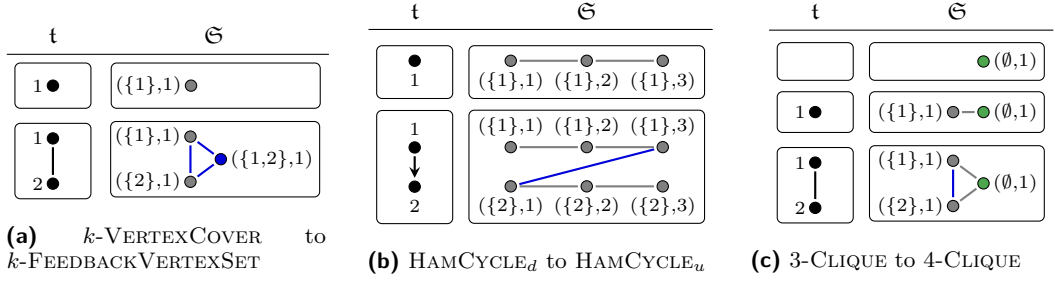
\begin{figure}[t]
\begin{subfigure}[t]{0.3\textwidth}
\centering
\begin{tabular}{c c}
    $\isotype$ & $\redS$  \\
    \hline
\begin{tikzpicture}
\node at (0.5,.8) {}; %
 \draw[rounded corners = 2pt] (0,.2) rectangle (1,.8);
 \node[originnode, label={[shift={(0,0)}, label distance=-2pt]left:{\scriptsize $ 1 $}}] (v) at (0.5,0.5) {};
\end{tikzpicture} &
\begin{tikzpicture}
 \draw[rounded corners = 2pt] (-0.75,.2) rectangle (2,.8);
 \node[cbnodetypenode, label={[shift={(0,0.0)}, label distance=-2pt]left:{\scriptsize $ (\{1\}\text{,}1) $}}] (v) at (0.3,0.5) {};
\end{tikzpicture}\\
\begin{tikzpicture}
 \draw[rounded corners = 2pt] (0,0) rectangle (1,1);
 \node[originnode, label={[shift={(0,0)}, label distance=-2pt]left:{\scriptsize $ 1 $}}] (v2) at (0.5,0.8) {};
 \node[originnode, label={[shift={(0,0)}, label distance=-2pt]left:{\scriptsize $ 2 $}}] (v1) at (0.5,0.2) {};
 \draw[originedge] (v1) edge (v2);
\end{tikzpicture} &
\begin{tikzpicture}
 \draw[rounded corners = 2pt] (-0.75,0) rectangle (2.0,1);
 \node[cbnodetypenode, label={[shift={(0,0.0)}, label distance=-2pt]left:{\scriptsize $ (\{1\}\text{,}1) $}}] (v2) at (0.3,0.8) {};
 \node[cbnodetypenode, label={[shift={(0,0.0)}, label distance=-2pt]left:{\scriptsize $ (\{2\}\text{,}1) $}}] (v1) at (0.3,0.2) {};
 \node[cbedgetypenode, label={[shift={(0,0.0)}, label distance=-2pt]right:{\scriptsize $ (\{1\text{,}2\}\text{,}1) $}}] (v12) at (0.7,0.5) {};
 \draw[cbedgetypeedge] (v1) edge (v2);
 \draw[cbedgetypeedge] (v1) edge (v12)
     (v12) edge (v2);
\end{tikzpicture}
\end{tabular}
 \caption{$k$-\VC to \mbox{$k$-\FVS}}
\end{subfigure}
\hfill
\begin{subfigure}[t]{0.32\textwidth}
\centering
\begin{tabular}{c c}
    $\isotype$ & $\redS$  \\
\hline
\begin{tikzpicture}
\node at (0.5,0.8) {}; %
 \draw[rounded corners = 2pt] (0,0.0) rectangle (1,.7);
 \node[originnode, label={[shift={(0,-0.045)}, label distance=-2pt]below:{\scriptsize $ 1 $}}] (v) at (0.5,0.5) {};
\end{tikzpicture} &
\begin{tikzpicture}
 \draw[rounded corners = 2pt] (-1.05,0.0) rectangle (2.05,0.7);
 \node[cbnodetypenode, label={[shift={(0,-0.02)}, label distance=-2pt]below:{\scriptsize $ (\{1\}\text{,}1) $}}] (v1) at (-0.5,0.5) {};
 \node[cbnodetypenode, label={[shift={(0,-0.02)}, label distance=-2pt]below:{\scriptsize $ (\{1\}\text{,}2) $}}] (v2) at (0.5,0.5) {};
 \node[cbnodetypenode, label={[shift={(0,-0.02)}, label distance=-2pt]below:{\scriptsize $ (\{1\}\text{,}3) $}}] (v3) at (1.5,0.5) {};
 \draw[cbnodetypeedge] (v1) edge (v2)
  (v2) edge (v3);
\end{tikzpicture}\\
\begin{tikzpicture}
 \draw[rounded corners = 2pt] (0,-0.25) rectangle (1,1.25);
 \node[originnode, label={[shift={(0,0.045)}, label distance=-2pt]above:{\scriptsize $ 1 $}}] (v2) at (0.5,0.76) {};
 \node[originnode, label={[shift={(0,-0.045)}, label distance=-2pt]below:{\scriptsize $ 2 $}}] (v1) at (0.5,0.24) {};
 \draw[originedge, -stealth] (v2) edge (v1);
\end{tikzpicture} &
\begin{tikzpicture}
 \draw[rounded corners = 2pt] (-1.05,-0.25) rectangle (2.05,1.25);
 \node[cbnodetypenode, label={[shift={(0,-0.02)}, label distance=-2pt]below:{\scriptsize $ (\{2\}\text{,}1) $}}] (v11) at (-0.5,0.24) {};
 \node[cbnodetypenode, label={[shift={(0,-0.02)}, label distance=-2pt]below:{\scriptsize $ (\{2\}\text{,}2) $}}] (v12) at (0.5,0.24) {};
 \node[cbnodetypenode, label={[shift={(0,-0.02)}, label distance=-2pt]below:{\scriptsize $ (\{2\}\text{,}3) $}}] (v13) at (1.5,0.24) {};
 \node[cbnodetypenode, label={[shift={(0,0.02)}, label distance=-2pt]above:{\scriptsize $ (\{1\}\text{,}1) $}}] (v21) at (-0.5,0.76) {};
 \node[cbnodetypenode, label={[shift={(0,0.02)}, label distance=-2pt]above:{\scriptsize $ (\{1\}\text{,}2) $}}] (v22) at (0.5,0.76) {};
 \node[cbnodetypenode, label={[shift={(0,0.02)}, label distance=-2pt]above:{\scriptsize $ (\{1\}\text{,}3) $}}] (v23) at (1.5,0.76) {};
 \draw[cbnodetypeedge] (v11) edge (v12)
  (v12) edge (v13)
  (v21) edge (v22)
  (v22) edge (v23);
 \draw[cbedgetypeedge] (v23) edge (v11);
\end{tikzpicture}
\end{tabular}
 \caption{\HCd to \HCu}
\end{subfigure}
\hfill
\begin{subfigure}[t]{0.3\textwidth}
    \centering
    \begin{tabular}{c c}
    $\isotype$ & $\redS$  \\
    \hline
    \begin{tikzpicture}
    \node at (0.5,.8) {}; %
     \draw[rounded corners = 2pt] (0,.25) rectangle (1,.75);
    \end{tikzpicture} &
    \begin{tikzpicture}
     \draw[rounded corners = 2pt] (-0.8,.25) rectangle (1.5,.75);
     \node[cbglobalnode, label={[shift={(0,0.0)}, label distance=-2pt]right:{\scriptsize $ (\emptyset\text{,}1) $}}] (v) at (0.7,0.5) {};
    \end{tikzpicture} \\
    \begin{tikzpicture}
     \draw[rounded corners = 2pt] (0,.23) rectangle (1,.77);
     \node[originnode, label={[shift={(0,0)}, label distance=-2pt]left:{\scriptsize $ 1 $}}] (v) at (0.5,0.5) {};
    \end{tikzpicture} &
    \begin{tikzpicture}
     \draw[rounded corners = 2pt] (-0.8,.23) rectangle (1.5,.77);
     \node[cbnodetypenode, label={[shift={(0,0.0)}, label distance=-2pt]left:{\scriptsize $ (\{1\}\text{,}1) $}}] (v1) at (0.3,0.5) {};
     \node[cbglobalnode, label={[shift={(0,0.0)}, label distance=-2pt]right:{\scriptsize $ (\emptyset\text{,}1) $}}] (v2) at (0.7,0.5) {};
     \draw[cbnodetypeedge] (v1) edge (v2);
    \end{tikzpicture} \\
    \begin{tikzpicture}
     \draw[rounded corners = 2pt] (0,0) rectangle (1,1);
     \node[originnode, label={[shift={(0,0)}, label distance=-2pt]left:{\scriptsize $ 2 $}}] (v1) at (0.5,0.2) {};
     \node[originnode, label={[shift={(0,0)}, label distance=-2pt]left:{\scriptsize $ 1 $}}] (v2) at (0.5,0.8) {};
     \draw[originedge] (v1) edge (v2);
    \end{tikzpicture} &
    \begin{tikzpicture}
     \draw[rounded corners = 2pt] (-0.8,0) rectangle (1.5,1);
     \node[cbnodetypenode, label={[shift={(0,0.0)}, label distance=-2pt]left:{\scriptsize $ (\{2\}\text{,}1) $}}] (v1) at (0.3,0.2) {};
     \node[cbnodetypenode, label={[shift={(0,0.0)}, label distance=-2pt]left:{\scriptsize $ (\{1\}\text{,}1) $}}] (v2) at (0.3,0.8) {};
     \node[cbglobalnode, label={[shift={(0,0.0)}, label distance=-2pt]right:{\scriptsize $ (\emptyset\text{,}1) $}}] (v12) at (0.7,0.5) {};
     \draw[cbnodetypeedge] (v1) edge (v12);
     \draw[cbnodetypeedge] (v2) edge (v12);
     \draw[cbedgetypeedge] (v1) edge (v2);
    \end{tikzpicture} 
    \end{tabular}
 \caption{\ThreeClique to \FourClique}
\end{subfigure}
\caption{Three reductions formalized as cookbook reductions. Nodes introduced for type $ \isotypeGlobal $ are coloured green, nodes and edges introduced for type $ \isotypeVertex $ are coloured grey, and nodes and edges introduced for types $ \isotypeEdge $ and $ \isotypeDirEdge $ are coloured blue. %
Compare to Figure~\ref{figure:example:informal-cookbook}(b), (c), and (d).}\label{figure:example:cookbook-reduction} 
\end{figure} 

See Figure~\ref{figure:example:cookbook-reduction} for examples of cookbook reductions. %

We give some more explanations for the conditions (P1)--(P4).
Intuitively, an instruction $(\isotype, \redS) \in \rho$ means that for every occurrence of the type $\isotype$ in the source structure, a copy of the structure $\redS$ is included in the target structure.
The conditions (P1) and (P2) are concerned with the universe $\univ{\redS}$ of $\redS$. 
If $\isotype$ is an isomorphism type of $k$ elements, the universe of $\redS$ partly consists of elements $([k], 1), \ldots, ([k], m)$, for some number $m$. These elements are added to the target structure for every occurrence of the type $\isotype$. We also call these $m$ elements \emph{fresh} and write $\fresh{\isotype} = m$ (and $\fresh{\isotype} = 0$ if no such element exists). 
The universe of $\redS$ also contains further elements of the form $(A, j)$ with $A \subsetneq [k]$. These represent elements that are added  for sets of elements with size $k' < k$ (in the intuitive explanation: in previous stages). If such an element $(A, j)$ occurs in the universe of $\redS$, there has to be a corresponding instruction to add this element, that is, the type $\isotype'$ of the set $A$ in $\isotype$ has to be in the support of $\rho$ and the element $([k'],j)$ has to be a fresh element in $\redS(\isotype')$.

The conditions (P3) and (P4) concern the relations of $\redS$. A tuple $((A_1, j_1), \ldots, (A_\ell, j_\ell))$ with $\bigcup_{i \leq \ell} A_i = [k]$ in a relation of $\redS$ says that this tuple is to be added to the target structure for every set of elements of type $\isotype$. No further conditions on these tuples are imposed by (P3) and (P4). 
If $A' \df \bigcup_{i \leq \ell} A_i$ is a proper subset of $[k]$, this tuple is added for the subset $A'$ of elements (intuitively: in a previous stage). Again, there needs to be another instruction that adds this tuple, that is, the isomorphism type $\isotype'$ of $A'$ needs to be in the support of $\rho$. 

If a subtype $\isotype'$ of $\isotype$ is in the support of $\rho$ then the corresponding instruction $(\isotype', \redS')$ needs to be respected: for every occurrence of $\isotype'$ in $\isotype$, a copy of the structure $\redS'$ needs to be present in $\redS$. Formally, if a set $A \subsetneq [k]$ with $|A| = k'$ has type $\isotype'$ in $\isotype$, as witnessed by some isomorphism $\pi$ from $\isotype'$ to $\restrict{\isotype}{A}$, the substructure of $\redS$ that is induced by the set $\{(A_i,j_i) \mid A_i \subseteq \pi([k']) \}$ is isomorphic to $\redS'$.

We now define the semantics of cookbook reductions. 
A cookbook reduction $\rho = \{(\isotype_1, \redS_1), \ldots, (\isotype_m, \redS_m)\}$ maps a $\schema$-structure $\struc$ to a set $\rho(\struc)$ of $\schema^\star$-structures, %
where $\schema$ is the schema of the isomorphism types $\isotype_i$ and $\schema^\star$ is the schema of the structures $\redS_i$. 
For some $\schema$-structure $\struc$, the $\schema^\star$-structure $\struc^\star$ is in $\rho(\struc)$ if the following conditions hold: %
\begin{enumerate}[(S1)]
    \item The universe $\univ{\struc^\star}$ of $\struc^\star$ consists of exactly those elements $(A, j)$ with $A \subseteq \univ{\struc}$ such that 
\begin{itemize}
 \item the isomorphism type $\isotype = \subtype_\struc(A)$ is in the support of $\rho$, and
 \item the structure $\redS$ with $(\isotype, \redS) \in \rho$ has the element $(\{1, \ldots, |A|\}, j)$ in its universe.
\end{itemize}
\item If a tuple $((A_1, j_1), \ldots, (A_\ell, j_\ell))$ is in some relation $R^{\struc^\star}$ of $\struc^\star$, for any $R \in \schema^\star$, then the isomorphism type $\subtype_\struc(\bigcup_{i \leq \ell} A_i)$ is in the support of $\rho$.
\item For any $(\isotype, \redS) \in \rho$ and any $A \subseteq \univ{\struc}$ with $\subtype_\struc(A) = \isotype$, there is an isomorphism $\pi$ from $\isotype$ to $\restrict{\struc}{A}$ such that the injective mapping $\hat\pi$ with $\hat\pi((A',j')) = (\pi(A'),j')$, for all $(A',j')$ in the universe of $\redS$, is an embedding from $\redS$ into $\struc^\star$.
\end{enumerate}

Intuitively, these conditions state that the elements (S1) and tuples (S3) of $\struc^\star$ can be obtained by transforming occurrences of an isomorphism type $\isotype$ in $\struc$ into $\redS$, for any $(\isotype, \redS) \in \rho$, and that no other tuples are present (S2). 

A cookbook reduction $\rho$ needs to satisfy the following semantic property\footnote{For global and node gadget reductions as introduced in Section~\ref{section:building-blocks}, this property is trivially satisfied, for edge gadget reductions it is satisfied if the gadget graph is symmetric. In general, the following syntactic restriction is necessary: For every $(\isotype, \redS) \in \rho$ and any automorphism $\pi$ of $\isotype$ there is an automorphism $\hat\pi$ of $\redS$ with $\hat\pi((A,j)) = (\pi(A),j')$, for any $(A,j)$ in the universe of $\redS$.}.
\begin{enumerate}[(P1)]
\setcounter{enumi}{4}
\item For every $\schema$-structure $\struc$, the set $\rho(\struc)$ is a non-empty set of isomorphic structures.
\end{enumerate}
Abusing notation, we usually write $\rho(\struc)$ to denote some arbitrary structure $\struc^\star \in \rho(\struc)$.

\section{The expressive power of cookbook reductions}\label{section:expressive-power}
In this section we study the expressive power of cookbook reductions. First, we explain how the  building blocks from Section \ref{section:specification-language} are captured by restricted cookbook reductions. Afterwards, we discuss the expressive power of general cookbook reductions and relate them to quantifier-free first-order interpretations. 

\subsection{From building blocks to cookbook reductions}

Cookbook reductions are a versatile reduction concept and as we have seen in the examples depicted in Figure~\ref{figure:example:informal-cookbook} and Figure~\ref{figure:example:cookbook-reduction}, many reductions have a small and easily understandable representation as cookbook reductions that have only few isomorphism types in their support. 

In fact, the building blocks for graph problems that we discussed as motivation for cookbook reductions can be recovered as restricted variants of cookbook reductions. %
For undirected graphs with only the binary edge relation $E$ and no self-loops, only four  isomorphism types of arity at most $2$ are relevant: the type $\isotypeGlobal$ of the graph with $0$ nodes, the type $\isotypeVertex$ of a single node, the type $\isotypeEdge$ of an undirected edge, and the type $\isotypeNonEdge$ of non-edges. 

We obtain the following characterization:
\begin{itemize}
    \item For a \emph{global gadget reduction}, the inserted global graph $\redS(\isotypeGlobal)$ is arbitrary. Nodes of the source instance are copied, so we fix $\fresh{\isotypeVertex} = 1$, but allow $\redS(\isotypeVertex)$ to arbitrarily select nodes from the global graph that are connected to every source node. Edges of the source are copied, so $\fresh{\isotypeEdge} =0$ and $\redS(\isotypeEdge)$ just adds the edge.%
 \item A \emph{node gadget reduction} replaces every node by some gadget, so $\redS(\isotypeVertex)$ is arbitrary. The reduction can define how these gadgets are connected in case there is an edge between the corresponding nodes in the source instance, resulting in $\fresh{\isotypeEdge} =0$ and $\redS(\isotypeEdge)$ being arbitrary apart from that.
\item An \emph{edge gadget reduction} replaces edges by some gadget. As every node from the source is copied to the target, $\redS(\isotypeVertex)$ is a single node. We allow any symmetric $\redS(\isotypeEdge)$.
\end{itemize}
Only the mentioned isomorphism types are in the support of the cookbook reduction. %

A similar characterization holds if the source graph is directed. %

Global, node or edge gadget reductions constitute expressive subclasses of cookbook reductions that are relatively easy to comprehend. More fragments can be defined by, e.g., setting an upper bound for $\fresh{\isotypeVertex}$ in a node gadget reduction, or selecting a different set of isomorphism types $\isotype$ for which $\redS(\isotype)$ needs to be provided.
This modularity of cookbook reductions helps finding decidable cases of the \textsc{Reduction?} problem. In a teaching context, instructors can select the degree of freedom students have.

\subsection{Relating cookbook reductions to quantifier-free interpretations}\label{subsec:cookbook-qf}
Quantifier-free first-order (FO) interpretations constitute a widely-used class of reductions with very low complexity, see, e.g., \cite{Immerman87}. They are still expressive enough to show hardness of problems: \SAT, the satisfiability problem for propositional formulas, is \NP-hard even under quantifier-free FO interpretations \cite{Dahlhaus83}.

In this section, we show that cookbook reductions can be expressed as quantifier-free FO interpretations. If we assume a linear order on the input structures, mildly restricted quantifier-free FO interpretations can be expressed as cookbook reductions. It follows that if input structures are linearly ordered, \SAT is \NP-hard under cookbook reductions.

We say that two reductions $\rho_1$ and $\rho_2$ are \emph{equivalent} for a source structure $\struc$ over the appropriate schema, if the target structures $\rho_1(\struc)$ and $\rho_2(\struc)$ are isomorphic.%

\begin{restatable}{theorem}{theoremCookbookIsQf}\label{th:cookbook_is_qf}
For every cookbook reduction $\rho$ there is a $d$-dimensional quantifier-free first-order interpretation $\Psi$, for some number $d$, such that $\rho$ and $\Psi$ are equivalent for every structure with at least $2$ elements.
\end{restatable}
\begin{proofidea}
Suppose that for a cookbook reduction $\rho = \{(\isotype_1, \redS_1), \ldots, (\isotype_m, \redS_m)\}$ the maximal arity of an isomorphism type $\isotype_i$ is $k$ %
and $\ell$ is the maximal size of the universe of a structure $\redS_i$. The interpretation $\Psi$ intuitively creates for each set of elements of type $\isotype_i$ a copy of the structure $\redS_i$, so, defines a universe of elements of the form $(A, i)$, where $|A| \leq k$ and $i \leq \ell$. Such elements can be encoded by tuples of length $d \df k+\ell+1$. Quantifier-free formulas can determine the isomorphism type of a set of elements and, by the properties of a cookbook reduction, whether a tuple $((A, i_1), \ldots, (A,i_r))$ exists in the interpreted structure only depends on the isomorphism type of $A$. Details can be found in the appendix.
\end{proofidea}

We call a first-order interpretation \emph{set-respecting} if, for the equivalence relation defined by the formula $\varphi_\sim(\tpl x_1, \tpl x_2)$, two tuples $\tpl a_1, \tpl a_2$ are only in the same equivalence class if $\tpl a_1$ and $\tpl a_2$ contain the same set of elements.

\begin{restatable}{theorem}{theoremQftoCookbook}\label{th:qf_to_cookbook}
For every set-respecting quantifier-free first-order interpretation $\Psi$ there is a cookbook reduction $\rho$ such that $\rho$ and $\Psi$ are equivalent for every structure with a linearly ordered universe.
\end{restatable}
\begin{proofidea}
Let $d$ be the dimension of $\Psi$. For every isomorphism type $\isotype$ of $k \leq d$ elements, the number $\ell$ of elements $([k],1), \ldots, ([k], \ell)$ in the universe of $\redS(\isotype)$, so, the number of elements added to the target structure because of a set of elements with isomorphism type $\isotype$, is equal to the number of equivalence classes of the congruence defined by $\varphi_\sim$ on the set of $d$-tuples that contain exactly the $k$ elements of $\isotype$ and satisfy the formula $\varphi_\dom$ of $\Psi$. We identify each of the $\ell$ elements with a particular $d$-tuple over the set $[k]$, which is possible as $[k]$ is linearly ordered. The structure $\redS(\isotype)$ is then defined as dictated by $\Psi$.
\end{proofidea}

As \SAT is \NP-hard under set-respecting quantifier-free FO interpretations \cite{Dahlhaus83}, we obtain:

\begin{corollary}\label{cor:nphard_cookbook}
Assuming that input structures are linearly ordered, \SAT is \NP-hard under cookbook reductions.
\end{corollary}

Note that in descriptive complexity theory one often studies relational input structures that are not linearly ordered (although Immerman usually assumes a linear order to be present \cite[Proviso 1.14]{Immerman87}). However, when considering Turing machines as models of computation in complexity theory, inputs are binary string encodings and therefore linearly ordered. 

\section{Towards automated correctness tests and feedback}
\label{section:algorithms}

We now turn to the problem of checking whether a given reduction candidate is a valid reduction between two computational problems $P$ and $P^\star$. 
In a first variation of this problem, a corresponding algorithm gets as input the reduction candidate $\rho \in \calR$ as well as the two problems $P \in \calC$ and $P^\star \in \calC^\star$, for a fixed class $\calR$ of reductions and fixed complexity classes $\calC$ and $\calC^\star$. Formally, this corresponds to solving the following algorithmic problem \redgen[\calR][\calC][\calC^\star], parameterized by $\calC$, $\calC^\star$, and $\calR$. Also fixing the problems $P$ and $P^\star$ yields the special case \redgen[\calR][P][P^\star].

\begin{minipage}[t]{0.47\textwidth}
\vspace{1mm}
\algorithmicProblemDescription[4.7cm]{\redgen[\calR][\calC][\calC^\star]}{Algorithmic problems $P \in \calC$, $P^\star \in \calC^\star$, and a reduction $\rho \in \calR$.}{Is $\rho$ a reduction from $P$ to $P^\star$?}
\vspace{3mm}
\end{minipage}
\begin{minipage}[t]{0.47\textwidth}
\vspace{1mm}
\algorithmicProblemDescription[4.7cm]{\redgen[\calR][P][P^\star]}{A reduction $\rho \in \calR$.}{Is $\rho$ a reduction from $P$ to $P^\star$?}
\vspace{3mm}
\end{minipage}

We are slightly vague here, as for the moment we leave open how algorithmic problems and reductions are represented. It will be clear how these are represented for all classes $\calC$, $\calC^\star$ and $\calR$ we will consider.
For standard classes of reductions, 
-- including reductions computable in polynomial time or logarithmic space, as well as first-order definable reductions -- 
already the second, more restricted problem is clearly undecidable for all non-trivial $P$ and $P^\star$. Already testing whether a quantifier-free interpretation or even an edge gadget reduction reduces from some problem $P$ to another problem $P^\star$ is undecidable, for simple $P$ and $P^\star$. As soon as $ P $ or $ P^\star $ are part of the input, the \textsc{Reduction?} problem is undecidable in most cases in which one of the classes $ \calC $ or $ \calC^\star $ is defined by an undecidable fragment of second-order logic, even for very simple classes of reductions. The proof of the following theorem is in Appendix \ref{section:algorithms:undecidability}.

\begin{restatable}{theorem}{theoremAlgorithmicHardnessOfReductions}\label{theorem:algorithmic-hardness-of-reductions}
 	\begin{enumerate}
	 \item 	\redgen[\calR][P][P^\star] is undecidable for the following parameters:
	\begin{enumerate}
		\item The class $\calR$ of first-order interpretations, $P = \emptyset$ and arbitrary  $P^\star$ (or vice versa, i.e. arbitrary $P$ and $P^\star = \emptyset$).
		\item The class $\calR$ of edge gadget reductions, $P = \emptyset$ and some graph problem $P^\star$ definable in first-order logic with arithmetic. 
	 \item The class $\calR$ of quantifier-free interpretations, $P = \emptyset$ and the graph problem $P^\star$ defined by the first-order formula $\formbis \df  \forall x \exists y E(x,y)$.

	\end{enumerate}
	\item \redgen[\calR][\calC][\calC^\star] is undecidable for the following parameters:
	\begin{enumerate}
		\item A class $\calR$ containing the identity mapping, a class $\calC$ containing the empty problem, and a class $\calC^\star$ defined by a fragment of second-order logic with undecidable finite satisfiability problem.
		\item A class $\calR$ containing the identity mapping, a class $\calC$ defined by a fragment of second-order logic with undecidable finite satisfiability problem, and a class $\calC^\star$  containing the empty problem.
	\end{enumerate}
	\end{enumerate}
\end{restatable}

In the rest of this section, we explore how to overcome the undecidability barriers. That is, we explore for which parameters one can obtain algorithms for solving \redgen[\calR][P][P^\star]  and \redgen[\calR][\calC][\calC^\star]. Our focus is on (restrictions of) cookbook reductions.

We start by exhibiting toy examples for algorithms for \algorithmicProblem{Reduction?$(P, P^\star, \calR)$}  for concrete algorithmic problems $P$ and $P^\star$ in Section \ref{section:explicit-problems}. 
For these examples, counterexamples can be provided if the input is not a correct reduction. A generalized view is taken in Section \ref{section:algorithm-templates}, where we exhibit algorithm templates for \algorithmicProblem{Reduction?$(P, P^\star, \calR)$} for algorithmic problems $P$ and $P^\star$ selected from classes of problems. Then, in Section \ref{section:formulas-as-inputs}, we consider algorithmic problems as part of the input by studying \algorithmicProblem{Reduction?$(\calC, \calC^\star, \calR)$}.

\subsection{Warm-up: Reductions between explicit algorithmic problems}\label{section:explicit-problems}

In this section we provide toy examples of how \algorithmicProblem{Reduction?$(P, P^\star, \calR)$} can be decided for very restricted classes $\calR$: (1) for reducing $k$-\Clique to $\ell$-\Clique via global gadgets, for $k < \ell$, (2) for reducing  $k$-\VC to $k$-\FVS via edge gadgets, and (3) for reducing \HCd to \HCu via restricted node gadgets. 
In all cases, the decision procedures are obtained by characterizing the class of correct gadgets.

While not deep, these characterizations and the algorithms resulting from them are a first step towards more general results. %

We start by characterizing those global gadgets that reduce $k$-\Clique to $\ell$-\Clique. For simplicity, we represent global gadget reductions $\rho$ by a global gadget $\gadget_\rho$ and a distinguished subset $A$ of its nodes. When applying $\rho$ to a  graph $G = (V, E)$, the gadget $\gadget_\rho$ is disjointly added to $G$ and edges $(u,v)$ are introduced for all $u \in A$ and all~$v \in V$.

\begin{restatable}{proposition}{theoremConcreteProblemskCliquetolClique}\label{theorem:kclique-to-lclique}
    Let $\rho$ be a global gadget reduction with global gadget $\gadget_\rho$ and a distinguished subset $A$ of its nodes. Let $k, \ell \in \N$ with $k < \ell$. Then the following are equivalent:
    \begin{enumerate}
     \item $ \rho $ is a reduction from $k$-{\Clique} to
    $\ell$-{\Clique}
    \item $\gadget_\rho$ and $A$ satisfy the following conditions:
    \begin{enumerate}
        \item $\gadget_\rho$ has no $\ell$-clique
        \item $\gadget_\rho$ has an $(\ell-k)$-clique contained in $A$
        \item $\gadget_\rho$ has no $(\ell-k+1)$-clique contained in $A$
    \end{enumerate}%
	\end{enumerate}
	Furthermore, if $\rho$ is not a reduction from $k$-{\Clique} to $\ell$-{\Clique}, then a counterexample can be computed efficiently.
\end{restatable}

We next characterize those edge gadgets that constitute a reduction from $k$-\VC to $k$-\FVS. We represent edge gadget reductions $\rho$ by an edge gadget $\gadget_\rho$ with two distinguished nodes $c$ and $d$. When applying $\rho$ to a  graph $G = (V, E)$, all edges $(u,v) \in E$ are replaced by disjoint copies of $\gadget_\rho$, where $u, v$ are identified with $c, d$, respectively. %

\begin{restatable}{proposition}{theoremConcreteProblemsVCtoFVS}\label{theorem:vc-to-fvs}
    Let $ \rho $ be an edge gadget reduction based on the edge gadget $ \gadget_\rho $ with distinguished nodes $c$ and $d$. Then the following are equivalent:
    \begin{enumerate}
     \item $ \rho $ is a reduction from $k$-{\VC} to
    $k$-{\FVS}
    \item $\gadget_\rho$ satisfies the following conditions:
    \begin{enumerate}
        \item\label{char:vc-fvs-fvs} $ \{ c \} $  and $ \{ d \} $ are feedback vertex
            sets of $ \gadget_\rho $
        \item\label{char:vc-fvs-cycle} $ \emptyset $ is not a feedback vertex set of $ \gadget_\rho $.
    \end{enumerate}%
  \end{enumerate}
	Furthermore, if $\rho$ is not a reduction from $k$-{\VC} to $k$-{\FVS}, then a counterexample can be computed efficiently.
\end{restatable}

Lastly, we characterize restricted node gadget reductions from the directed Hamiltonian cycle problem \HCd to the undirected variant \HCu. For simplicity, we represent node gadget reductions $\rho$ by node gadgets~$\gadget_\rho$. A node gadget  $\gadget_\rho$ consists of two copies of a \emph{node graph} $ \redS(\isotypeVertex) $ and a set of additional edges between these copies. As an example, the standard reduction from $\HCd$ to $\HCu$ is represented by the node gadget  \begin{tikzpicture}[baseline, scale=0.5]
                 \node[originnode-small] (c1) at (0,0.5) {};
                 \node[originnode-small] (c2) at (0.5,0.5) {};
                 \node[originnode-small] (c3) at (1,0.5) {};
                 \node[originnode-small] (d1) at (0,0) {};
                 \node[originnode-small] (d2) at (0.5,0) {};
                 \node[originnode-small] (d3) at (1,0) {};
                 \draw[originedge-small] (c1) edge (c2);
                 \draw[originedge-small] (c2) edge (c3);
                 \draw[originedge-small] (d1) edge (d2);
                 \draw[originedge-small] (d2) edge (d3);

                 \draw[originedge-small] (c3) edge (d1);
                 \end{tikzpicture} consisting of two copies of the node graph  \begin{tikzpicture}[baseline, scale=0.5]
                 \node[originnode-small] (c1) at (0,0.25) {};
                 \node[originnode-small] (c2) at (0.5,0.25) {};
                 \node[originnode-small] (c3) at (1,0.25) {};
                 \draw[originedge-small] (c1) edge (c2);
                 \draw[originedge-small] (c2) edge (c3);
                 \end{tikzpicture} with one additional edge between them (cf. Figures \ref{figure:example:informal-cookbook}(c) and \ref{figure:example:cookbook-reduction}(b)).
When applying $ \gadget_\rho $ to a graph $ G=(V,E) $, all nodes in $V$ are replaced by a copy of the node graph and two such copies for nodes $u, v$ are connected accordingly by the additional set of edges, if $(u,v) \in E$.

As a first step towards characterizing node gadget reductions between \HCd and \HCu, we characterize all correct node gadget reductions whose node graph has at most three nodes. %
\begin{proposition}\label{thm:propositions-hcd-hcu}
    Let $ \rho $ be a node gadget reduction with node gadget $\gadget_\rho$ whose node graph has at most three nodes. Then the following are equivalent: 
    \begin{enumerate}
        \item $ \rho $ is a reduction from {\HCd} to {\HCu}
        \item $\gadget_\rho$ is either of the following node gadgets (with the two copies of the node graphs depicted at top and bottom), up to symmetries:
        \setlength{\tabcolsep}{5mm}
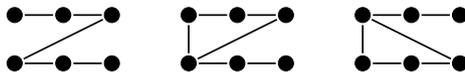
        \begin{figure}[H]
						\vspace{-6mm}
            \centering
						\scalebox{0.8}{
            \begin{tabular}{l l l}
                \begin{tikzpicture}[baseline, scale=0.8]

                 \node[originnode-bigger] (c1) at (0,1) {};
                 \node[originnode-bigger] (c2) at (1,1) {};
                 \node[originnode-bigger] (c3) at (2,1) {};
                 \node[originnode-bigger] (d1) at (0,0) {};
                 \node[originnode-bigger] (d2) at (1,0) {};
                 \node[originnode-bigger] (d3) at (2,0) {};
                 \draw[originedge] (c1) edge (c2);
                 \draw[originedge] (c2) edge (c3);
                 \draw[originedge] (d1) edge (d2);
                 \draw[originedge] (d2) edge (d3);

                 \draw[originedge] (c3) edge (d1);
                 \end{tikzpicture}
                 &
                \begin{tikzpicture}[baseline, scale=0.8]

                 \node[originnode-bigger] (c1) at (0,1) {};
                 \node[originnode-bigger] (c2) at (1,1) {};
                 \node[originnode-bigger] (c3) at (2,1) {};
                 \node[originnode-bigger] (d1) at (0,0) {};
                 \node[originnode-bigger] (d2) at (1,0) {};
                 \node[originnode-bigger] (d3) at (2,0) {};
                 \draw[originedge] (c1) edge (c2);
                 \draw[originedge] (c2) edge (c3);
                 \draw[originedge] (d1) edge (d2);
                 \draw[originedge] (d2) edge (d3);

                 \draw[originedge] (c1) edge (d1);
                 \draw[originedge] (c3) edge (d1);
                 \end{tikzpicture}
                 &
                \begin{tikzpicture}[baseline, scale=0.8]

                 \node[originnode-bigger] (c1) at (0,1) {};
                 \node[originnode-bigger] (c2) at (1,1) {};
                 \node[originnode-bigger] (c3) at (2,1) {};
                 \node[originnode-bigger] (d1) at (0,0) {};
                 \node[originnode-bigger] (d2) at (1,0) {};
                 \node[originnode-bigger] (d3) at (2,0) {};
                 \draw[originedge] (c1) edge (c2);
                 \draw[originedge] (c2) edge (c3);
                 \draw[originedge] (d1) edge (d2);
                 \draw[originedge] (d2) edge (d3);

                 \draw[originedge] (c1) edge (d1);
                 \draw[originedge] (c1) edge (d3);
                 \end{tikzpicture}
            \end{tabular}%
            }
            \vspace{-4mm}
        \end{figure}
    \end{enumerate}%
	Furthermore, if $ \rho $ is not a reduction from {\HCd} to {\HCu}, a counterexample can be computed efficiently.
\end{proposition}
The proofs are given in Appendix \ref{app:explicit-characterisations}.

\subsection{Decidable cases for classes of (fixed) algorithmic problems}
\label{section:algorithm-templates}

So far, we discussed that checking correctness of reductions is often undecidable, yet there are relevant problems $P$ and $P^\star$ for which reductions constructed from typical building blocks can be tested for correctness. 
In this section, we study the question whether there are classes $\calC$ and $\calC^\star$ of algorithmic problems as well as classes $\calR$ of reductions, such that after fixing $P \in \calC$ and $P^\star \in \calC^\star$ there is an algorithm that tests correctness of inputs $\rho \in \calR$.

 We first give an example that decidability results are possible for non-trivial classes of reductions and problems. Afterwards, we sketch how the technique employed in the proof can be generalized.
Recall that the arity of a cookbook reduction is the maximal arity of a type $\isotype$ in the support of the reduction.

\begin{theorem}\label{theorem:algorthmic-templates}
	\redgen[\calR][P][P^\star] is decidable for the class $\calR$ of cookbook reductions with arity bounded by some $r>0$, arbitrary $P$, and $P^\star$ definable in first-order logic.
\end{theorem}

The proof idea is to represent cookbook reductions $\rho$ by ``recipe structures'' \recipe such that $\red(\struct)$ can be constructed from the disjoint union $\struct\uplus\recipe$ of $\struct$ and $\recipe$ via an \FO-interpretation which depends on the arity and schema of \red, but is independent of \red itself. Then we prove that correctness of reductions in the setting of Theorem \ref{theorem:algorthmic-templates} only depends on the \FO-similarity type of their recipe.

Intuitively, the recipe of a cookbook reduction $\rho$ is the disjoint union of the structures $\redS(\isotype)$ for all relevant isomorphism types $\isotype$, where additional unary relations indicate the source structure and an additional binary relation identifies inherited elements (those $(A,j)$ where $A$ is a strict subset of the domain of $\isotype$) with their origin. 
Formally, fix two schemas $\sigma$ and $\sigma^\star$, an arity $r \in \N$, and define $\types_{\leq r}$ to be the finite set of all isomorphism types $\isotype$ over the schema $\sigma$  of arity at most $r$. The \emph{recipe} $\recipe$ of a cookbook reduction $\rho$ of arity at most $r$ from $\sigma$ to $\sigma^\star$ is a structure over the schema $\vocabbis \cup \{\inh\} \cup \{\colt \mid \isotype\in\typesr\}$, where $\inh$ is binary and all $\colt$ are unary. The restriction of \recipe to the schema $\vocabbis \cup \{\colt \mid \isotype\in\typesr\}$ is the disjoint union $\biguplus_{\isotype\in\typesr}\redS(\isotype)$, where we set $\redS(\isotype) = \rho(\isotype)$ if $\isotype$ is not in the support of $\rho$, and each $\colt$ is interpreted as the universe of $\redS(\isotype)$.  The relation $\inh$ ``identifies'' inherited elements and their original version: for every $\isotype,\isotype'\in\typesr$ such that $\isotype$ is the type of a strict subset of the elements of $\isotype'$, if $a'$ is an element of $\redS(\isotype')$ inherited from $\redS(\isotype)$'s element $a$, then $a'\inh a$ holds in \recipe.

\begin{figure}
  \centering
  \begin{tikzpicture}[scale=0.6]

    \node[originnode] (empty) at (.5,2.5) {};

    \node[originnode] (vertex-new) at (0,0) {};
    \node[originnode] (vertex-inh) at (1,0) {};
 
    \node[originnode] (edge-inh-v1) at (3,.5) {};
    \node[originnode] (edge-inh-v2) at (3,-.5) {};
    \node[originnode] (edge-inh-empty) at (3.85,0) {};

    \node[originnode] (nedge-inh-v1) at (3,3) {};
    \node[originnode] (nedge-inh-v2) at (3,2) {};
    \node[originnode] (nedge-inh-empty) at (3.85,2.5) {};

    \foreach \s/\t in {vertex-new/vertex-inh, edge-inh-v1/edge-inh-v2, nedge-inh-v1/nedge-inh-empty, nedge-inh-v2/nedge-inh-empty, edge-inh-v1/edge-inh-empty, edge-inh-v2/edge-inh-empty}
    \draw[originedge] (\s) edge (\t);

    \foreach \s/\t in {vertex-inh/empty,nedge-inh-empty/empty}
    \draw[dotted] (\s) edge (\t);
    
    \foreach \s/\t in {edge-inh-v1/vertex-new,nedge-inh-v1/vertex-new,nedge-inh-v2/vertex-new,edge-inh-empty/empty,vertex-new/edge-inh-v2}
    \draw[dotted] (\s) edge[bend right=15] (\t);

    \foreach \a/\b/\c/\d in {0/2/1/3,-.5/-.5/1.5/.5,2.5/-1/4.25/1,2.5/1.5/4.25/3.5}
    \draw[rounded corners = 2pt] (\a,\b) rectangle (\c,\d);

    \node at (-.5,2.5) {$\colempty$};
    \node at (-1,0) {$\colvertex$};
    \node at (5,0) {$\coledge$};
    \node at (5,2.5) {$\colnedge$};
  \end{tikzpicture}
  \caption{The recipe \recipe for the cookbook reduction of arity $2$ from \ThreeClique to \FourClique from Figure~\ref{figure:example:cookbook-reduction}. There are four unary relations for the types
  $\isotype_\emptyset$, $\isotypeVertex$, $\isotypeEdge$, and $\isotypeNonEdge$ of loopless undirected graphs. The dotted edges represent the binary inheritance relation $\inh$.
  }
  \label{fig:recipe_example}
 
\end{figure}
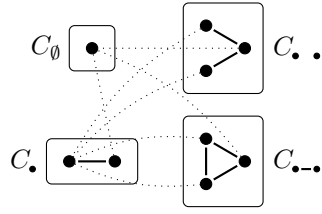
The structure \recipe representing the cookbook reduction \red from \ThreeClique to \FourClique given in Figure~\ref{figure:example:cookbook-reduction} can be found in Figure~\ref{fig:recipe_example}.

  There is an \FO-interpretation that applies a recipe $\recipe$ to a structure $\struct$ by interpreting $\struct\uplus\recipe$.
  \begin{restatable}{lemma}{lemmaRecipeFO}\label{lemma:recipes-FO}
    Fix $r>0$ and two schemas $\vocab, \vocabbis$. There is an \FO-interpretation $\interrecipe$ such that $\red(\struct)$ and $\interrecipe(\struct\uplus\recipe)$ are isomorphic, for every cookbook reduction $\red$ from $\vocab$ to $\vocabbis$ of arity at most $r$ and for every $\vocab$-structure $\struct$.%
  \end{restatable}

  As \FO-interpretations preserve \FO-similarity, there is a function $\frecipe:\N\to\N$ such that for every $k\in\N$, $\struct\foeq{\frecipe(k)}\structbis$ entails $\interrecipe(\struct)\foeq{k}\interrecipe(\structbis)$ (see, e.g., \cite[Section 3.2]{DBLP:books/daglib/0095988}).

We now prove Theorem \ref{theorem:algorthmic-templates}.
\begin{proof}[Proof of Theorem \ref{theorem:algorthmic-templates}]
  We show that whether a cookbook reduction $\rho$ is a reduction from $P$ to $P^\star$ solely depends on the $\qdFO[m]$-type of $\recipe$, for some large enough $m$ that depends only on $r$, $P$, and $P^\star$. As there are only finitely many such $\qdFO[m]$-types and because the type of $\recipe$ can be determined, the statement follows. 
  
  Let $k$ be the quantifier rank of a formula $\formbis\in\FO$ defining \problembis. If the recipes of two reductions $\red$ and $\redbis$ of arity at most $r$ are $\frecipe$-similar, then so are $\struct\uplus\recipe$ and $\struct\uplus\recipebis$ for all \vocab-structures $\struct$ (due to a simple Ehrenfeucht-Fra\"{\i}sse argument). But then $\interrecipe(\struct\uplus\recipe)$  and  $\interrecipe(\struct\uplus\recipebis)$ -- and therefore also $\red(\struct)$ and $\redbis(\struct)$ --, are $k$-similar.  In particular, the reductions $\rho$ and $\rho'$ behave in the same way for all $\sigma$-structures $\calA$, that is $\red(\struct)\models\formbis$ if and only if $\redbis(\struct)\models\formbis$. 
  
  We conclude that whether $\red(\struct)$ satisfies $\formbis$ only depends on the $\qdFO[\frecipe(k)]$-type of $\recipe$ for all $\struct$. Hence, the recipe of positive instances of $\redgen[\calR][P][P^\star]$ is a union of equivalence classes for $\foeq{\frecipe(k)}$. For a reduction $\rho$ it can now be evaluated whether its recipe satisfies the type of one of these equivalence classes.
\end{proof}

In the rest of this section, we explore how the technique used in the proof above can be generalized to logics beyond $\FO$. Our focus is on monadic-second order logic (\MSO), which extends $\FO$ by quantifiers for sets of elements. One of the key ingredients, that \FO-interpretations preserve $\FO$-similarity, does not translate to \MSO for interpretations of dimension greater than one (not even for quantifier-free interpretations). An example is provided in the appendix. 
Yet, decidability is retained for problems $P^\star \in \MSO$ if we restrict ourselves to \emph{edge gadget reductions} (on graphs), instead of general cookbook reductions. This generalizes Proposition~\ref{theorem:vc-to-fvs}.
\begin{restatable}{theorem}{theoremAlgorithmTemplatesMSO}\label{th:fSO_fMSO}
  \redgen[\calR][P][P^\star] is decidable for the class $\calR$ of edge gadget reductions, arbitrary $P$, and $P^\star$ definable in monadic second-order logic.
\end{restatable}
The proof exploits compositionality of \MSO and can be generalized to other subclasses of cookbook reductions. A discussion of such subclasses is postponed to the long version of this paper, see the appendix for an example.

\begin{proofsketch}
  An edge gadget reduction $\rho$ is specified as a graph $\gadget_\rho$, with two distinguished nodes. As in the proof of Theorem~\ref{theorem:algorthmic-templates}, the idea is to show that there is an integer $m$ such that whether $\rho$ is a reduction from $P$ to $P^\star$ only depends on the $\qdMSO[m]$-type of~$\gadget_\rho$. More precisely, for all gadget graphs $\gadget_\rho$ and $\gadget_{\rho'}$ with $\gadget_\rho\msoeq{m}\gadget_{\rho'}$, one proves that $\rho(G)\msoeq{k}\rho'(G)$  for all graphs $G$, where $k$ is the quantifier rank of an \MSO-sentence describing $P^\star$.

  For proving $\qdMSO[k]$-similarity of $\rho(G)$ and $\rho'(G)$, one can use \EF games for \MSO (see, e.g., \cite[Section 7.2]{Libkin2004}). The graphs $\rho(G)$ and $\rho'(G)$ are a composition of $G$ with the edge gadgets $\gadget_\rho$ and $\gadget_{\rho'}$, respectively. Duplicator has a winning strategy for the \MSO-game played on $(G, G)$ as well as for the \MSO-game played on $(\gadget_\rho, \gadget_{\rho'})$. Her strategy for the game  on $\rho(G)$ and $\rho'(G)$ is to combine these two winning strategies. For instance, if Spoiler moves on $\rho(G)$ and part of his move is on the edge gadget  inserted for an edge $(u, v)$ of $G$, then Duplicator's response for this part of the move is derived from her strategy for the game on $(\gadget_\rho, \gadget_{\rho'})$. The partial answers for individual edges are then combined. 

    For a formal proof, instead of making explicit the combinations of strategies, one can rely on Shelah's result~\cite{shelah1975monadic} on the compositionality of \MSO, see Appendix \ref{app:mso-fixed}. On top of abstracting the details of the games, it allows a straightforward extension of Theorem~\ref{th:fSO_fMSO} to a broader subclass of cookbook reductions.
\end{proofsketch}

For both \FO and \MSO, the proof uses that the respective classes of reductions can be finitely partitioned into similarity classes and that all reductions in one class are either correct or not correct. This provides a basis for characterizations akin to the ones in Section \ref{section:explicit-problems} for concrete, arbitrary problems $P$ and concrete $P^\star$ definable in \FO or \MSO.%

\newcommand{\indepgadg}{disjoint-gadgets\xspace}

\subsection{Algorithmic problems as input: decidable cases}
\label{section:formulas-as-inputs}

We now explore decidability when source and/or target problems are part of the input. 
We consider classes $\calC$ and $\calC^\star$ captured by logics $\calL$ and $\calL^\star$, respectively, and write, e.g., \redgen[\calR][\logic][\logicbis] for the algorithmic problem where we ask, given $\varphi \in \logic$, $\varphi^\star \in \logicbis$ and $\rho \in \calR$, whether $\rho$ is a reduction from the problem defined by $\varphi$ to the one defined by $\varphi^\star$.

One approach for obtaining decidability for the problem \redgen[\calR][\logic][\logicbis] is by restating it as a satisfiability question for a decidable logic. 
For a quantifier-free interpretation $\calI$ from $\sigma$-structures to $\sigma^\star$-structures, denote by $\calI^{-1}(\varphi^\star)$ the $\sigma$-formula obtained from a $\sigma^\star$-formula $\varphi^\star$ by replacing atoms in $\varphi^\star$ according to their definition in $\calI$. Whether a quantifier-free interpretation $\calI$ is a reduction from the algorithmic problem defined by $\varphi \in \calL$ to the one defined by $\varphi^\star \in \calL^\star$ is equivalent to whether $\calA\models\form$ if and only if $\inter(\calA)\models\formbis$, for all structures~$\calA$.
 This in turn is equivalent to checking whether $\form\leftrightarrow\interinv(\formbis)$ is a tautology. 

These observations yield, for instance, the following decidable variants, some involving the class \QF of quantifier-free first-order interpretations, a class that includes all cookbook reductions, see Theorem~\ref{th:cookbook_is_qf}. The proof is in the appendix. %
\begin{restatable}{theorem}{theoremProblemsAsInputsDecidability}\label{theorem:problems-as-inputs:decidability}
\begin{enumerate}
 \item $\redqf[\EFO][\EFO]$ is decidable.
 \item \redqf[\problem][\EFO] is decidable for every fixed algorithmic problem \problem.
	\item \redgen[\calR][\EFO][\problembis] is decidable for every fixed algorithmic problem \problembis definable in $\MSO$ and the class $\calR$ of edge gadget reductions.
\end{enumerate}
\end{restatable}

\section{Summary and discussion}\label{section:summary}

We studied variants of the algorithmic problem \textsc{Reduction?} which asks whether a given mapping is a computational reduction between two algorithmic problems. In addition to studying this problem for standard classes of reductions, we also proposed a graphical and compositional language for computational reductions, called cookbook reductions, and compared their expressive power to quantifier-free first-order interpretations. While \textsc{Reduction?} is undecidable in many restricted settings, we identified multiple decidable cases involving (restricted) cookbook reductions and quantifier-free first-order interpretations. Due to its graphical and compositional nature, cookbook reductions are well-suited to be used in educational support systems for learning tasks tackling the design of computational reductions.

A prototype\footnote{See \url{https://iltis.cs.tu-dortmund.de/computational-reductions}} of our formal framework has been integrated into the educational support system \emph{Iltis} \cite{SchmellenkampVZ24}. Recently it has been used in introductory courses \emph{Theoretical Computer Science} with $> 300$ students at Ruhr University Bochum and TU Dortmund in  workflows covering (i) understanding computational problems, (ii) exploring reductions via examples, and (iii) designing reductions.%

\bibliography{bibliography}

\newpage

\appendix

\section{Appendix for Section \ref{section:expressive-power}: The power of cookbook reductions}
\theoremCookbookIsQf*

For proving the theorem, instead of directly translating cookbook reductions to quantifier-free first-order interpretations, we use a variant of the latter reductions as an intermediate step.
A \emph{$d$-dimensional $\ell$-copying quantifier-free first-order interpretation} from $\schema$-structures to $\schema^\star$-structures is, similarly to $d$-dimensional quantifier-free first-order interpretations, a tuple $\Psi = (\varphi_\dom(\tpl x), \varphi_\sim(\tpl x_1, \tpl x_2), (\varphi_R(\tpl x_1, \ldots, \tpl x_{\arity(R)}))_{R \in \schema^\star})$ of quantifier-free first-order formulas, but each tuple $\tpl x = (x_1, \ldots, x_d, j), \tpl x_i = (x_{i,1}, \ldots, x_{i,d}, j_{i})$ consists of $d+1$ variables. All formulas are over the schema $\schema \cup \{1, \ldots, \ell\}$, where $1, \ldots, \ell$ are constant symbols that do not appear in~$\schema$.
Given a $\schema$-structure $\struc$ with universe $\dom$, the $\schema^\star$-structure $\Psi(\struc)$ has as universe tuples from $\dom^d \times \{1, \ldots, \ell\}$. The remaining semantics are analogous to the semantics of quantifier-free first-order interpretation as presented in Section~\ref{section:preliminaries}.

The following proposition implies Theorem~\ref{th:cookbook_is_qf}.
\begin{proposition}\label{prop:cookbook_is_qf}
\begin{enumerate}[(a)]
 \item For every cookbook reduction $\rho$ there is a $d$-dimensional $\ell$-copying quantifier-free first-order interpretation $\Psi$, for some numbers $d$ and $\ell$, such that $\rho$ and $\Psi$ are equivalent for every structure with at least $2$ elements.
 \item For any numbers $d, \ell$ and every $d$-dimensional $\ell$-copying quantifier-free first-order interpretation $\Psi_1$ there is a $(d+\ell)$-dimensional quantifier-free first-order interpretation $\Psi_2$ such that $\Psi_1$ and $\Psi_2$ are equivalent for every structure with at least $2$ elements.
\end{enumerate}
\end{proposition}

We start with proving the first part of the proposition.

\begin{proofof}{Proposition~\ref{prop:cookbook_is_qf}(a)}
Let a cookbook reduction $\rho = \{(\isotype_1, \redS_1), \ldots, (\isotype_n, \redS_n)\}$ be given. Let $k$ be the maximal arity of an isomorphism type $\isotype_i$ and let $\ell$ be the maximal size of the universe of a structure $\redS_i$. We show that there is a $(k+1)$-dimensional $\ell$-copying quantifier-free first-order interpretation $\Psi$ that is equivalent to $\rho$ for every structure with at least $2$ elements.

In the following, we denote by $\AllTypes = \{\isotype_1, \ldots, \isotype_n\}$ the support of $\rho$. For any natural number $k$, we denote by $\AllTypes_k$ the subset of $\AllTypes$ that consists of isomorphism types of structures with $k$ elements.

The cookbook reduction $\rho$ defines a structure based on sets of input elements and their isomorphism types, where the sets may have different sizes. A quantifier-free interpretation defines a structure based on the type of tuples of elements with fixed length. So, we first explain how sets of size at most $k$ can be encoded by tuples of length exactly $k+1$, provided there are at least two different elements.

Let $A$ be a set of $m \leq k$ elements and let $a_1, \ldots, a_m$ be an arbitrary enumeration of the elements in $A$. We encode the set $A$ by any $(k+1)$-tuple $(a_1, \ldots, a_m, a_0, \ldots, a_0)$ such that $a_0$ is an element different from $a_m$. So, a tuple $(a_1, \ldots, a_{k+1})$ of length $k+1$, where for some $m$ it holds that (1) $a_m \neq a_{m+1}$, (2) $a_{m+1} = \cdots = a_{k+1}$, and (3) all elements $a_1, \ldots, a_m$ are distinct, encodes the set $\{a_1, \ldots, a_m\}$. For example, $(a,a,\ldots,a)$ encodes the empty set and $(a,b,a,a, \ldots, a)$ encodes the set $\{a,b\}$. 

The formula \[\psi_m(x_1, \ldots, x_{k+1}) \df x_m \neq x_{m+1} \wedge \bigwedge_{m+1 \leq i \leq k} x_i = x_{i+1} \wedge \bigwedge_{i_1<i_2 \leq m} x_{i_1} \neq x_{i_2}\] expresses that $(x_1, \ldots, x_{k+1})$ encodes the set $\{x_1, \ldots, x_{m}\}$ of $m$ elements, for $1 \leq m \leq k$. 
The formula \[\psi_0(x_1, \ldots, x_{k+1}) \df \bigwedge_{i \leq k} x_i = x_{i+1}\] expresses that $(x_1, \ldots, x_{k+1})$ encodes the empty set.

The interpretation $\Psi$ intuitively works as follows, given a source structure $\struc$. For every set $A = \{a_1, \ldots, a_k\}$ of elements in $\struc$ that has some isomorphism type $\isotype$ that is in the support of $\rho$, we want to introduce a copy of $\redS(\isotype)$ to the interpreted structure. Actually, as multiple tuples represent the set $A$, we introduce more copies: one copy for every automorphism of $\isotype$ (and every element $a_0$ that is used to fill unused positions of a $(k+1)$-tuple). 
All tuples that represent the same universe element, either because they are from different copies of some $\redS(\isotype)$ or because they represent ``inherited'' elements that are introduced for some subtype $\isotype'$ and are repeated in $\redS(\isotype)$, then have to be identified using the formula $\varphi_\sim$.

With this intuition in mind, we sketch the formulas of $\Psi$.
We start with the formula $\varphi_\dom$ that defines the set of all tuples that represent an element of the universe of the defined structure. For any set $A = \{a_1, \ldots, a_k\}$ of elements that has some isomorphism type $\isotype \in \AllTypes$ in the given source structure $\struc$, intuitively, $\varphi_\dom$ ``creates'' the elements $(A, 1), \ldots, (A, \ell_\isotype)$, where $\ell_\isotype$ is the size of the universe of $\redS(\isotype)$. %

Remember that every isomorphism type is represented by a structure $\isotype$ with universe $[p]$, for some natural number $p$. The formula $\varphi_\dom$ selects a tuple $(a_1, \ldots, a_{k+1}, j)$, where $(a_1, \ldots, a_{k+1})$ encodes the set $\{a_1, \ldots, a_m\}$ for some $m \leq k$ and $j \leq \ell$ is a number, to represent a universe element if the substructure of $\struc$ induced by $\{a_1, \ldots, a_m\}$ is isomorphic to an isomorphism type $\isotype$ with $\univcard{\redS(\isotype)} \geq j$ via the isomorphism that maps $a_i$ to $i$, for any $i \in [m]$. In that case we say that the tuple $(a_1, \ldots, a_{k+1})$ \emph{encodes the type} of $\isotype$.

\begin{align*}
 \varphi_\dom(x_1, \ldots, x_{k+1}, j) & = 
   \bigvee_{0 \leq m \leq k} \; \bigvee_{\isotype \in \AllTypes_m} 
   \Big[ \psi_m(x_1, \ldots, x_{k+1}) \wedge \varphi_\isotype(x_1,\ldots, x_m) \wedge j \leq \univcard{\redS(\isotype)}   \Big]
\end{align*}
Here, $\varphi_\isotype(x_1,\ldots, x_m)$ is a formula that describes that $(x_1, \ldots, x_m)$ is isomorphic to $\isotype$ via the isomorphism that maps $x_i$ to $i$, for all $i \leq m$.

Now we discuss the formula $\varphi_\sim$ that is used to identify tuples that represent the same element of the universe. 

First we deal with tuples that represent elements from different copies of the same structure $\redS(\isotype)$.
Suppose that the tuples $(a_1, \ldots, a_{k+1})$ and $(b_1, \ldots, b_{k+1})$ encode the same set $A$ of $m$ elements and both encode the type of some $\isotype$. These properties imply that the function $\pi \colon [m] \to [m]$ with $\pi(i) = i'$ if $a_i = b_{i'}$ is an automorphism of $\isotype$, that is, an embedding of $\isotype$ into $\isotype$.
As to the properties of a cookbook reduction, there is a corresponding embedding $\hat\pi$ of $\redS(\isotype)$ into $\redS(\isotype)$, so, an automorphism of $\redS(\isotype)$.  
Any tuples $(a_1, \ldots, a_{k+1}, j_1)$ and $(b_1, \ldots, b_{k+1}, j_2)$ have to be identified if the $j_1$-th node of $\redS(\isotype)$ is mapped to the $j_2$-th node by $\hat\pi$, according to some arbitrary ordering of the elements of $\redS(\isotype)$.
This case is handled by the following formula.

\begin{align*}
 & \varphi^1_\sim(x_1, \ldots, x_{k+1}, j_1, y_1, \ldots, y_{k+1}, j_2) \df \\ %
  & \quad  \bigvee_{0 \leq m \leq k} \; \bigvee_{\isotype \in \AllTypes_m} \; \bigvee_{\pi \in \Aut(\isotype)}\; \bigvee_{\substack{i_1, i_2 \\ \hat\pi \text{ maps node } i_1 \text{ of } \redS(\isotype) \text{ to node } i_2}}
   \Big[ \psi_m(x_1, \ldots, x_{k+1}) \wedge \psi_m(y_1, \ldots, y_{k+1}) \wedge  \\
 & \qquad \varphi_\isotype(x_1,\ldots, x_m) \wedge \bigwedge_{i \leq m} x_i = y_{\pi(i)} \wedge j_1 = i_1 \wedge j_2 = i_2 \Big] 
\end{align*}

A similar formula $\varphi^2_\sim(x_1, \ldots, x_{k+1}, j_1, y_1, \ldots, y_{k+1}, j_2)$ is used to identify tuples that represent an element that is introduced for some isomorphism type $\isotype'$ and then ``inherited'' in the structure $\redS(\isotype)$ of a type $\isotype$ that includes $\isotype'$ as an induced substructure.
Let $(a_1, \ldots, a_{k+1}, j_1)$ and $(b_1, \ldots, b_{k+1}, j_2)$ be two tuples that represent sets $A$ and $B$ of isomorphism type $\isotype_A$ and $\isotype_B$, respectively, such that the intersection $C = A \cap B$ has some isomorphism type $\isotype_C$ and all types are in the support of $\rho$. 
Suppose the $j_1$-th node of $\redS(\isotype_A)$ is $(I_A,i_1)$ and the isomorphism $\pi$ from $A$ into $\isotype_A$ that maps $a_i$ to $i$, for all $i \leq |A|$, maps the set $\{p \mid a_p \in C \}$ to $I_A$. That is, the $j_1$-th node of the copy of $\redS(\isotype_A)$ represented by $(a_1, \ldots, a_{k+1})$ is inherited from a copy of $\redS(\isotype_C)$ for the set $C$ of elements. Suppose analogously that also the $j_2$-th node of the copy of $\redS(\isotype_B)$ represented by $(b_1, \ldots, b_{k+1})$ is inherited from a copy of $\redS(\isotype_C)$ for the set $C$.
Then, we consider the lexicographically smallest isomorphisms $\pi_A$ and $\pi_B$ from the sets $\{p \mid a_p \in C \}$ and $\{p \mid b_p \in C \}$ into $\isotype_C$. These mappings induce an automorphism $\pi_C$ of $\isotype_C$:  $\pi_C$ maps a number $i \leq \univcard{\isotype_C}$ to $i'$ if the element $a_p$ with $\pi_A(p) = i$ is equal to the element $b_{p'}$ with $\pi_B(p') = i'$.

The formula $\varphi^2_\sim$ then identifies $(a_1, \ldots, a_{k+1}, j_1)$ and $(b_1, \ldots, b_{k+1}, j_2)$ if the $j_1$ node of $\redS(\isotype_A)$ is the $r_1$-th node of the embedding of $\redS(\isotype_C)$ in $\redS(\isotype_A)$, the $j_2$ node of $\redS(\isotype_B)$ is the $r_2$-th node of the embedding of $\redS(\isotype_C)$ in $\redS(\isotype_B)$, and the automorphism $\hat\pi_C$ of $\red(\isotype_C)$ that corresponds to the automorphism $\pi_C$ maps the $r_1$-th node to the $r_2$-th node.

This can be expressed by a quantifier-free formula, as the structures $\isotype_i$ and $\redS(\isotype_i)$ are fixed and all mentioned mappings can be enumerated, but we omit explicitly constructing the formula.

It remains to discuss the formula $\varphi_R$ that defines a $\sigma^\star$-relation $R$ in the interpreted structure. 
For ease of presentation, we assume that $R$ is a binary relation symbol; the approach can be generalised to arbitrary arities. %

For tuples $(a_1, \ldots, a_{k+1}, j_1)$ and $(b_1, \ldots, b_{k+1}, j_2)$ that represent elements from the same copy of a structure $\redS(\isotype)$ we can look up in this structure whether these elements are connected by an $R$-edge. The represented elements are from the same copy if $(a_1, \ldots, a_{k+1})$ and $(b_1, \ldots, b_{k+1})$ represent the same set of some size $m$ and if $a_i = b_i$ for all $i \leq m$.
Also tuples $(a'_1, \ldots, a'_{k+1}, j'_1)$ and $(b'_1, \ldots, b'_{k+1}, j'_2)$ that are not from the same copy of a structure $\redS(\isotype)$ are in the relation defined by $\varphi_R$; this is the case if there are tuples $(a_1, \ldots, a_{k+1}, j_1)$ and $(b_1, \ldots, b_{k+1}, j_2)$ that satisfy the conditions above such that $(a_1, \ldots, a_{k+1}, j_1)$ and $(a'_1, \ldots, a'_{k+1}, j'_1)$ as well as $(b_1, \ldots, b_{k+1}, j_2)$ and $(b'_1, \ldots, b'_{k+1}, j'_2)$ are identified by $\varphi_\sim$. Note that by the properties of cookbook reductions, such tuples $(a_1, \ldots, a_{k+1}, j_1)$ and $(b_1, \ldots, b_{k+1}, j_2)$ can be constructed using only the elements that appear in $(a'_1, \ldots, a'_{k+1}, j'_1)$ and $(b'_1, \ldots, b'_{k+1}, j'_2)$, respectively, so this property can be expressed by a quantifier-free formula.

\begin{align*}
 \varphi_R & (x_1, \ldots, x_{k+1}, j_1, y_1, \ldots, y_{k+1}, j_2) = \bigvee_{f,g \colon [k+1] \to [k+1]} \: \bigvee_{j'_1, j'_2 \leq \ell} \: \bigvee_{0 \leq m \leq k} \: \bigvee_{\isotype \in \AllTypes_m} \\
 & \varphi_\sim(x_1, \ldots, x_{k+1}, j_1,x_{f(1)}, \ldots, x_{f(k+1)}, j'_1) \wedge \varphi_\sim(y_1, \ldots, y_{k+1}, j_2,y_{g(1)}, \ldots, y_{g(k+1)}, j'_2) \wedge \\
 & \psi_{m}(x_{f(1)}, \ldots, x_{f(k+1)}) \wedge \psi_{m}(y_{g(1)}, \ldots, y_{g(k+1)}) \wedge \bigwedge_{i \leq m} x_{f(i)} = y_{g(i)} \wedge \varphi_{\isotype}(x_{f(i)}, \ldots, x_{f(m)}) \wedge\\
 & \bigvee_{\substack{i_1, i_2 \\ \text{$u$ is $i_1$-th node, $v$ is $i_2$ node in $\redS(\isotype)$} \\ (u,v) \in R^{\redS(\isotype) }}} j'_1 = i_1 \wedge j'_2 = i_2
\end{align*}

\end{proofof}

We give a proof sketch for the second part of Proposition~\ref{prop:cookbook_is_qf}.

\begin{proofsketchof}{Proposition~\ref{prop:cookbook_is_qf}(b)}
We need to explain how one can avoid using the additional constants $\{1, \ldots, \ell\}$. Without loss of generality we assume that $\ell \geq 3$. Then a constant $i$ can be encoded by any $\ell$-tuple $(a_1, \ldots, a_\ell)$ with $a_i = b_1$ and $a_j = b_2$ for all $j \neq i$,  using two different elements $b_1, b_2$. A quantifier-free formula can check whether an $\ell$-tuple encodes some constant by checking that the tuple consists of exactly two elements and one element is used exactly once. The position of that element gives the encoded constant. 
The formula $\varphi_\sim$ is used to identify all different encodings of the same constant.  
\end{proofsketchof}
 
\section{Appendix for Section \ref{section:algorithms}: Towards automated correctness tests and feedback}

\subsection{Appendix for the introduction of Section \ref{section:algorithms}: Proofs of undecidability} \label{section:algorithms:undecidability}

\theoremAlgorithmicHardnessOfReductions*

Recall that the problems definable in first-order logic with arithmetic are exactly the problems computable by uniform $\AC^0$ circuits, and thus in particular contained in $\LOGSPACE$ and $\PTIME$. %

\begin{proof}[Proof sketch]
  Part (1a) follows immediately from the undecidability of the finite satisfiability problem of first-order logic.

  We now sketch the proof of part (1b). As a first step, we show this part for an algorithmic problem $\problembis$ in uniform $\TC^0$, afterwards we lift this to $\FO$ with arithmetic. Recall that uniform $\TC^0$ is the class of problems computable by uniform circuits of constant-depth and polynomial size  with unbounded fan-in $\wedge$-, $\vee$- and majority-gates (see, e.g., \cite{Vollmer1999}). 
  
  For the first step, we reduce the Post Correspondence Problem \PCP to $\redgen[\calR][\emptyset][\problembis]$ where $\calR$ is the class of (directed) edge gadget reductions and $\problembis \in \TC^0$ is specified below. As a reminder, an instance of \PCP is a sequence of pairs $(u_1,v_1),\cdots,(u_n,v_n)$ of non-empty words over the alphabet $\{0,1\}$. It is a yes-instance iff there exists a non-empty sequence of indexes $i_1,\cdots,i_k$ such that $u_{i_1}\cdots u_{i_k}=v_{i_1}\cdots v_{i_k}$.
  
  We reduce an instance $I\df(u_1,v_1),\cdots,(u_n,v_n)$ of \PCP to a gadget graph $\gadget_I$ (corresponding to an edge gadget reduction $\rho_I$) which encodes the $n$ pairs of words. The gadget graph $\gadget_I$ consists of an edge from $c^{\gadget_I}$ to $d^{\gadget_I}$ (so that $\rho_I(G)$ contains $G$ as an isolated subgraph, for every $G$), and of the disjoint union of the gadgets $\gadgetb_i$, illustrated in Figure~\ref{fig:gadget_PCP}, for each $1\leq i\leq n$. Each $\gadgetb_i$ consists of
  \begin{itemize}
  \item a unary encoding of the integer $i$ (as the number of in-neighbors of a central node),
  \item an encoding of $u_i$ as a path, where each node (i) is connected to the central node of $\gadgetb_i$, (ii) is marked with its position (encoded in unary as the number of out-neighbors), and (iii) is marked with some small gadget (not represented in the figure) describing whether the corresponding letter in $u_i$ is $0$ or $1$,
  \item a similar path for $v_i$, but with a double-sided edge from its node to the central node, in order to differentiate it from the path representing $u_i$.
  \end{itemize}

  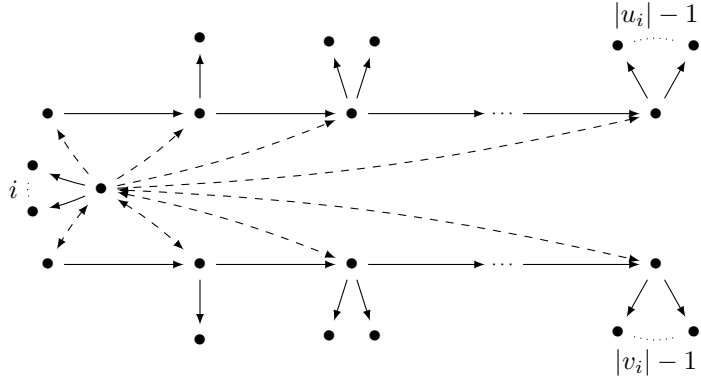
\begin{figure}
    \centering
    \begin{tikzpicture}[scale=1]

      \node at (.7,1) (o) {$\bullet$};

      \node at (-.2,.7) (o1) {$\bullet$};
      \node at (-.2,1.3) (ol) {$\bullet$};

      \draw[->,>=latex] (o) to[bend left=5] (o1);
      \draw[->,>=latex] (o) to[bend left=5] (ol);
      \draw[dotted] (o1) to[bend left=15] node[left] {$i$} (ol);

      \foreach \a in {0,1,2,4} { 
        \node at ({{2*\a}},2) (u\a) {$\bullet$};
        \node at ({{2*\a}},0) (v\a) {$\bullet$};
      }

      \draw[->,>=latex,dashed] (o) to[bend left=5] (u0);
      \draw[<->,>=latex,dashed] (o) to[bend right=5] (v0);
      
      \foreach \a in {1,2,4} { 
        \draw[->,>=latex,dashed] (o) to[bend right=5] (u\a);
        \draw[<->,>=latex,dashed] (o) to[bend left=5] (v\a);
      }
      
      \node[scale=.7] at ({{2*3}},2) (u3) {$\cdots$};
      \node[scale=.7] at ({{2*3}},0) (v3) {$\cdots$};

      \foreach \a/\b in {0/1,1/2,2/3,3/4} {
         \draw[->,>=latex] (u\a) -- (u\b);
         \draw[->,>=latex] (v\a) -- (v\b);
      }

      \node at (2,3) (u11) {$\bullet$};
      \node at (2,-1) (v11) {$\bullet$};

      \node at (3.7,2.95) (u21) {$\bullet$};
      \node at (4.3,2.95) (u22) {$\bullet$};
      \node at (3.7,-.95) (v21) {$\bullet$};
      \node at (4.3,-.95) (v22) {$\bullet$};

      \node at (7.5,2.9) (ul1) {$\bullet$};
      \node at (8.5,2.9) (ull) {$\bullet$};
      \node at (7.5,-.9) (vl1) {$\bullet$};
      \node at (8.5,-.9) (vll) {$\bullet$};

      \foreach \a/\b in {1/11,2/21,2/22,4/l1,4/ll} {
        \draw[->,>=latex] (u\a) -- (u\b);
        \draw[->,>=latex] (v\a) -- (v\b);
      }

      \draw[dotted] (ul1) to[bend left=15] node[above] {$|u_i|-1$} (ull);
      \draw[dotted] (vl1) to[bend right=15] node[below] {$|v_i|-1$} (vll);

    \end{tikzpicture}
        \caption{The gadget $\mathfrak h_i$, encoding the pair $(u_i,v_i)$. It consists of a central node with $i$ isolated neighbors, and of two paths representing $u_i$ and $v_i$. The node representing each letter of $u_i$ and $v_i$ has as many isolated neighbors as its position in the word, and is marked with a small gadget (not represented here) which encodes whether the letter is a $0$ or $1$. The final gadget $\mathfrak g_I$ is the disjoint union of all these $\mathfrak h_i$, for $1\leq i\leq n$, and of two nodes which correspond to constants $c$ and $d$, with an edge from $c$ to $d$.}
    \label{fig:gadget_PCP}                                                       
  \end{figure}

  Let us now describe the problem $P^\star\in\TC^0$ such that for every graph $G$:
  
  \begin{itemize}
   \item[] $\rho_I(G)\in P^\star$ iff $G$ encodes a sequence $i_1,\cdots,i_k$ such that $u \df u_{i_1}\cdots u_{i_k}=v_{i_1}\cdots v_{i_k} \df v$
  \end{itemize}

  Such an encoding, represented in Figure~\ref{fig:witness_PCP}, is composed of
  \begin{itemize}
  \item a path of length $k$, where node $l$ has $i_l$ isolated neighbors,
  \item a \emph{ladder} consisting of two paths (one which will correspond to $u$, and the other to $v$) with $|u|=|v|$ nodes, where each node is marked with its position in the path (encoded in unary by the number of neighbors),
  \item for each $1\leq l\leq k$, one edge from the node representing index $i_l$ to the node in the $u$-path at position $\sum_{j<l}|u_{i_j}|$, and another edge to the node in the $v$-path at position $\sum_{j<l}|v_{i_j}|$.
  \end{itemize}

  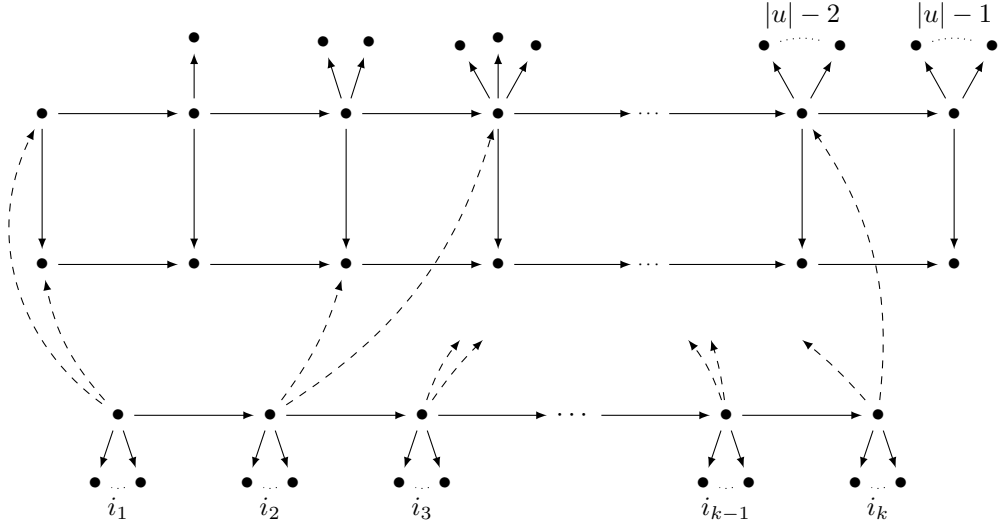
\begin{figure}
    \centering
    \begin{tikzpicture}[scale=1]

      \foreach \a in {0,2,4,6,10,12} { 
        \node at (\a,2) (t\a) {$\bullet$};
        \node at (\a,0) (b\a) {$\bullet$};
        \draw[->,>=latex] (t\a) -- (b\a);
      }

      \node[scale=.7] at (8,2) (t8) {$\cdots$};
      \node[scale=.7] at (8,0) (b8) {$\cdots$};
      
      \foreach \a/\b in {0/2,2/4,4/6,6/8,8/10,10/12} {
        \draw[->,>=latex] (t\a) -- (t\b);
        \draw[->,>=latex] (b\a) -- (b\b);
      }

      \node at (2,3) (t21) {$\bullet$};

      \node at (3.7,2.95) (t41) {$\bullet$};
      \node at (4.3,2.95) (t42) {$\bullet$};

      \node at (5.5,2.9) (t61) {$\bullet$};
      \node at (6,3) (t62) {$\bullet$};
      \node at (6.5,2.9) (t63) {$\bullet$};

      \node at (9.5,2.9) (tbl1) {$\bullet$};
      \node at (10.5,2.9) (tbll) {$\bullet$};

      \node at (11.5,2.9) (tl1) {$\bullet$};
      \node at (12.5,2.9) (tll) {$\bullet$};

      \foreach \a/\b in {2/21,4/41,4/42,6/61,6/62,6/63,10/bl1,10/bll,12/l1,12/ll} {
        \draw[->,>=latex] (t\a) -- (t\b);
      }

      \draw[dotted] (tbl1) to[bend left=15] node[above] {$|u|-2$} (tbll);
      \draw[dotted] (tl1) to[bend left=15] node[above] {$|u|-1$} (tll);
      
      \foreach \a in {1,2,3,5,6} { 
        \node at ({{2*\a-1}},-2) (i\a) {$\bullet$};
        \node at ({{2*\a-1.3}},-2.9) (il\a) {$\bullet$};
        \node at ({{2*\a-.7}},-2.9) (ir\a) {$\bullet$};
        \draw[->,>=latex] (i\a) -- (il\a);
        \draw[->,>=latex] (i\a) -- (ir\a);
      }

      \draw[dotted] (ir1) to[bend left=15] node[below]{$i_1$} (il1);
      \draw[dotted] (ir2) to[bend left=15] node[below]{$i_2$} (il2);
      \draw[dotted] (ir3) to[bend left=15] node[below]{$i_3$} (il3);
      \draw[dotted] (ir5) to[bend left=15] node[below]{$i_{k-1}$} (il5);
      \draw[dotted] (ir6) to[bend left=15] node[below]{$i_k$} (il6);
      
      \node at ({{3+4}},-2) (i4) {$\cdots$};
      
      \foreach \a/\b in {1/2,2/3,3/4,4/5,5/6} { 
        \draw[->,>=latex] (i\a) -- (i\b);
      }

      \draw[->,>=latex,dashed] (i1) to[bend left=40] (t0);
      \draw[->,>=latex,dashed] (i1) to[bend left=15] (b0);

      \draw[->,>=latex,dashed] (i2) to[bend right=20] (t6);
      \draw[->,>=latex,dashed] (i2) to[bend right=10] (b4);

      \draw[->,>=latex,dashed] (i3) to[bend left=10] (5.5,-1);
      \draw[->,>=latex,dashed] (i3) to[bend left=5] (5.8,-1);
      
      \draw[->,>=latex,dashed] (i5) to[bend right=10] (8.5,-1);
      \draw[->,>=latex,dashed] (i5) to[bend right=5] (8.8,-1);

      \draw[->,>=latex,dashed] (i6) to[bend right=20] (t10);
      \draw[->,>=latex,dashed] (i6) to[bend right=5] (10,-1);

    \end{tikzpicture}
    \caption{The encoding of a sequence of indexes $i_1,\cdots,i_k$ such that $u\df u_{i_1}\cdots u_{i_k}=v_{i_1}\cdots v_{i_k}\df v$. The bottom path encodes the sequence itself, while the top part (resp. bottom part) of the ladder indicates where each subword $u_{i_l}$ (resp. $v_{i_l}$) starts in $u$ (resp. $v$). In this example, we have $|u_{i_1}|=3$, $|v_{i_1}|=2$ and $|u_{i_k}|=2$.}
    \label{fig:witness_PCP}                                                       
  \end{figure}

  In $\TC^0$, one can check, given $\rho_I(G)$, that the sequence encoded by $G$ is indeed a witness to $I$ being a yes-instance for \PCP. Indeed, the arithmetic available in $\TC^0$ (in particular, additions and subtraction between numbers of neighbors, which encode positions) allows to ``fill'' the ladder with $0$'s and $1$'s, according to the sequence of indexes $i_1,\cdots,i_k$, and to check that both paths of the ladder coincide.

  By construction, the \PCP-instance $I$ is a yes-instance iff there exists some graph $G$ such that $\rho_I(G)\in\problembis$, i.e. iff $\rho_I\notin\redgen[\calR][\emptyset][\problembis]$.

  For lifting this proof sketch to first-order logic with arithmetic, we recall that in this logic one can add $\log^c n$ numbers, if the size of the domain is $n$ (see, e.g., \cite[Theorems 1.21 and 4.73]{Vollmer1999}). Thus, by ensuring that $\rho_I(G)$ has a large domain, but only a logarithmically small part encodes a solution for \PCP, all computations done in the sketch above can also be done in first-order logic with arithmetic. Implementing this idea is tedious, but not difficult.

	For proving part (1c), we reduce the finite satisfiability problem of the fragment $\forall^\star\exists^\star\FO$ of first-order logic to \redqf[\emptyset][\formbis]. The former problem is undecidable already on the vocabulary of graphs (see, for instance, \cite{BorgerGG2001}). 
  Given a formula \[\formule:=\forall x_1,\cdots\forall x_n,\ \exists y_1,\cdots,\exists y_m,\ \formulebis(x_1,\cdots,x_n,y_1,\cdots,y_m)\] belonging to this fragment, where $\formulebis$ is quantifier-free, we compute the following quantifier-free \FO-interpretation $\inter[\formule]$ from $\{\edgerel\}$ to $\{\edgerel\}$:
  \begin{itemize}
  \item its dimension is $k:=\max(n,m)$,
  \item it is $1$-copying,
  \item it interprets the relation $\edgerel$ as the set of $2k$-uples $(x_1,\cdots,x_k,y_1,\cdots,y_k)$ satisfying \[\formulebis(x_1,\cdots,x_n,y_1,\cdots,y_m)\,.\]
  \end{itemize}
  Now $\inter[\formule]$ is a yes-instance of \redqf[\emptyset][\formbis] iff no graph $\graphe$ is such that $\inter[\formule](\graphe)\models\formbis$, i.e. iff  $\inter[\formule][-1](\formbis)$ is not satisfiable. We conclude by noticing that $\inter[\formule][-1](\formbis)$ is equivalent to $\formule$. 
  
	For proving part (2), suppose that $\calL$ is a fragment of second-order logic with undecidable  finite satisfiability problem. We reduce finite satisfiability of $\calL$ to the complement of $\redgen[\calR][P][\logic]$ by mapping a formula $\varphi \in \calL$ to $(\varphi, \rho_\text{id})$, where $\rho_\text{id}$ is the identity mapping. Then $\varphi$ is satisfiable if and only if $\rho_\text{id}$ is not a reduction from $\emptyset$ to the problem defined by $\varphi$. Similarly, for $\redgen[\calR][\logic][P]$.

\end{proof}

\subsection{Proofs of Section \ref{section:explicit-problems}: Warm-up: Reductions between explicit algorithmic problems}\label{app:explicit-characterisations}

\subsubsection{\texorpdfstring{Characterization of global gadget reductions from {$k$-{\Clique}} to {$\ell$-{\Clique}}}{Characterization of global gadget reductions from k-Clique to l-Clique}}

\theoremConcreteProblemskCliquetolClique*
\begin{proofsketch}
  We first prove that (2) implies (1). So, suppose that $\gadget_\rho \df (V_\rho, E_\rho)$ satisfies conditions (a) -- (c). Let $G = (V, E)$ be an arbitrary graph and define $G^\star \df \rho(G)$. Let $V^\star$ and $E^\star$ be the nodes and edges of $G^\star$. Note that $V^\star = V \uplus V_\rho$. Denote by $A^\star$ the nodes introduced by the set $A$ of nodes of $\gadget_\rho$.

  Suppose that $G$ is a positive instance of $k$-{\Clique}, witnessed by some $k$-clique $U \subseteq V$. Let $B^\star$ be an $(\ell - k)$-clique in $A^\star$, which exists due to condition (b). Then $U \cup B^\star$ is an $\ell$-clique in $G^\star$ and thus $G^\star$ is a positive instance of $\ell$-\Clique.
		
    Now, suppose that $G^\star$ is a positive instance of $\ell$-\Clique and let $U^\star \subseteq V^\star$ be an $\ell$-clique in $G^\star$. By (a), $\gadget_\rho$ has no $\ell$-clique and thus, since only nodes in $A^\star$ have an edge to nodes from~$V$, the clique $U^\star$ consists of nodes from $V$ and nodes from $A^\star$ only. By (c), the largest clique in $A^\star$ has at most size $\ell - k$. Therefore $U^\star$ must have a subclique $U$ of size $k$ with $U \subseteq V$. Thus $G$ is a positive instance of $k$-\Clique.

    We now prove that (1) implies (2). Suppose that condition (b) is violated, but condition (a) holds. Then the $k$-clique $C_{k}$, a positive $k$-\Clique instance, is mapped to a negative $\ell$-\Clique instance (if it was positive, then condition (b) would be true); a contradiction. Suppose now that condition (a) is violated. Then the empty graph (a negative $k$-\Clique instance) is mapped to a positive $\ell$-\Clique instance (consisting of the gadget only); a contradiction.
    Finally suppose that condition (c) is violated. Then the $(k-1)$-clique $C_{k-1}$, a negative $k$-\Clique instance, is mapped to a positive $\ell$-\Clique instance (the clique consisting of $C_{k-1}$ and the $(\ell - k +1)$-clique contained in $A$); a contradiction.
    
    The counterexamples can be easily extracted from the above three cases. 
\end{proofsketch}

\subsubsection{Characterization of edge gadget reductions from {\VC} to {\FVS}}
\theoremConcreteProblemsVCtoFVS*

\begin{proofsketch}
		We first prove that (2) implies (1). Suppose that $\gadget_\rho$ satisfies the conditions (a) and (b). Let $G = (V, E)$ be an arbitrary graph and define $G^\star \df \rho(G)$. Let $V^\star$ and $E^\star$ be the nodes and edges of $G^\star$.

		Suppose that $G$ is a positive instance of $k$-{\VC} with  $k$-vertex cover $U \subseteq V$. Then $U$ is a $k$-feedback vertex set of $G^\star$, since each edge $(u, v)$ of $G$ intersects $U$ and both $\{u\}$ and $\{v\}$ are feedback vertex sets of the gadget introduced for $(u,v)$ in $G^\star$, due  to (a). Thus $G^\star$ is a positive instance of $k$-\FVS. 

    Now, suppose that $G^\star$ is a positive instance of $k$-\FVS with $k$-feedback vertex set $U^\star \subseteq V^\star$. We observe that all nodes of $U^\star$ introduced by one gadget $G_{(u,v)}$ for some edge $(u,v) \in E$ (i.e. nodes that where not present in $V$) can be replaced by $u$ or $v$ due to condition (a). As feedback vertex sets of $G_{(u,v)}$ have at least size 1 due to condition (b), there is a feedback vertex set $U$ of $G^\star$ of size $\leq k$ which only uses nodes from $V$. This $U$ is a $k$-vertex cover of $G$ and thus $G$ is a positive $k$-\VC instance.
    
    We now prove that (1) implies (2) via contraposition. More precisely, we show that if one of the conditions (a) or (b) is violated then $ \rho $ is not an edge gadget reduction from $k$-\VC to $k$-\FVS. 
    
    Suppose that condition (a) is violated but (b) holds. Since then  neither $ \{ c \} $ nor $ \{ d \} $ are feedback vertex sets of $ \gadget_\rho $, this implies that $ \gadget_\rho $ must contain a cycle and that any feedback vertex set of $ \gadget_\rho $ must use (i) at least one gadget node that is neither $ c $ nor $ d $ or (ii) both $ c $ and $ d $. Consider paths $P_\ell$ with $\ell$ nodes and define $P^\star_\ell \df \rho(P_\ell)$. Now, $P_{2k+1}$ is a positive $k$-\VC instance. In case (i), $P^\star_{2k+1}$ is a negative $k$-\FVS instance as it requires more than $ k$ nodes to cover all cycles, since at least one gadget node different from $ u $ or $ v $ must be picked for each edge of $ P_{2k+1} $. In case (ii), two nodes are required for each of the $ 2k $ edges of $ P_{2k+1} $ resulting in a total of at least $ 2k+1 > k $ nodes. Therefore, $P^\star_{2k+1}$ is also a negative $ k $-\FVS instance in case (ii). Thus $ \rho $ is not a reduction.
    
    Now, suppose that condition (b) is violated. The empty set can only be a feedback vertex set of
    $ \gadget_\rho $ if $ \gadget_\rho $ does not contain any cycles. Thus for each graph $G$, the graph $G^\star \df \rho(G)$  contains at most as many cycles as $ G $. Consider the negative $k$-\VC instance $P_{3k+1}$. Since $P_{3k+1}^\star$ contains no cycle, there is a feedback vertex set of size $0$ for $P_{3k+1}^\star$ and therefore it is a positive $k$-\FVS instance. Thus  $ \rho $ is not a reduction. 
    
    This concludes the proof of (1) implies (2).
    
    The graphs $P_\ell$, for suitable $\ell$, can be used as counterexamples.
\end{proofsketch}

\subsubsection{\texorpdfstring{Characterization of node gadget reductions from {\HCd} to {\HCu}}{Characterization of node gadget reductions from HamiltonCycle-d to HamiltonCycle-u}}\label{section:hcd-hcu}

In this section, we characterize restricted node gadget reductions which can be used for reducing \HCd to \HCu. 

For simplicity, we represent node gadget reductions $\rho$ by node gadgets~$\gadget_\rho$. The following definitions are equivalent to the ones provided in Section \ref{section:expressive-power}, yet they fix useful notation.

Informally, a node gadget  $\gadget_\rho$ consists of two copies of a node graph and a set of additional edges between these copies. As an example, the standard reduction from $\HCd$ to $\HCu$ is represented by the node gadget  \begin{tikzpicture}[baseline, scale=0.5]

                 \node[targetnode-old-small] (c1) at (0,0.5) {};
                 \node[targetnode-old-small] (c2) at (0.5,0.5) {};
                 \node[targetnode-old-small] (c3) at (1,0.5) {};
                 \node[targetnode-old-small] (d1) at (0,0) {};
                 \node[targetnode-old-small] (d2) at (0.5,0) {};
                 \node[targetnode-old-small] (d3) at (1,0) {};
                 \draw[targetedge-old-small] (c1) edge (c2);
                 \draw[targetedge-old-small] (c2) edge (c3);
                 \draw[targetedge-old-small] (d1) edge (d2);
                 \draw[targetedge-old-small] (d2) edge (d3);

                 \draw[targetedge-new-small] (c3) edge (d1);
                 \end{tikzpicture} consisting of two copies of the node graph  \begin{tikzpicture}[baseline, scale=0.5]
                 \node[targetnode-old-small] (c1) at (0,0.25) {};
                 \node[targetnode-old-small] (c2) at (0.5,0.25) {};
                 \node[targetnode-old-small] (c3) at (1,0.25) {};
                 \draw[targetedge-old-small] (c1) edge (c2);
                 \draw[targetedge-old-small] (c2) edge (c3);
                 \end{tikzpicture} with one additional edge between them (highlighted in blue).
Formally, a node gadget $\gadget_\rho$ is a graph $(V_\rho, E_\rho)$ such that $V_\rho = V_\triangleright \biguplus V_\triangleleft$ and $E_\rho = E_\triangleright \biguplus E_{\triangleright\triangleleft} \biguplus E_\triangleleft$ such that $V_\triangleright = \{1_\triangleright, \ldots, k_\triangleright\}$ for some $k \in \N$,  $V_\triangleleft = \{1_\triangleleft, \ldots, k_\triangleleft\}$, $E_{\triangleright\triangleleft} \subseteq V_\triangleright \times V_\triangleleft $, and the graphs $G_\triangleright \df (V_\triangleright, E_\triangleright)$ and $G_\triangleleft\df(V_\triangleleft, E_\triangleleft)$ are isomorphic via $\pi$ with $\pi(i_\triangleright) = i_\triangleleft$. The edges in $E_{\triangleright\triangleleft} $ are called \emph{cross-edges} (the edge highlighted in blue in the gadget above is a cross-edge). The isomorphic graphs $G_\triangleright$ and $G_\triangleleft$ are called \emph{node graphs}. 

Applying a node gadget reduction $\rho$ with node gadget $\gadget_\rho$ to a graph $ G=(V,E) $ yields a graph $G^\star \df (V^\star, E^\star)$ in which each node $v \in V$ is replaced by a copy of the node graph, and each edge $(u,v) \in E$ induces cross-edges between the introduced node gadgets. Formally, the node set $V^\star$ is the union of sets $V^\star_u \df \{1_u, \ldots, k_u\}$ of nodes for each node $u \in V$. The edge set $E^\star$ contains the following edges
\begin{itemize}
 \item for all $u \in V$: edges $(i_u, j_u)$, if $(i_\triangleright, j_\triangleright) \in E_\triangleright$ (the edges induced by the node graph for node $u$); and
 \item for all $(u,v) \in E$: edges $(i_u, j_v)$, if $(i_\triangleright, j_\triangleleft) \in E_{\triangleright\triangleleft}$ (the edges induced by cross-edges for the edge $(u,v)$).
 
\end{itemize}

We now prove the following characterization.

\begin{proposition}\label{thm:proposition-hcd-hcu-appendix}
    Let $ \rho $ be a node gadget reduction with node gadget $\gadget_\rho$ and node graphs $ G_\triangleright $ and $ G_\triangleleft $ such that $ G_\triangleright $ and $ G_\triangleleft $ have at most three nodes each. Then the following are equivalent: 
    \begin{enumerate}
        \item $ \rho $ is a reduction from {\HCd} to {\HCu}
        \item $\gadget_\rho$ is either of the following node gadgets:
        \setlength{\tabcolsep}{2mm}
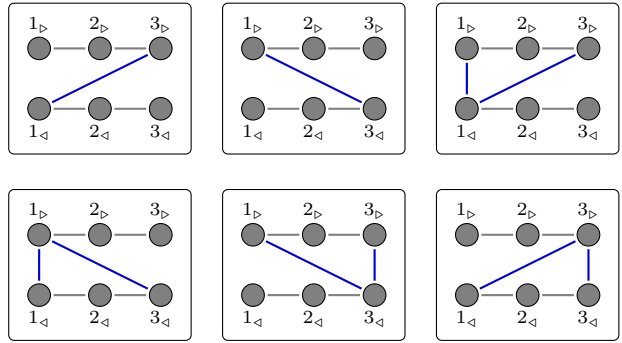
        \begin{figure}[H]
            \centering
            \begin{tabular}{l l l}
                \begin{tikzpicture}[baseline, scale=0.8]
                 \draw[rounded corners = 2pt] (-0.5,-0.75) rectangle (2.5,1.75);
                 \node[targetnode-old-big, label={[above=-0.5mm]:{\scriptsize $1_\triangleright$}}] (c1) at (0,1) {};
                 \node[targetnode-old-big, label={[above=-0.5mm]:{\scriptsize $2_\triangleright$}}] (c2) at (1,1) {};
                 \node[targetnode-old-big, label={[above=-0.5mm]:{\scriptsize $3_\triangleright$}}] (c3) at (2,1) {};
                 \node[targetnode-old-big, label={[below=2.5mm]:{\scriptsize $1_\triangleleft$}}] (d1) at (0,0) {};
                 \node[targetnode-old-big, label={[below=2.5mm]:{\scriptsize $2_\triangleleft$}}] (d2) at (1,0) {};
                 \node[targetnode-old-big, label={[below=2.5mm]:{\scriptsize $3_\triangleleft$}}] (d3) at (2,0) {};
                 \draw[targetedge-old] (c1) edge (c2);
                 \draw[targetedge-old] (c2) edge (c3);
                 \draw[targetedge-old] (d1) edge (d2);
                 \draw[targetedge-old] (d2) edge (d3);

                 \draw[targetedge-new] (c3) edge (d1);
                 \end{tikzpicture}
                 &
                \begin{tikzpicture}[baseline, scale=0.8]
                 \draw[rounded corners = 2pt] (-0.5,-0.75) rectangle (2.5,1.75);
                 \node[targetnode-old-big, label={[above=-0.5mm]:{\scriptsize $1_\triangleright$}}] (c1) at (0,1) {};
                 \node[targetnode-old-big, label={[above=-0.5mm]:{\scriptsize $2_\triangleright$}}] (c2) at (1,1) {};
                 \node[targetnode-old-big, label={[above=-0.5mm]:{\scriptsize $3_\triangleright$}}] (c3) at (2,1) {};
                 \node[targetnode-old-big, label={[below=2.5mm]:{\scriptsize $1_\triangleleft$}}] (d1) at (0,0) {};
                 \node[targetnode-old-big, label={[below=2.5mm]:{\scriptsize $2_\triangleleft$}}] (d2) at (1,0) {};
                 \node[targetnode-old-big, label={[below=2.5mm]:{\scriptsize $3_\triangleleft$}}] (d3) at (2,0) {};
                 \draw[targetedge-old] (c1) edge (c2);
                 \draw[targetedge-old] (c2) edge (c3);
                 \draw[targetedge-old] (d1) edge (d2);
                 \draw[targetedge-old] (d2) edge (d3);

                 \draw[targetedge-new] (c1) edge (d3);
                 \end{tikzpicture}
                 &
                \begin{tikzpicture}[baseline, scale=0.8]
                 \draw[rounded corners = 2pt] (-0.5,-0.75) rectangle (2.5,1.75);
                 \node[targetnode-old-big, label={[above=-0.5mm]:{\scriptsize $1_\triangleright$}}] (c1) at (0,1) {};
                 \node[targetnode-old-big, label={[above=-0.5mm]:{\scriptsize $2_\triangleright$}}] (c2) at (1,1) {};
                 \node[targetnode-old-big, label={[above=-0.5mm]:{\scriptsize $3_\triangleright$}}] (c3) at (2,1) {};
                 \node[targetnode-old-big, label={[below=2.5mm]:{\scriptsize $1_\triangleleft$}}] (d1) at (0,0) {};
                 \node[targetnode-old-big, label={[below=2.5mm]:{\scriptsize $2_\triangleleft$}}] (d2) at (1,0) {};
                 \node[targetnode-old-big, label={[below=2.5mm]:{\scriptsize $3_\triangleleft$}}] (d3) at (2,0) {};
                 \draw[targetedge-old] (c1) edge (c2);
                 \draw[targetedge-old] (c2) edge (c3);
                 \draw[targetedge-old] (d1) edge (d2);
                 \draw[targetedge-old] (d2) edge (d3);

                 \draw[targetedge-new] (c1) edge (d1);
                 \draw[targetedge-new] (c3) edge (d1);
                 \end{tikzpicture}\\~\\
                \begin{tikzpicture}[baseline, scale=0.8]
                 \draw[rounded corners = 2pt] (-0.5,-0.75) rectangle (2.5,1.75);
                 \node[targetnode-old-big, label={[above=-0.5mm]:{\scriptsize $1_\triangleright$}}] (c1) at (0,1) {};
                 \node[targetnode-old-big, label={[above=-0.5mm]:{\scriptsize $2_\triangleright$}}] (c2) at (1,1) {};
                 \node[targetnode-old-big, label={[above=-0.5mm]:{\scriptsize $3_\triangleright$}}] (c3) at (2,1) {};
                 \node[targetnode-old-big, label={[below=2.5mm]:{\scriptsize $1_\triangleleft$}}] (d1) at (0,0) {};
                 \node[targetnode-old-big, label={[below=2.5mm]:{\scriptsize $2_\triangleleft$}}] (d2) at (1,0) {};
                 \node[targetnode-old-big, label={[below=2.5mm]:{\scriptsize $3_\triangleleft$}}] (d3) at (2,0) {};
                 \draw[targetedge-old] (c1) edge (c2);
                 \draw[targetedge-old] (c2) edge (c3);
                 \draw[targetedge-old] (d1) edge (d2);
                 \draw[targetedge-old] (d2) edge (d3);

                 \draw[targetedge-new] (c1) edge (d1);
                 \draw[targetedge-new] (c1) edge (d3);
                 \end{tikzpicture}
                &
                \begin{tikzpicture}[baseline, scale=0.8]
                 \draw[rounded corners = 2pt] (-0.5,-0.75) rectangle (2.5,1.75);
                 \node[targetnode-old-big, label={[above=-0.5mm]:{\scriptsize $1_\triangleright$}}] (c1) at (0,1) {};
                 \node[targetnode-old-big, label={[above=-0.5mm]:{\scriptsize $2_\triangleright$}}] (c2) at (1,1) {};
                 \node[targetnode-old-big, label={[above=-0.5mm]:{\scriptsize $3_\triangleright$}}] (c3) at (2,1) {};
                 \node[targetnode-old-big, label={[below=2.5mm]:{\scriptsize $1_\triangleleft$}}] (d1) at (0,0) {};
                 \node[targetnode-old-big, label={[below=2.5mm]:{\scriptsize $2_\triangleleft$}}] (d2) at (1,0) {};
                 \node[targetnode-old-big, label={[below=2.5mm]:{\scriptsize $3_\triangleleft$}}] (d3) at (2,0) {};
                 \draw[targetedge-old] (c1) edge (c2);
                 \draw[targetedge-old] (c2) edge (c3);
                 \draw[targetedge-old] (d1) edge (d2);
                 \draw[targetedge-old] (d2) edge (d3);

                 \draw[targetedge-new] (c1) edge (d3);
                 \draw[targetedge-new] (c3) edge (d3);
                 \end{tikzpicture}
                &
                \begin{tikzpicture}[baseline, scale=0.8]
                 \draw[rounded corners = 2pt] (-0.5,-0.75) rectangle (2.5,1.75);
                 \node[targetnode-old-big, label={[above=-0.5mm]:{\scriptsize $1_\triangleright$}}] (c1) at (0,1) {};
                 \node[targetnode-old-big, label={[above=-0.5mm]:{\scriptsize $2_\triangleright$}}] (c2) at (1,1) {};
                 \node[targetnode-old-big, label={[above=-0.5mm]:{\scriptsize $3_\triangleright$}}] (c3) at (2,1) {};
                 \node[targetnode-old-big, label={[below=2.5mm]:{\scriptsize $1_\triangleleft$}}] (d1) at (0,0) {};
                 \node[targetnode-old-big, label={[below=2.5mm]:{\scriptsize $2_\triangleleft$}}] (d2) at (1,0) {};
                 \node[targetnode-old-big, label={[below=2.5mm]:{\scriptsize $3_\triangleleft$}}] (d3) at (2,0) {};
                 \draw[targetedge-old] (c1) edge (c2);
                 \draw[targetedge-old] (c2) edge (c3);
                 \draw[targetedge-old] (d1) edge (d2);
                 \draw[targetedge-old] (d2) edge (d3);

                 \draw[targetedge-new] (c3) edge (d1);
                 \draw[targetedge-new] (c3) edge (d3);
                 \end{tikzpicture}
            \end{tabular}%
        \end{figure}
    \end{enumerate}%

    Furthermore, if $ \rho $ is not a reduction from {\HCd} to {\HCu}, a counterexample can be computed efficiently.
\end{proposition}

The proof approach is to brute force through all node gadgets whose node gadget graphs have at most three nodes. To reduce the number of gadgets to be considered, we exploit symmetries. The following two lemmas state useful symmetries and properties for reductions concerning Hamiltonian cycles.

The first lemma holds for directed and undirected Hamiltonian cycles. We say that a node gadget $\gadget_{\rho'} \df (V_{\rho'}, E_{\rho'})$ is a sub-gadget of a node gadget $\gadget_\rho \df (V_\rho, E_\rho)$, if $V_{\rho'} = V_\rho$ and $E_{\rho'} \subseteq E_{\rho}$. Conversely, $ \gadget_\rho $ is called a super-gadget of $ \gadget_{\rho'} $.

\begin{lemma}\label{thm:hc_removal-and-addition-of-edges}
	\begin{enumerate}
	 \item If $\rho, \rho'$ are node gadget reductions such that there is an isomorphism between their node gadgets $\gadget_\rho$ and $\gadget_{\rho'}$ which maps $G_\triangleright$ to $G'_\triangleright$ and $G_\triangleleft$ to $G'_\triangleleft$, then for all graphs $G$: $\rho(G) $ has a Hamiltonian cycle if and only if $ \rho'(G) $ has a Hamiltonian cycle.

    \item Let $\rho, \rho'$ be node gadget reductions with node gadgets $\gadget_\rho \df (V_\rho, E_\rho)$,$\gadget_{\rho'} \df (V_{\rho'}, E_{\rho'})$, respectively, and suppose that $\gadget_{\rho'}$ is a sub-gadget of $\gadget_\rho$.  
    Then for all graphs $ G $:
    \begin{enumerate}
        \item\label{thm:hc_removal-and-addition-of-edges_stay-negative} If $ \rho(G) $ does not have a Hamiltonian cycle, then $\rho'(G)$ does not have a Hamiltonian cycle either.
        \item\label{thm:hc_removal-and-addition-of-edges_stay-positive} If $ \rho'(G) $ has a Hamiltonian cycle, then $\rho(G)$ also has a Hamiltonian cycle.
    \end{enumerate}%
	\end{enumerate}
  
\end{lemma}

\begin{proofsketch}
    \begin{enumerate}
        \item If the condition is satisfied, then $ \rho(G) $ and $ \rho'(G) $ are isomorphic.
        \item If $\gadget_{\rho'}$ is a sub-gadget of $\gadget_{\rho}$ then $\rho(G)$ and  $\rho'(G)$ have the same set of nodes, and edges of $\rho'(G)$ are also edges in  $\rho(G)$.
    \end{enumerate}%
\end{proofsketch}

The next lemma holds for undirected Hamiltonian cycles.

\begin{lemma}\label{lemma:undirectedHam}
	Suppose $\rho, \rho'$ are node gadget reductions with node gadgets $\gadget_\rho = (V_\rho, E_\rho)$,  $\gadget_{\rho'} = (V_{\rho'}, E_{\rho'})$ such that $V_\rho = V_{\rho'} = \{1_\triangleright, \dots, k_\triangleright\} \cup \{1_\triangleleft, \dots, k_\triangleleft\}$
	
	If $\pi: V_\rho \rightarrow V_{\rho'}$ defined as $\pi(i_\triangleright) = i_\triangleleft$ and $\pi(i_\triangleleft) = i_\triangleright$ is an isomorphism, then for every graph $G$: $\rho(G)$ has an undirected Hamiltonian cycle if and only if $\rho'(G)$ has an undirected Hamiltonian cycle. 
\end{lemma}
\begin{proofsketch}
	Undirected Hamiltonian paths in $\rho(G)$ translate to undirected Hamiltonian paths in $\rho'(G)$ via the isomorphism $\pi$.
\end{proofsketch}

The following naming scheme for node gadgets will be helpful. Fix the node set $V_\rho \df \{1_\triangleright, \ldots, k_\triangleright\} \cup \{1_\triangleleft, \ldots, k_\triangleleft\}$. By $P_{(i^1_\triangleright, j^1_\triangleleft), \dots, (i^m_\triangleright, j^m_\triangleleft)}$ we denote the node gadget with cross-edges $(i^1_\triangleright, j^1_\triangleleft), \dots, (i^m_\triangleright, j^m_\triangleleft)$ and where the node graphs form paths $1_\triangleright, \dots, k_\triangleright$ and $1_\triangleleft, \dots, k_\triangleleft$. So, for instance, the standard reduction from directed to undirected Hamiltonian cycle is denoted by $P_{(3_\triangleright, 1_\triangleleft)}$. %

\begin{proof}[Proof sketch (of Proposition \ref{thm:proposition-hcd-hcu-appendix})]
    We first prove that (2) implies (1). The gadget $ P_{(3_\triangleright,1_\triangleleft)} $ is the standard gadget for this reduction \cite{Karp72}. It remains to prove that the gadget $P_{(1_\triangleright,1_\triangleleft), (3_\triangleright,1_\triangleleft)}$ is correct, since all other gadgets stated in the proposition are symmetric to $ P_{(3_\triangleright,1_\triangleleft)} $ or $P_{(1_\triangleright,1_\triangleleft), (3_\triangleright,1_\triangleleft)}$, see Figure~\ref{fig:valid-hc-gadgets}.

    We now argue that the gadget $ P_{(1_\triangleright,1_\triangleleft), (3_\triangleright,1_\triangleleft)} $ is a correct node gadget. 
    
    The gadget correctly maps positive instances of \HCd to positive instances of $ \HCu $: this follows from Lemma \ref{thm:hc_removal-and-addition-of-edges}, since this gadget is obtained by adding the edge $ (1_\triangleright,1_\triangleleft) $ to $P_{(3_\triangleright,1_\triangleleft)}$.
    
    We now show that this gadget maps negative instances of \HCd to negative instances of $\HCu$ by contradiction.  Suppose that a negative instance $G$ of $\HCd$ is mapped to a positive instance $G^\star$ of $\HCu$. Any undirected Hamiltonian cycle of $G^\star$ must use at least one edge induced by the cross-edge $(1_\triangleright, 1_\triangleleft)$, because the sub-gadget $P_{(3_\triangleright,1_\triangleleft)}$ without this edge is a correct gadget. So suppose the edge $(1_u, 1_v)$ introduced for the edge $ (u,v) $ is used, see Figure \ref{fig:valid-gadget-edge-instantiation}. The edges $ (1_u, 2_u), (2_u, 3_u), (1_v, 2_v) $ and $ (2_v, 3_v) $ must be used on any Hamiltonian cycle, since $2_u$ and $2_v$ only have two adjacent edges. Thus, the Hamiltonian cycle must pass the gadget introduced for $(u,v)$ via the sequence $3_u, 2_u, 1_u, 1_v, 2_v, 3_v$ or via its reverse. This leaves two "open" ends, $ 3_u $ and $ 3_v $, indicated by a red border in Figure \ref{fig:valid-gadget-edge-instantiation}.
    
    The only way to leave $ 3_v $ is via a gadget inserted for an edge $ (v,w) $, adding edges $ (3_v, 1_w), (1_w, 2_w), (2_w, 3_w) $ to the supposed Hamiltonian cycle, as indicated in Figure \ref{fig:valid-gadget-path-instantiation}. But now we are in the same situation as before: two "open" ends which are now $ 3_u $ and $3_w$. By induction, it follows that there are no undirected Hamiltonian cycles in $G^\star$ which use an edge induced by the node gadget edge $(1_\triangleright, 1_\triangleleft)$. Thus the only undirected Hamiltonian cycles are those which are also present when applying $ P_{(3_\triangleright,1_\triangleleft)}$ to $G$. Since $ P_{(3_\triangleright,1_\triangleleft)}$ is a correct gadget, $ G^\star $ does not have an undirected Hamiltonian cycle; a contradiction.

    This concludes the implication from (2) to (1).

    \tikzstyle{openEndNode} = [draw = black!30!red, thick]
    
    \begin{figure}[tp]
    \begin{subfigure}[tp]{0.45\textwidth}
        \centering
        \begin{tikzpicture}[baseline, scale=0.8]
         \node[vgnode-big, label={[above=-0.5mm]:{\scriptsize $ 1_u $}}] (c1) at (0,1) {};
         \node[vgnode-big, label={[above=-0.5mm]:{\scriptsize $ 2_u $}}] (c2) at (1,1) {};
         \node[vgnode-big, openEndNode, label={[above=-0.5mm]:{\scriptsize $ 3_u $}}] (c3) at (2,1) {};
         \node[vgnode-big, label={[below=2.5mm]:{\scriptsize $ 1_v $}}] (d1) at (0,0) {};
         \node[vgnode-big, label={[below=2.5mm]:{\scriptsize $ 2_v $}}] (d2) at (1,0) {};
         \node[vgnode-big, openEndNode, label={[below=2.5mm]:{\scriptsize $ 3_v $}}] (d3) at (2,0) {};

         \draw[witnessedge] (d1) edge (d2);
         \draw[witnessedge] (c2) edge (c3);
         \draw[witnessedge] (c1) edge (d1);
         \draw[witnessedge] (d2) edge (d3);
         \draw[witnessedge] (c1) edge (c2);

         \draw[vgedge] (c1) edge (c2);
         \draw[vgedge] (c2) edge (c3);
         \draw[vgedge] (d1) edge (d2);
         \draw[vgedge] (d2) edge (d3);

         \draw[crossedge] (c1) edge (d1);
         \draw[crossedge] (c3) edge (d1);
         \end{tikzpicture}
         \caption{Instantiation of the node gadget $ P_{(1_\triangleright,1_\triangleleft),(3_\triangleright,1_\triangleleft)} $ for a directed edge $ (u,v) $.}
         \label{fig:valid-gadget-edge-instantiation}
    \end{subfigure}
    \hfill
    \begin{subfigure}[tp]{0.45\textwidth}
        \centering
        \begin{tikzpicture}[baseline, scale=0.8]
         \node[vgnode-big, label={[above=-0.5mm]:{\scriptsize $ 1_u $}}] (c1) at (0,1) {};
         \node[vgnode-big, label={[above=-0.5mm]:{\scriptsize $ 2_u $}}] (c2) at (1,1) {};
         \node[vgnode-big, openEndNode, label={[above=-0.5mm]:{\scriptsize $ 3_u $}}] (c3) at (2,1) {};
         \node[vgnode-big, label={[left=0.2mm]:{\scriptsize $ 1_v $}}] (d1) at (0,0) {};
         \node[vgnode-big, label={[right=0.2mm]:{\scriptsize $ 2_v $}}] (d2) at (1,0) {};
         \node[vgnode-big, label={[right=0.2mm]:{\scriptsize $ 3_v $}}] (d3) at (2,0) {};
         \node[vgnode-big, label={[below=2.5mm]:{\scriptsize $ 1_w $}}] (e1) at (0,-1) {};
         \node[vgnode-big, label={[below=2.5mm]:{\scriptsize $ 2_w $}}] (e2) at (1,-1) {};
         \node[vgnode-big, openEndNode, label={[below=2.5mm]:{\scriptsize $ 3_w $}}] (e3) at (2,-1) {};

         \draw[witnessedge] (c1) edge (c2);
         \draw[witnessedge] (c2) edge (c3);
         \draw[witnessedge] (d1) edge (d2);
         \draw[witnessedge] (d2) edge (d3);
         \draw[witnessedge] (e1) edge (e2);
         \draw[witnessedge] (e2) edge (e3);
         \draw[witnessedge] (d3) edge (e1);
         \draw[witnessedge] (c1) edge (d1);

         \draw[vgedge] (c1) edge (c2);
         \draw[vgedge] (c2) edge (c3);
         \draw[vgedge] (d1) edge (d2);
         \draw[vgedge] (d2) edge (d3);
         \draw[vgedge] (e1) edge (e2);
         \draw[vgedge] (e2) edge (e3);

         \draw[crossedge] (c1) edge (d1);
         \draw[crossedge] (c3) edge (d1);
         \draw[crossedge] (d1) edge (e1);
         \draw[crossedge] (d3) edge (e1);
         \end{tikzpicture}
         \caption{Instantiation of the node gadget $ P_{(1_\triangleright,1_\triangleleft),(3_\triangleright,1_\triangleleft)} $ for a directed path $ (u,v),(v,w) $.}
         \label{fig:valid-gadget-path-instantiation}
    \end{subfigure}
    \caption{Illustration that an edge created by the node gadget edge $ (1_\triangleright, 1_\triangleleft) $ of the node gadget $ P_{(1_\triangleright,1_\triangleleft),(3_\triangleright,1_\triangleleft)} $ is never used in an undirected Hamiltonian cycle.}
    \label{fig:valid-gadget-edge-path-instantiation}
    \end{figure}
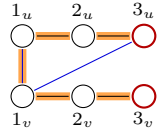
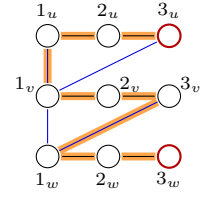

    We now show that (1) implies (2) by outlining why all other gadgets are invalid. A gadget can be invalid for two reasons: (a) it maps some positive instance of \HCd to a negative instance of \HCu, or (b) it maps some negative instance of \HCd to a positive instance of \HCu. We will call the gadget a $ \postoneg $ gadget in case~(a), and a $ \negtopos $ gadget in case (b). In particular, a gadget can be both a $\postoneg$ and a $\negtopos$ gadget.  Due to Lemma \ref{thm:hc_removal-and-addition-of-edges} (\ref{thm:hc_removal-and-addition-of-edges_stay-negative}) and (\ref{thm:hc_removal-and-addition-of-edges_stay-positive}), we do not need to check sub-gadgets of $ \postoneg $ gadgets and super-gadgets of $ \negtopos $ gadgets. Also, if the node graphs themselves already form a Hamiltonian cycle, the gadget is invalid since the (directed) graph with a single node and no self-loops is a counterexample.

    The gadgets with node graphs with one or two nodes are all invalid, with counterexamples provided by Figure \ref{fig:invalid-hc-lt3-gadgets}. All node gadgets not present in the figure  are invalid because they are symmetric to one of the illustrated gadgets, sub-gadgets of an illustrated \postoneg gadget, or super-gadgets of an illustrated \negtopos gadget.
    
    It remains to show that node gadgets with node graphs with exactly three nodes are invalid, if they are not one of the gadgets in the statement of the proposition. 
    We distinguish between the cases that (a) the node graphs are paths or (b) they are not.

    We start with case (a). Most invalid gadgets of this type can be disregarded due to the  $\postoneg$ gadget depicted in Figure \ref{fig:full-pos-to-neg}. In this gadget, only three cross-edges are missing: $ (1_\triangleright, 3_\triangleleft), (3_\triangleright,1_\triangleleft) $ and $ (3_\triangleright,3_\triangleleft) $. Since this is a \postoneg gadget (see below), all of its sub-gadgets are also \postoneg gadgets and therefore invalid. The remaining invalid gadgets for case (a) are covered in Figure \ref{fig:invalid-hc-gadgets}. They are all \negtopos gadgets, as evidenced by the example graphs in the figure, and therefore all of their super-gadgets are also \negtopos gadgets by Lemma~\ref{thm:hc_removal-and-addition-of-edges}~(\ref{thm:hc_removal-and-addition-of-edges_stay-negative}). In summary, all gadgets for case (a) are either valid gadgets, sub-gadgets of the gadget in Figure \ref{fig:full-pos-to-neg_gadget}, or super-gadgets of a gadget in Figure \ref{fig:invalid-hc-gadgets}, or symmetric to any of these (according to Lemmas \ref{thm:hc_removal-and-addition-of-edges} and \ref{lemma:undirectedHam}).

    For case (a), it remains to show that the gadget in Figure \ref{fig:full-pos-to-neg} is indeed a $\postoneg$ gadget. Applying it to the directed graph in Figure \ref{fig:full-pos-to-neg_cex} yields the undirected graph $ G^\star $ in Figure \ref{fig:full-pos-to-neg_applied}. This graph does not have a Hamiltonian cycle.
    To see this, it is helpful to use the following representation in which the edge relation is implicitly given by a function $ e $ that maps nodes to the set of nodes they are adjacent to:
    \begin{align*}
        G^\star &= (V, E)\\
        V &= \{i_u, i_v, i_w \mid 1 \le i \le 3 \}\\
        E &= \{ (x,y) \mid x \in V, y \in e(x) \}\\
        e(x) &= \begin{cases}
            \{ 2_u, 2_v, 2_w\}, & x \in \{ 3_u, 3_v, 3_w \}\\
            V \setminus \{ x \}, & x \in \{ 2_u, 2_v, 2_w \}\\
            \{ 2_u, 2_v, 2_w\} \cup (\{ 1_u, 1_v, 1_w \} \setminus \{ x \}), & x \in \{ 1_u, 1_v, 1_w \}
        \end{cases}
    \end{align*}
    We will refer to the nodes $ 1_u, 1_v, 1_w $ as the in-nodes, to $ 2_u, 2_v, 2_w $ as the middle-nodes and to $ 3_u, 3_v, 3_w $ as the out-nodes. Note that every node is adjacent to the middle-nodes while out-nodes only have edges to the middle-nodes. Also, an in-node only has edges to the other in-nodes and the middle-nodes. The best possible way to visit all out-nodes on a potential Hamiltonian cycle would be a sequence like  
    $ 3_u - 2_u - 3_v - 2_v - 3_w - 2_w $, i.e. to connect two out-nodes via a single middle-node. This is desirable since the middle-nodes are the ones that can be used most flexibly%
    \footnote{The alternatives would require one to use more than one middle-node for a single pair of out-nodes which would remove the ability to "leave" the remaining out-node.}. %
    The nodes do not have to be in this exact order, in general the sequence has the form $ 3_x - 2_x - 3_y - 2_y - 3_z - 2_z $, where $ \{ x,y,z \} = \{ u,v,w \} $%
    . Note that this sequence already contains all middle-nodes. While it is now possible to extend this sequence to a Hamiltonian path ending in some in-node, it cannot be completed to a Hamiltonian cycle since the final in-node can only be left via a middle-node and all of them have previously been visited but are not the start of the Hamiltonian path. As it was necessary to include the sequence above, it follows that the graph does not have a Hamiltonian cycle. 

    This only leaves gadgets corresponding to case (b), i.e. gadgets whose node graphs do not contain a path. The general approach to exhaustively check all gadgets is to start from a \negtopos gadget and remove edges until the result is a \postoneg gadget. This can be observed in Figure \ref{fig:invalid-nonpath-hc-gadgets} in which the first gadget $ I_{(1_\triangleright,1_\triangleleft),(2_\triangleright,1_\triangleleft), (2_\triangleright,2_\triangleleft),(3_\triangleright,2_\triangleleft) (3_\triangleright,3_\triangleleft)} $  is a \negtopos gadget, and the two following gadgets $ I_{(1_\triangleright,1_\triangleleft),(2_\triangleright,1_\triangleleft),(2_\triangleright,2_\triangleleft),(3_\triangleright,3_\triangleleft)} $ and $ I_{(1_\triangleright,1_\triangleleft),(2_\triangleright,1_\triangleleft),(3_\triangleright,2_\triangleleft),(3_\triangleright,3_\triangleleft)} $ represent all remaining sub-gadgets (due to symmetries). The only remaining gadgets are obtained by adding $ (1_\triangleright, 3_\triangleleft) $ or $ (3_\triangleright, 1_\triangleleft) $ to the \postoneg gadgets in Figure \ref{fig:invalid-nonpath-hc-gadgets} as well as their sub-gadgets. An example of such a gadget is $ N_{(1_\triangleright,3_\triangleleft), (3_\triangleright,2_\triangleleft), (3_\triangleright,3_\triangleleft)} $ shown in Figure \ref{fig:diag-pos-to-neg}. By using the described approach one can easily verify that neither of these are valid, either.

    The counter-examples given in Figures \ref{fig:invalid-hc-lt3-gadgets}, \ref{fig:full-pos-to-neg}, \ref{fig:invalid-hc-gadgets}, \ref{fig:invalid-nonpath-hc-gadgets} and \ref{fig:diag-pos-to-neg} can be used to provide counter-example feedback.
\end{proof}

\newcommand{\exampleAppliedSpace}{\hspace{2em}}

\begin{figure}
    \centering
    \def\arraystretch{3}
    \begin{tabular}{l l l}
        Name & Gadget & Reason for correctness \vspace{-1em}\\
        \hline
        $ P_{(3_\triangleright,1_\triangleleft)} $ &
        \begin{tikzpicture}[baseline, scale=0.6]
         \draw[rounded corners = 2pt] (-0.5,-0.5) rectangle (2.5,1.5);
         \node[targetnode-old] (c1) at (0,1) {};
         \node[targetnode-old] (c2) at (1,1) {};
         \node[targetnode-old] (c3) at (2,1) {};
         \node[targetnode-old] (d1) at (0,0) {};
         \node[targetnode-old] (d2) at (1,0) {};
         \node[targetnode-old] (d3) at (2,0) {};
         \draw[targetedge-old] (c1) edge (c2);
         \draw[targetedge-old] (c2) edge (c3);
         \draw[targetedge-old] (d1) edge (d2);
         \draw[targetedge-old] (d2) edge (d3);

         \draw[targetedge-new] (c3) edge (d1);
         \end{tikzpicture}
                      & This is the standard gadget\\%
        $ P_{(1_\triangleright,3_\triangleleft)} $ &
        \begin{tikzpicture}[baseline, scale=0.6]
         \draw[rounded corners = 2pt] (-0.5,-0.5) rectangle (2.5,1.5);

         \node[targetnode-old] (c1) at (0,1) {};
         \node[targetnode-old] (c2) at (1,1) {};
         \node[targetnode-old] (c3) at (2,1) {};
         \node[targetnode-old] (d1) at (0,0) {};
         \node[targetnode-old] (d2) at (1,0) {};
         \node[targetnode-old] (d3) at (2,0) {};
         \draw[targetedge-old] (c1) edge (c2);
         \draw[targetedge-old] (c2) edge (c3);
         \draw[targetedge-old] (d1) edge (d2);
         \draw[targetedge-old] (d2) edge (d3);

         \draw[targetedge-new] (c1) edge (d3);
         \end{tikzpicture}
                      & Symmetric to $P_{(1_\triangleright,3_\triangleleft)}$\\
        $ P_{(1_\triangleright,1_\triangleleft), (3_\triangleright,1_\triangleleft)} $ & 
        \begin{tikzpicture}[baseline, scale=0.6]
         \draw[rounded corners = 2pt] (-0.5,-0.5) rectangle (2.5,1.5);

         \node[targetnode-old] (c1) at (0,1) {};
         \node[targetnode-old] (c2) at (1,1) {};
         \node[targetnode-old] (c3) at (2,1) {};
         \node[targetnode-old] (d1) at (0,0) {};
         \node[targetnode-old] (d2) at (1,0) {};
         \node[targetnode-old] (d3) at (2,0) {};
         \draw[targetedge-old] (c1) edge (c2);
         \draw[targetedge-old] (c2) edge (c3);
         \draw[targetedge-old] (d1) edge (d2);
         \draw[targetedge-old] (d2) edge (d3);

         \draw[targetedge-new] (c1) edge (d1);
         \draw[targetedge-new] (c3) edge (d1);
         \end{tikzpicture}
                             & See proof of Proposition \ref{thm:proposition-hcd-hcu-appendix}\\
        $ P_{(1_\triangleright,1_\triangleleft), (1_\triangleright,3_\triangleleft)} $ &
        \begin{tikzpicture}[baseline, scale=0.6]
         \draw[rounded corners = 2pt] (-0.5,-0.5) rectangle (2.5,1.5);

         \node[targetnode-old] (c1) at (0,1) {};
         \node[targetnode-old] (c2) at (1,1) {};
         \node[targetnode-old] (c3) at (2,1) {};
         \node[targetnode-old] (d1) at (0,0) {};
         \node[targetnode-old] (d2) at (1,0) {};
         \node[targetnode-old] (d3) at (2,0) {};
         \draw[targetedge-old] (c1) edge (c2);
         \draw[targetedge-old] (c2) edge (c3);
         \draw[targetedge-old] (d1) edge (d2);
         \draw[targetedge-old] (d2) edge (d3);

         \draw[targetedge-new] (c1) edge (d1);
         \draw[targetedge-new] (c1) edge (d3);
         \end{tikzpicture}
                             & Symmetric to $ P_{(1_\triangleright,1_\triangleleft), (3_\triangleright,1_\triangleleft)} $\\
        $ P_{(1_\triangleright,3_\triangleleft), (3_\triangleright,3_\triangleleft)} $ &
        \begin{tikzpicture}[baseline, scale=0.6]
         \draw[rounded corners = 2pt] (-0.5,-0.5) rectangle (2.5,1.5);

         \node[targetnode-old] (c1) at (0,1) {};
         \node[targetnode-old] (c2) at (1,1) {};
         \node[targetnode-old] (c3) at (2,1) {};
         \node[targetnode-old] (d1) at (0,0) {};
         \node[targetnode-old] (d2) at (1,0) {};
         \node[targetnode-old] (d3) at (2,0) {};
         \draw[targetedge-old] (c1) edge (c2);
         \draw[targetedge-old] (c2) edge (c3);
         \draw[targetedge-old] (d1) edge (d2);
         \draw[targetedge-old] (d2) edge (d3);

         \draw[targetedge-new] (c1) edge (d3);
         \draw[targetedge-new] (c3) edge (d3);
         \end{tikzpicture}
                             & Symmetric to $ P_{(1_\triangleright,1_\triangleleft), (3_\triangleright,1_\triangleleft)} $\\
        $ P_{(3_\triangleright,1_\triangleleft), (3_\triangleright,3_\triangleleft)} $ &
        \begin{tikzpicture}[baseline, scale=0.6]
         \draw[rounded corners = 2pt] (-0.5,-0.5) rectangle (2.5,1.5);

         \node[targetnode-old] (c1) at (0,1) {};
         \node[targetnode-old] (c2) at (1,1) {};
         \node[targetnode-old] (c3) at (2,1) {};
         \node[targetnode-old] (d1) at (0,0) {};
         \node[targetnode-old] (d2) at (1,0) {};
         \node[targetnode-old] (d3) at (2,0) {};
         \draw[targetedge-old] (c1) edge (c2);
         \draw[targetedge-old] (c2) edge (c3);
         \draw[targetedge-old] (d1) edge (d2);
         \draw[targetedge-old] (d2) edge (d3);

         \draw[targetedge-new] (c3) edge (d1);
         \draw[targetedge-new] (c3) edge (d3);
         \end{tikzpicture}
                             & Symmetric to $ P_{(1_\triangleright,1_\triangleleft), (3_\triangleright,1_\triangleleft)} $
    \end{tabular}%
    \caption{Correct node gadgets for reducing \HCd to \HCu with node graphs with at most three nodes.}
    \label{fig:valid-hc-gadgets}
\end{figure}
\begin{figure}[tp]
    \centering
    \def\arraystretch{3}
    \begin{tabular}{c c c c}
        Gadget & Classification & $G$ & \hspace{2.5em}$\rho(G)$ \vspace{-1em}\\
        \hline
        \begin{tikzpicture}[baseline, scale=0.6]
         \draw[rounded corners = 2pt] (-0.4,-0.4) rectangle (1.4,1.4);

         \node[targetnode-old] (c1) at (0.5,1) {};
         \node[targetnode-old] (d1) at (0.5,0) {};
         \end{tikzpicture}
                           & \postoneg &
        \begin{tikzpicture}[baseline, scale=0.6]
         \node[hcdnode] (v1) at (0,1) {};
         \node[hcdnode] (v2) at (2,1) {};
         \node[hcdnode] (v3) at (1,0) {};

         \draw[hcdedge] (v1) edge (v2);
         \draw[hcdedge] (v2) edge (v3);
         \draw[hcdedge] (v3) edge (v1);
         \end{tikzpicture}
                           &
        \exampleAppliedSpace
        \begin{tikzpicture}[baseline, scale=0.6]
         \node[hcunode] (v1) at (0,1) {};
         \node[hcunode] (v2) at (2,1) {};
         \node[hcunode] (v3) at (1,0) {};
         \end{tikzpicture}\\
        \begin{tikzpicture}[baseline, scale=0.6]
         \draw[rounded corners = 2pt] (-0.4,-0.4) rectangle (1.4,1.4);

         \node[targetnode-old] (c1) at (0.5,1) {};
         \node[targetnode-old] (d1) at (0.5,0) {};

         \draw[targetedge-new-small] (c1) edge (d1);
         \end{tikzpicture}
                           & \negtopos &
        \begin{tikzpicture}[baseline, scale=0.6]
         \node[hcdnode] (v1) at (0,1) {};
         \node[hcdnode] (v2) at (2,1) {};
         \node[hcdnode] (v3) at (1,0) {};

         \draw[hcdedge] (v1) edge (v2);
         \draw[hcdedge] (v2) edge (v3);
         \draw[hcdedge] (v1) edge (v3);
         \end{tikzpicture}
                           &
        \exampleAppliedSpace
        \begin{tikzpicture}[baseline, scale=0.6]
         \node[hcunode] (v1) at (0,1) {};
         \node[hcunode] (v2) at (2,1) {};
         \node[hcunode] (v3) at (1,0) {};

         \draw[hcuedge] (v1) edge (v2);
         \draw[hcuedge] (v2) edge (v3);
         \draw[hcuedge] (v1) edge (v3);
         \end{tikzpicture}\\
        \begin{tikzpicture}[baseline, scale=0.6]
         \draw[rounded corners = 2pt] (-0.4,-0.4) rectangle (1.4,1.4);

         \node[targetnode-old] (c1) at (0,1) {};
         \node[targetnode-old] (c2) at (1,1) {};
         \node[targetnode-old] (d1) at (0,0) {};
         \node[targetnode-old] (d2) at (1,0) {};
         \end{tikzpicture}
                           & \postoneg &
        \begin{tikzpicture}[baseline, scale=0.6]
         \node[hcdnode] (v1) at (0,1) {};
         \node[hcdnode] (v2) at (2,1) {};
         \node[hcdnode] (v3) at (1,0) {};

         \draw[hcdedge] (v1) edge (v2);
         \draw[hcdedge] (v2) edge (v3);
         \draw[hcdedge] (v3) edge (v1);
         \end{tikzpicture}
                           &
        \exampleAppliedSpace
        \begin{tikzpicture}[baseline, scale=0.6]
         \node[hcunode] (11) at (0,1) {};
         \node[hcunode] (12) at (1,1) {};
         \node[hcunode] (21) at (3,1) {};
         \node[hcunode] (22) at (4,1) {};
         \node[hcunode] (31) at (1.5,0) {};
         \node[hcunode] (32) at (2.5,0) {};
         \end{tikzpicture}\\
        \begin{tikzpicture}[baseline, scale=0.6]
         \draw[rounded corners = 2pt] (-0.4,-0.4) rectangle (1.4,1.4);

         \node[targetnode-old] (c1) at (0,1) {};
         \node[targetnode-old] (c2) at (1,1) {};
         \node[targetnode-old] (d1) at (0,0) {};
         \node[targetnode-old] (d2) at (1,0) {};

         \draw[targetedge-new-small] (c1) edge (d1);
         \end{tikzpicture}
                           & \postoneg &
        \begin{tikzpicture}[baseline, scale=0.6]
         \node[hcdnode] (v1) at (0,1) {};
         \node[hcdnode] (v2) at (2,1) {};
         \node[hcdnode] (v3) at (1,0) {};

         \draw[hcdedge] (v1) edge (v2);
         \draw[hcdedge] (v2) edge (v3);
         \draw[hcdedge] (v3) edge (v1);
         \end{tikzpicture}
                           &
        \exampleAppliedSpace
        \begin{tikzpicture}[baseline, scale=0.6]
         \node[hcunode] (11) at (0,1) {};
         \node[hcunode] (12) at (1,1) {};
         \node[hcunode] (21) at (3,1) {};
         \node[hcunode] (22) at (4,1) {};
         \node[hcunode] (31) at (1.5,0) {};
         \node[hcunode] (32) at (2.5,0) {};

         \draw[crossedge] (11) edge[bend left] (21);
         \draw[crossedge] (21) edge (31);
         \draw[crossedge] (31) edge (11);
         \end{tikzpicture}\\
        \begin{tikzpicture}[baseline, scale=0.6]
         \draw[rounded corners = 2pt] (-0.4,-0.4) rectangle (1.4,1.4);

         \node[targetnode-old] (c1) at (0,1) {};
         \node[targetnode-old] (c2) at (1,1) {};
         \node[targetnode-old] (d1) at (0,0) {};
         \node[targetnode-old] (d2) at (1,0) {};
         \draw[targetedge-old-small] (c1) edge (c2);
         \draw[targetedge-old-small] (d1) edge (d2);
         \end{tikzpicture}
                           & \postoneg &
        \begin{tikzpicture}[baseline, scale=0.6]
         \node[hcdnode] (v1) at (0,1) {};
         \node[hcdnode] (v2) at (2,1) {};
         \node[hcdnode] (v3) at (1,0) {};

         \draw[hcdedge] (v1) edge (v2);
         \draw[hcdedge] (v2) edge (v3);
         \draw[hcdedge] (v3) edge (v1);
         \end{tikzpicture}
                           &
        \exampleAppliedSpace
        \begin{tikzpicture}[baseline, scale=0.6]
         \node[hcunode] (11) at (0,1) {};
         \node[hcunode] (12) at (1,1) {};
         \node[hcunode] (21) at (3,1) {};
         \node[hcunode] (22) at (4,1) {};
         \node[hcunode] (31) at (1.5,0) {};
         \node[hcunode] (32) at (2.5,0) {};
         \draw[hcuedge] (11) edge (12);
         \draw[hcuedge] (21) edge (22);
         \draw[hcuedge] (31) edge (32);
         \end{tikzpicture}\\
        \begin{tikzpicture}[baseline, scale=0.6]
         \draw[rounded corners = 2pt] (-0.4,-0.4) rectangle (1.4,1.4);

         \node[targetnode-old] (c1) at (0,1) {};
         \node[targetnode-old] (c2) at (1,1) {};
         \node[targetnode-old] (d1) at (0,0) {};
         \node[targetnode-old] (d2) at (1,0) {};
         \draw[targetedge-old-small] (c1) edge (c2);
         \draw[targetedge-old-small] (d1) edge (d2);

         \draw[targetedge-new-small] (c1) edge (d1);
         \end{tikzpicture}
                           & \postoneg &
        \begin{tikzpicture}[baseline, scale=0.6]
         \node[hcdnode] (v1) at (0,1) {};
         \node[hcdnode] (v2) at (2,1) {};
         \node[hcdnode] (v3) at (1,0) {};

         \draw[hcdedge] (v1) edge (v2);
         \draw[hcdedge] (v2) edge (v3);
         \draw[hcdedge] (v3) edge (v1);
         \end{tikzpicture}
                           &
        \exampleAppliedSpace
        \begin{tikzpicture}[baseline, scale=0.6]
         \node[hcunode] (11) at (0,1) {};
         \node[hcunode] (12) at (1,1) {};
         \node[hcunode] (21) at (3,1) {};
         \node[hcunode] (22) at (4,1) {};
         \node[hcunode] (31) at (1.5,0) {};
         \node[hcunode] (32) at (2.5,0) {};
         \draw[hcuedge] (11) edge (12);
         \draw[hcuedge] (21) edge (22);
         \draw[hcuedge] (31) edge (32);
         \draw[crossedge] (11) edge[bend left] (21);
         \draw[crossedge] (21) edge (31);
         \draw[crossedge] (31) edge (11);
         \end{tikzpicture}\\
        \begin{tikzpicture}[baseline, scale=0.6]
         \draw[rounded corners = 2pt] (-0.4,-0.4) rectangle (1.4,1.4);

         \node[targetnode-old] (c1) at (0,1) {};
         \node[targetnode-old] (c2) at (1,1) {};
         \node[targetnode-old] (d1) at (0,0) {};
         \node[targetnode-old] (d2) at (1,0) {};
         \draw[targetedge-old-small] (c1) edge (c2);
         \draw[targetedge-old-small] (d1) edge (d2);
         \draw[targetedge-new-small] (c2) edge (d1);
         \end{tikzpicture}
                           & \negtopos &
        \begin{tikzpicture}[baseline, scale=0.6]
         \node[hcdnode] (v1) at (0,2) {};
         \node[hcdnode] (v2) at (0,1) {};
         \node[hcdnode] (v3) at (0,0) {};

         \draw[hcdedge] (v1) edge[bend left=45] (v2);
         \draw[hcdedge] (v2) edge[bend left=45] (v1);
         \draw[hcdedge] (v3) edge[bend left=45] (v2);
         \draw[hcdedge] (v2) edge[bend left=45] (v3);
         \end{tikzpicture}
                           &
        \exampleAppliedSpace
        \begin{tikzpicture}[baseline, scale=0.6]
         \node[hcunode] (11) at (0,2) {};
         \node[hcunode] (12) at (1,2) {};
         \node[hcunode] (21) at (0,1) {};
         \node[hcunode] (22) at (1,1) {};
         \node[hcunode] (31) at (0,0) {};
         \node[hcunode] (32) at (1,0) {};

         \draw[hcuedge] (11) edge (12);
         \draw[hcuedge] (21) edge (22);
         \draw[hcuedge] (31) edge (32);

         \draw[crossedge] (11) edge (22);
         \draw[crossedge] (21) edge (12);
         \draw[crossedge] (21) edge (32);
         \draw[crossedge] (31) edge (22);
         \end{tikzpicture}
    \end{tabular}
    \caption{Representative node gadgets with node graphs with 1 or 2 nodes which do not reduce \HCd to \HCu. }
    \label{fig:invalid-hc-lt3-gadgets}
\end{figure}
\begin{figure}[tp]
    \centering
    \begin{subfigure}[t]{0.3\textwidth}
        \centering
        \begin{tikzpicture}
         \draw[rounded corners = 2pt] (-0.5,-0.55) rectangle (2.5,1.55);

         \node[targetnode-old-big, label={[above=-0.5mm]:{\scriptsize $ 1_\triangleright $}}] (c1) at (0,1) {};
         \node[targetnode-old-big, label={[above=-0.5mm]:{\scriptsize $ 2_\triangleright $}}] (c2) at (1,1) {};
         \node[targetnode-old-big, label={[above=-0.5mm]:{\scriptsize $ 3_\triangleright $}}] (c3) at (2,1) {};
         \node[targetnode-old-big, label={[below=2.5mm]:{\scriptsize $ 1_\triangleleft $}}] (d1) at (0,0) {};
         \node[targetnode-old-big, label={[below=2.5mm]:{\scriptsize $ 2_\triangleleft $}}] (d2) at (1,0) {};
         \node[targetnode-old-big, label={[below=2.5mm]:{\scriptsize $ 3_\triangleleft $}}] (d3) at (2,0) {};
         \draw[targetedge-old] (c1) edge (c2);
         \draw[targetedge-old] (c2) edge (c3);
         \draw[targetedge-old] (d1) edge (d2);
         \draw[targetedge-old] (d2) edge (d3);

         \draw[targetedge-new] (c1) edge (d1);
         \draw[targetedge-new] (c1) edge (d2);
         \draw[targetedge-new] (c2) edge (d1);
         \draw[targetedge-new] (c2) edge (d2);
         \draw[targetedge-new] (c2) edge (d3);
         \draw[targetedge-new] (c3) edge (d2);
         \end{tikzpicture}
         \caption{The $ \postoneg $ gadget 
             with cross-edges $ (1_\triangleright,1_\triangleleft)$, $(1_\triangleright,2_\triangleleft)$, $(2_\triangleright,1_\triangleleft)$, $(2_\triangleright,2_\triangleleft)$, $(2_\triangleright,3_\triangleleft) $ and $ (3_\triangleright,2_\triangleleft) $.
         }
        \label{fig:full-pos-to-neg_gadget}
    \end{subfigure}
    \hfill
    \begin{subfigure}[t]{0.3\textwidth}
        \centering
        \begin{tikzpicture}[scale=0.7]
         \node[hcdnode-big, label={[above=-0.5mm]:{\scriptsize $u$}}] (v1) at (0,1) {};
         \node[hcdnode-big, label={[above=-0.5mm]:{\scriptsize $ v $}}] (v2) at (2,1) {};
         \node[hcdnode-big, label={[below=2.5mm]:{\scriptsize $ w $}}] (v3) at (1,0) {};

         \draw[hcdedge] (v1) edge (v2);
         \draw[hcdedge] (v2) edge (v3);
         \draw[hcdedge] (v3) edge (v1);

         \end{tikzpicture}
         \caption{A graph $ G $ with Hamiltonian cycle $ u,v,w,u $.}
         \label{fig:full-pos-to-neg_cex}
    \end{subfigure}
    \hfill
    \begin{subfigure}[t]{0.3\textwidth}
        \centering
        \begin{tikzpicture}[scale=0.5]
         \node[hcunode-big, label={[left=0.5mm]:{\scriptsize $1_u$}}] (u1) at (0,1) {};
         \node[hcunode-big, label={[left=0.5mm]:{\scriptsize $ 2_u $}}] (u2) at (1,2) {};
         \node[hcunode-big, label={[left=0.5mm]:{\scriptsize $ 3_u $}}] (u3) at (2,3) {};
         \node[hcunode-big, label={[right=0.5mm]:{\scriptsize $1_v$}}] (v1) at (5,3) {};
         \node[hcunode-big, label={[right=0.5mm]:{\scriptsize $ 2_v $}}] (v2) at (6,2) {};
         \node[hcunode-big, label={[right=0.5mm]:{\scriptsize $ 3_v $}}] (v3) at (7,1) {};
         \node[hcunode-big, label={[below=2.5mm]:{\scriptsize $1_w$}}] (w1) at (2,-0.5) {};
         \node[hcunode-big, label={[below=2.5mm]:{\scriptsize $ 2_w $}}] (w2) at (3.5,-0.5) {};
         \node[hcunode-big, label={[below=2.5mm]:{\scriptsize $ 3_w $}}] (w3) at (5,-0.5) {};

         \draw[hcuedge] (u1) edge (u2);
         \draw[hcuedge] (u2) edge (u3);
         \draw[hcuedge] (v1) edge (v2);
         \draw[hcuedge] (v2) edge (v3);
         \draw[hcuedge] (w1) edge (w2);
         \draw[hcuedge] (w2) edge (w3);

         \draw[crossedge] (u1) edge (v1);
         \draw[crossedge] (u1) edge (v2);
         \draw[crossedge] (u2) edge (v1);
         \draw[crossedge] (u2) edge (v2);
         \draw[crossedge] (u2) edge (v3);
         \draw[crossedge] (u3) edge (v2);

         \draw[crossedge] (v1) edge (w1);
         \draw[crossedge] (v1) edge (w2);
         \draw[crossedge] (v2) edge (w1);
         \draw[crossedge] (v2) edge (w2);
         \draw[crossedge] (v2) edge (w3);
         \draw[crossedge] (v3) edge (w2);

         \draw[crossedge] (w1) edge (u1);
         \draw[crossedge] (w1) edge (u2);
         \draw[crossedge] (w2) edge (u1);
         \draw[crossedge] (w2) edge (u2);
         \draw[crossedge] (w2) edge (u3);
         \draw[crossedge] (w3) edge (u2);
         \end{tikzpicture}
         \caption{Graph $ G^\star $ without a Hamiltonian cycle resulting from applying the gadget to the example graph.}
         \label{fig:full-pos-to-neg_applied}
    \end{subfigure}\\
    \vspace{1em}
    \caption{An example which demonstrates that $ P_{(1_\triangleright,1_\triangleleft), (1_\triangleright,2_\triangleleft), (2_\triangleright,1_\triangleleft), (2_\triangleright,2_\triangleleft), (2_\triangleright,3_\triangleleft), (3_\triangleright,2_\triangleleft)} $ is a $ \postoneg $ gadget.}
    \label{fig:full-pos-to-neg}
\end{figure}
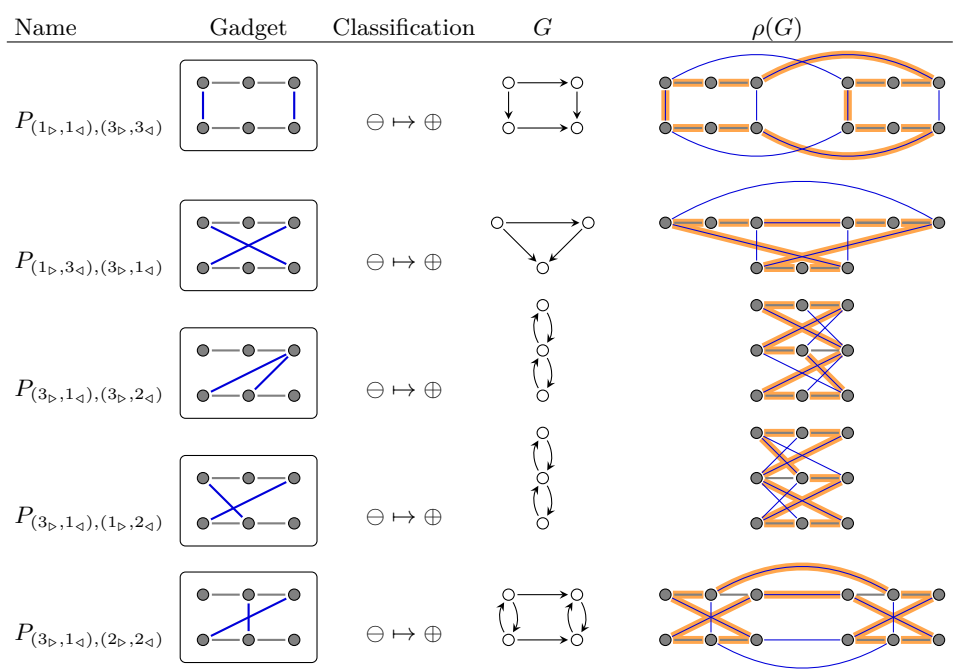
\begin{figure}[tp]
    \centering
    \def\arraystretch{3}
    \begin{tabular}{l c c c c}
        Name & Gadget & Classification & $ G $ & $ \rho(G) $\vspace{-1em}\\
        \hline
        $ P_{(1_\triangleright,1_\triangleleft), (3_\triangleright,3_\triangleleft)} $ &
        \begin{tikzpicture}[baseline, scale=0.6]
         \draw[rounded corners = 2pt] (-0.5,-0.5) rectangle (2.5,1.5);
         \node[targetnode-old] (c1) at (0,1) {};
         \node[targetnode-old] (c2) at (1,1) {};
         \node[targetnode-old] (c3) at (2,1) {};
         \node[targetnode-old] (d1) at (0,0) {};
         \node[targetnode-old] (d2) at (1,0) {};
         \node[targetnode-old] (d3) at (2,0) {};
         \draw[targetedge-old] (c1) edge (c2);
         \draw[targetedge-old] (c2) edge (c3);
         \draw[targetedge-old] (d1) edge (d2);
         \draw[targetedge-old] (d2) edge (d3);

         \draw[targetedge-new] (c1) edge (d1);
         \draw[targetedge-new] (c3) edge (d3);
         \end{tikzpicture}
                           & \negtopos &
        \begin{tikzpicture}[baseline, scale=0.6]
         \node[hcdnode] (v1) at (0,1) {};
         \node[hcdnode] (v2) at (0,0) {};
         \node[hcdnode] (v3) at (1.5,0) {};
         \node[hcdnode] (v4) at (1.5,1) {};

         \draw[hcdedge] (v1) edge (v2);
         \draw[hcdedge] (v2) edge (v3);
         \draw[hcdedge] (v4) edge (v3);
         \draw[hcdedge] (v1) edge (v4);
         \end{tikzpicture}
                           &
        \exampleAppliedSpace
        \begin{tikzpicture}[baseline, scale=0.6]
         \node[hcunode] (11) at (0,1) {};
         \node[hcunode] (12) at (1,1) {};
         \node[hcunode] (13) at (2,1) {};
         \node[hcunode] (21) at (0,0) {};
         \node[hcunode] (22) at (1,0) {};
         \node[hcunode] (23) at (2,0) {};
         \node[hcunode] (31) at (4,0) {};
         \node[hcunode] (32) at (5,0) {};
         \node[hcunode] (33) at (6,0) {};
         \node[hcunode] (41) at (4,1) {};
         \node[hcunode] (42) at (5,1) {};
         \node[hcunode] (43) at (6,1) {};

         \draw[witnessedge] (11) edge (12);
         \draw[witnessedge] (12) edge (13);
         \draw[witnessedge] (21) edge (22);
         \draw[witnessedge] (22) edge (23);
         \draw[witnessedge] (31) edge (32);
         \draw[witnessedge] (32) edge (33);
         \draw[witnessedge] (41) edge (42);
         \draw[witnessedge] (42) edge (43);
         \draw[witnessedge] (11) edge (21);
         \draw[witnessedge] (13) edge[bend left] (43);
         \draw[witnessedge] (23) edge[bend right] (33);
         \draw[witnessedge] (31) edge (41);

         \draw[hcuedge] (11) edge (12);
         \draw[hcuedge] (12) edge (13);
         \draw[hcuedge] (21) edge (22);
         \draw[hcuedge] (22) edge (23);
         \draw[hcuedge] (31) edge (32);
         \draw[hcuedge] (32) edge (33);
         \draw[hcuedge] (41) edge (42);
         \draw[hcuedge] (42) edge (43);

         \draw[crossedge] (11) edge (21);
         \draw[crossedge] (13) edge (23);
         \draw[crossedge] (11) edge[bend left] (41);
         \draw[crossedge] (13) edge[bend left] (43);

         \draw[crossedge] (21) edge[bend right] (31);
         \draw[crossedge] (23) edge[bend right] (33);
         \draw[crossedge] (31) edge (41);
         \draw[crossedge] (33) edge (43);
         \end{tikzpicture}\\
        $ P_{(1_\triangleright,3_\triangleleft), (3_\triangleright,1_\triangleleft)} $ &
        \begin{tikzpicture}[baseline, scale=0.6]
         \draw[rounded corners = 2pt] (-0.5,-0.5) rectangle (2.5,1.5);

         \node[targetnode-old] (c1) at (0,1) {};
         \node[targetnode-old] (c2) at (1,1) {};
         \node[targetnode-old] (c3) at (2,1) {};
         \node[targetnode-old] (d1) at (0,0) {};
         \node[targetnode-old] (d2) at (1,0) {};
         \node[targetnode-old] (d3) at (2,0) {};
         \draw[targetedge-old] (c1) edge (c2);
         \draw[targetedge-old] (c2) edge (c3);
         \draw[targetedge-old] (d1) edge (d2);
         \draw[targetedge-old] (d2) edge (d3);

         \draw[targetedge-new] (c1) edge (d3);
         \draw[targetedge-new] (c3) edge (d1);
         \end{tikzpicture}
                           & \negtopos &
        \begin{tikzpicture}[baseline, scale=0.6]
         \node[hcdnode] (v1) at (0,1) {};
         \node[hcdnode] (v2) at (2,1) {};
         \node[hcdnode] (v3) at (1,0) {};

         \draw[hcdedge] (v1) edge (v2);
         \draw[hcdedge] (v2) edge (v3);
         \draw[hcdedge] (v1) edge (v3);
         \end{tikzpicture}
                           &
        \exampleAppliedSpace
        \begin{tikzpicture}[baseline, scale=0.6]
         \node[hcunode] (11) at (0,1) {};
         \node[hcunode] (12) at (1,1) {};
         \node[hcunode] (13) at (2,1) {};
         \node[hcunode] (21) at (4,1) {};
         \node[hcunode] (22) at (5,1) {};
         \node[hcunode] (23) at (6,1) {};
         \node[hcunode] (31) at (2,0) {};
         \node[hcunode] (32) at (3,0) {};
         \node[hcunode] (33) at (4,0) {};

         \draw[witnessedge] (11) edge (12);
         \draw[witnessedge] (12) edge (13);
         \draw[witnessedge] (21) edge (22);
         \draw[witnessedge] (22) edge (23);
         \draw[witnessedge] (31) edge (32);
         \draw[witnessedge] (32) edge (33);
         \draw[witnessedge] (41) edge (42);
         \draw[witnessedge] (42) edge (43);
         \draw[witnessedge] (13) edge (21);
         \draw[witnessedge] (23) edge (31);
         \draw[witnessedge] (11) edge (33);

         \draw[hcuedge] (11) edge (12);
         \draw[hcuedge] (12) edge (13);
         \draw[hcuedge] (21) edge (22);
         \draw[hcuedge] (22) edge (23);
         \draw[hcuedge] (31) edge (32);
         \draw[hcuedge] (32) edge (33);
         \draw[hcuedge] (41) edge (42);
         \draw[hcuedge] (42) edge (43);

         \draw[crossedge] (11) edge[bend left] (23);
         \draw[crossedge] (13) edge (21);

         \draw[crossedge] (21) edge (33);
         \draw[crossedge] (23) edge (31);

         \draw[crossedge] (11) edge (33);
         \draw[crossedge] (13) edge (31);
         \end{tikzpicture}\\
        $ P_{(3_\triangleright,1_\triangleleft), (3_\triangleright,2_\triangleleft)} $ &
        \begin{tikzpicture}[baseline, scale=0.6]
         \draw[rounded corners = 2pt] (-0.5,-0.5) rectangle (2.5,1.5);

         \node[targetnode-old] (c1) at (0,1) {};
         \node[targetnode-old] (c2) at (1,1) {};
         \node[targetnode-old] (c3) at (2,1) {};
         \node[targetnode-old] (d1) at (0,0) {};
         \node[targetnode-old] (d2) at (1,0) {};
         \node[targetnode-old] (d3) at (2,0) {};
         \draw[targetedge-old] (c1) edge (c2);
         \draw[targetedge-old] (c2) edge (c3);
         \draw[targetedge-old] (d1) edge (d2);
         \draw[targetedge-old] (d2) edge (d3);

         \draw[targetedge-new] (c3) edge (d1);
         \draw[targetedge-new] (c3) edge (d2);
         \end{tikzpicture}
                           & \negtopos &
        \begin{tikzpicture}[baseline, scale=0.6]
         \node[hcdnode] (v1) at (0,2) {};
         \node[hcdnode] (v2) at (0,1) {};
         \node[hcdnode] (v3) at (0,0) {};

         \draw[hcdedge] (v1) edge[bend left] (v2);
         \draw[hcdedge] (v2) edge[bend left] (v1);
         \draw[hcdedge] (v3) edge[bend left] (v2);
         \draw[hcdedge] (v2) edge[bend left] (v3);
         \end{tikzpicture}
                           &
        \exampleAppliedSpace
        \begin{tikzpicture}[baseline, scale=0.6]
         \node[hcunode] (11) at (0,2) {};
         \node[hcunode] (12) at (1,2) {};
         \node[hcunode] (13) at (2,2) {};
         \node[hcunode] (21) at (0,1) {};
         \node[hcunode] (22) at (1,1) {};
         \node[hcunode] (23) at (2,1) {};
         \node[hcunode] (31) at (0,0) {};
         \node[hcunode] (32) at (1,0) {};
         \node[hcunode] (33) at (2,0) {};

         \draw[witnessedge] (11) edge (12);
         \draw[witnessedge] (12) edge (13);
         \draw[witnessedge] (21) edge (22);
         \draw[witnessedge] (31) edge (32);
         \draw[witnessedge] (32) edge (33);
         \draw[witnessedge] (23) edge (31);
         \draw[witnessedge] (33) edge (22);
         \draw[witnessedge] (13) edge (21);
         \draw[witnessedge] (23) edge (11);

         \draw[hcuedge] (11) edge (12);
         \draw[hcuedge] (12) edge (13);
         \draw[hcuedge] (21) edge (22);
         \draw[hcuedge] (22) edge (23);
         \draw[hcuedge] (31) edge (32);
         \draw[hcuedge] (32) edge (33);

         \draw[crossedge] (13) edge (21);
         \draw[crossedge] (13) edge (22);

         \draw[crossedge] (23) edge (11);
         \draw[crossedge] (23) edge (12);

         \draw[crossedge] (23) edge (31);
         \draw[crossedge] (23) edge (32);

         \draw[crossedge] (33) edge (21);
         \draw[crossedge] (33) edge (22);
         \end{tikzpicture}\\
        $ P_{(3_\triangleright,1_\triangleleft), (1_\triangleright,2_\triangleleft)} $ &
        \begin{tikzpicture}[baseline, scale=0.6]
         \draw[rounded corners = 2pt] (-0.5,-0.5) rectangle (2.5,1.5);

         \node[targetnode-old] (c1) at (0,1) {};
         \node[targetnode-old] (c2) at (1,1) {};
         \node[targetnode-old] (c3) at (2,1) {};
         \node[targetnode-old] (d1) at (0,0) {};
         \node[targetnode-old] (d2) at (1,0) {};
         \node[targetnode-old] (d3) at (2,0) {};
         \draw[targetedge-old] (c1) edge (c2);
         \draw[targetedge-old] (c2) edge (c3);
         \draw[targetedge-old] (d1) edge (d2);
         \draw[targetedge-old] (d2) edge (d3);

         \draw[targetedge-new] (c3) edge (d1);
         \draw[targetedge-new] (c1) edge (d2);
         \end{tikzpicture}
                           & \negtopos &
        \begin{tikzpicture}[baseline, scale=0.6]
         \node[hcdnode] (v1) at (0,2) {};
         \node[hcdnode] (v2) at (0,1) {};
         \node[hcdnode] (v3) at (0,0) {};

         \draw[hcdedge] (v1) edge[bend left] (v2);
         \draw[hcdedge] (v2) edge[bend left] (v1);
         \draw[hcdedge] (v3) edge[bend left] (v2);
         \draw[hcdedge] (v2) edge[bend left] (v3);
         \end{tikzpicture}
                           &
        \exampleAppliedSpace
        \begin{tikzpicture}[baseline, scale=0.6]
         \node[hcunode] (11) at (0,2) {};
         \node[hcunode] (12) at (1,2) {};
         \node[hcunode] (13) at (2,2) {};
         \node[hcunode] (21) at (0,1) {};
         \node[hcunode] (22) at (1,1) {};
         \node[hcunode] (23) at (2,1) {};
         \node[hcunode] (31) at (0,0) {};
         \node[hcunode] (32) at (1,0) {};
         \node[hcunode] (33) at (2,0) {};

         \draw[witnessedge] (11) edge (12);
         \draw[witnessedge] (12) edge (13);
         \draw[witnessedge] (22) edge (23);
         \draw[witnessedge] (31) edge (32);
         \draw[witnessedge] (32) edge (33);
         \draw[witnessedge] (13) edge (21);
         \draw[witnessedge] (11) edge (22);
         \draw[witnessedge] (23) edge (31);
         \draw[witnessedge] (33) edge (21);

         \draw[hcuedge] (11) edge (12);
         \draw[hcuedge] (12) edge (13);
         \draw[hcuedge] (21) edge (22);
         \draw[hcuedge] (22) edge (23);
         \draw[hcuedge] (31) edge (32);
         \draw[hcuedge] (32) edge (33);

         \draw[crossedge] (13) edge (21);
         \draw[crossedge] (11) edge (22);

         \draw[crossedge] (23) edge (11);
         \draw[crossedge] (21) edge (12);

         \draw[crossedge] (23) edge (31);
         \draw[crossedge] (21) edge (32);

         \draw[crossedge] (33) edge (21);
         \draw[crossedge] (31) edge (22);
         \end{tikzpicture}\\
        $ P_{(3_\triangleright,1_\triangleleft), (2_\triangleright,2_\triangleleft)} $ &
        \begin{tikzpicture}[baseline, scale=0.6]
         \draw[rounded corners = 2pt] (-0.5,-0.5) rectangle (2.5,1.5);

         \node[targetnode-old] (c1) at (0,1) {};
         \node[targetnode-old] (c2) at (1,1) {};
         \node[targetnode-old] (c3) at (2,1) {};
         \node[targetnode-old] (d1) at (0,0) {};
         \node[targetnode-old] (d2) at (1,0) {};
         \node[targetnode-old] (d3) at (2,0) {};
         \draw[targetedge-old] (c1) edge (c2);
         \draw[targetedge-old] (c2) edge (c3);
         \draw[targetedge-old] (d1) edge (d2);
         \draw[targetedge-old] (d2) edge (d3);

         \draw[targetedge-new] (c3) edge (d1);
         \draw[targetedge-new] (c2) edge (d2);
         \end{tikzpicture}
                           & \negtopos &
        \begin{tikzpicture}[baseline, scale=0.6]
         \node[hcdnode] (v1) at (0,1) {};
         \node[hcdnode] (v2) at (0,0) {};
         \node[hcdnode] (v3) at (1.5,0) {};
         \node[hcdnode] (v4) at (1.5,1) {};

         \draw[hcdedge] (v1) edge[bend left] (v2);
         \draw[hcdedge] (v2) edge[bend left] (v1);
         \draw[hcdedge] (v2) edge (v3);
         \draw[hcdedge] (v4) edge[bend left] (v3);
         \draw[hcdedge] (v3) edge[bend left] (v4);
         \draw[hcdedge] (v1) edge (v4);
         \end{tikzpicture}
                           &
        \exampleAppliedSpace
        \begin{tikzpicture}[baseline, scale=0.6]
         \node[hcunode] (11) at (0,1) {};
         \node[hcunode] (12) at (1,1) {};
         \node[hcunode] (13) at (2,1) {};
         \node[hcunode] (21) at (0,0) {};
         \node[hcunode] (22) at (1,0) {};
         \node[hcunode] (23) at (2,0) {};
         \node[hcunode] (31) at (4,0) {};
         \node[hcunode] (32) at (5,0) {};
         \node[hcunode] (33) at (6,0) {};
         \node[hcunode] (41) at (4,1) {};
         \node[hcunode] (42) at (5,1) {};
         \node[hcunode] (43) at (6,1) {};

         \draw[witnessedge] (11) edge (12);
         \draw[witnessedge] (21) edge (22);
         \draw[witnessedge] (22) edge (23);
         \draw[witnessedge] (31) edge (32);
         \draw[witnessedge] (32) edge (33);
         \draw[witnessedge] (42) edge (43);
         \draw[witnessedge] (13) edge (21);
         \draw[witnessedge] (12) edge[bend left] (42);
         \draw[witnessedge] (13) edge (41);
         \draw[witnessedge] (23) edge (11);
         \draw[witnessedge] (33) edge (41);
         \draw[witnessedge] (43) edge (31);

         \draw[hcuedge] (11) edge (12);
         \draw[hcuedge] (12) edge (13);
         \draw[hcuedge] (21) edge (22);
         \draw[hcuedge] (22) edge (23);
         \draw[hcuedge] (31) edge (32);
         \draw[hcuedge] (32) edge (33);
         \draw[hcuedge] (41) edge (42);
         \draw[hcuedge] (42) edge (43);

         \draw[crossedge] (13) edge (21);
         \draw[crossedge] (12) edge (22);
         \draw[crossedge] (12) edge[bend left] (42);
         \draw[crossedge] (13) edge (41);

         \draw[crossedge] (23) edge (11);
         \draw[crossedge] (23) edge (31);
         \draw[crossedge] (22) edge[bend right] (32);

         \draw[crossedge] (33) edge (41);
         \draw[crossedge] (43) edge (31);
         \draw[crossedge] (32) edge (42);
         \end{tikzpicture}
    \end{tabular}
    \caption{Gadgets for which the node graphs are paths with three nodes that do not induce a valid node gadget reduction from \HCd to \HCu. Since they are all $ \negtopos $ gadgets, a Hamiltonian cycle is indicated by edges highlighted in orange in each of the applied graphs.}
    \label{fig:invalid-hc-gadgets}
\end{figure}
\begin{figure}[tp]
    \centering
    \def\arraystretch{3}
    \begin{tabular}{l c c c c}
        Name & Gadget & Classification & $ G $ & $ \hspace{2.5em}\rho(G) $\vspace{-1em}\\
        \hline
        $ I_{(1_\triangleright,1_\triangleleft),(2_\triangleright,1_\triangleleft), (2_\triangleright,2_\triangleleft),(3_\triangleright,2_\triangleleft) (3_\triangleright,3_\triangleleft)} $ &
        \begin{tikzpicture}[baseline,scale=0.5]
         \draw[rounded corners = 2pt] (-0.5,-0.5) rectangle (2.5,1.5);

         \node[targetnode-old] (c1) at (0,1) {};
         \node[targetnode-old] (c2) at (1,1) {};
         \node[targetnode-old] (c3) at (2,1) {};
         \node[targetnode-old] (d1) at (0,0) {};
         \node[targetnode-old] (d2) at (1,0) {};
         \node[targetnode-old] (d3) at (2,0) {};

         \draw[targetedge-new] (c1) edge (d1);
         \draw[targetedge-new] (c2) edge (d1);
         \draw[targetedge-new] (c2) edge (d2);
         \draw[targetedge-new] (c3) edge (d2);
         \draw[targetedge-new] (c3) edge (d3);
         \end{tikzpicture}
                           & \negtopos &
        \begin{tikzpicture}[baseline, scale=0.5]
         \node[hcdnode] (v1) at (0,1) {};
         \node[hcdnode] (v2) at (0,0) {};
         \node[hcdnode] (v3) at (1.5,0) {};
         \node[hcdnode] (v4) at (1.5,1) {};

         \draw[hcdedge] (v1) edge (v2);
         \draw[hcdedge] (v2) edge (v3);
         \draw[hcdedge] (v4) edge (v3);
         \draw[hcdedge] (v1) edge (v4);
         \end{tikzpicture}
                           &
        \exampleAppliedSpace
        \begin{tikzpicture}[baseline, scale=0.5]
         \node[hcunode] (11) at (0,1) {};
         \node[hcunode] (12) at (1,1) {};
         \node[hcunode] (13) at (2,1) {};
         \node[hcunode] (21) at (0,0) {};
         \node[hcunode] (22) at (1,0) {};
         \node[hcunode] (23) at (2,0) {};
         \node[hcunode] (31) at (4,0) {};
         \node[hcunode] (32) at (5,0) {};
         \node[hcunode] (33) at (6,0) {};
         \node[hcunode] (41) at (4,1) {};
         \node[hcunode] (42) at (5,1) {};
         \node[hcunode] (43) at (6,1) {};

         \draw[witnessedge] (11) edge (21);
         \draw[witnessedge] (12) edge (21);
         \draw[witnessedge] (12) edge (22);
         \draw[witnessedge] (13) edge (22);
         \draw[witnessedge] (13) edge (23);
         \draw[witnessedge] (23) edge[bend right] (33);
         \draw[witnessedge] (11) edge[bend left] (41);
         \draw[witnessedge] (41) edge (31);
         \draw[witnessedge] (42) edge (31);
         \draw[witnessedge] (42) edge (32);
         \draw[witnessedge] (43) edge (32);
         \draw[witnessedge] (43) edge (33);

         \draw[crossedge] (11) edge (21);
         \draw[crossedge] (12) edge (21);
         \draw[crossedge] (12) edge (22);
         \draw[crossedge] (13) edge (22);
         \draw[crossedge] (13) edge (23);
         \draw[crossedge] (21) edge[bend right] (31);
         \draw[crossedge] (22) edge[bend right] (31);
         \draw[crossedge] (22) edge[bend right] (32);
         \draw[crossedge] (23) edge[bend right] (32);
         \draw[crossedge] (23) edge[bend right] (33);
         \draw[crossedge] (11) edge[bend left] (41);
         \draw[crossedge] (12) edge[bend left] (41);
         \draw[crossedge] (12) edge[bend left] (42);
         \draw[crossedge] (13) edge[bend left] (42);
         \draw[crossedge] (13) edge[bend left] (43);
         \draw[crossedge] (41) edge (31);
         \draw[crossedge] (42) edge (31);
         \draw[crossedge] (42) edge (32);
         \draw[crossedge] (43) edge (32);
         \draw[crossedge] (43) edge (33);
         \end{tikzpicture}\\
        $ I_{(1_\triangleright,1_\triangleleft),(2_\triangleright,1_\triangleleft),(2_\triangleright,2_\triangleleft),(3_\triangleright,3_\triangleleft)} $ &
        \begin{tikzpicture}[baseline, scale=0.5]
         \draw[rounded corners = 2pt] (-0.5,-0.5) rectangle (2.5,1.5);

         \node[targetnode-old] (c1) at (0,1) {};
         \node[targetnode-old] (c2) at (1,1) {};
         \node[targetnode-old] (c3) at (2,1) {};
         \node[targetnode-old] (d1) at (0,0) {};
         \node[targetnode-old] (d2) at (1,0) {};
         \node[targetnode-old] (d3) at (2,0) {};

         \draw[targetedge-new] (c1) edge (d1);
         \draw[targetedge-new] (c2) edge (d1);
         \draw[targetedge-new] (c2) edge (d2);
         \draw[targetedge-new] (c3) edge (d3);
         \end{tikzpicture}
                           & \postoneg &
        \begin{tikzpicture}[baseline, scale=0.5]
         \node[hcdnode] (v1) at (0,1) {};
         \node[hcdnode] (v2) at (2,1) {};
         \node[hcdnode] (v3) at (1,0) {};

         \draw[hcdedge] (v1) edge (v2);
         \draw[hcdedge] (v2) edge (v3);
         \draw[hcdedge] (v3) edge (v1);
         \end{tikzpicture}
                           &
        \exampleAppliedSpace
        \begin{tikzpicture}[baseline, scale=0.5]
         \node[hcunode] (11) at (0,1) {};
         \node[hcunode] (12) at (1,1) {};
         \node[hcunode] (13) at (2,1) {};
         \node[hcunode] (21) at (4,1) {};
         \node[hcunode] (22) at (5,1) {};
         \node[hcunode] (23) at (6,1) {};
         \node[hcunode] (31) at (2,0) {};
         \node[hcunode] (32) at (3,0) {};
         \node[hcunode] (33) at (4,0) {};

         \draw[witnessedge] (13) edge[bend left] (23);
         \draw[witnessedge] (23) edge (33);
         \draw[witnessedge] (33) edge (13);

         \draw[crossedge] (11) edge[bend left] (21);
         \draw[crossedge] (12) edge[bend left] (21);
         \draw[crossedge] (12) edge[bend left] (22);
         \draw[crossedge] (13) edge[bend left] (23);

         \draw[crossedge] (21) edge (31);
         \draw[crossedge] (22) edge (31);
         \draw[crossedge] (22) edge (32);
         \draw[crossedge] (23) edge (33);

         \draw[crossedge] (31) edge (11);
         \draw[crossedge] (32) edge (12);
         \draw[crossedge] (32) edge (11);
         \draw[crossedge] (33) edge (13);
         \end{tikzpicture}\\
        $ I_{(1_\triangleright,1_\triangleleft), (2_\triangleright,1_\triangleleft),(3_\triangleright,2_\triangleleft),(3_\triangleright,3_\triangleleft)} $ &
        \begin{tikzpicture}[baseline, scale=0.5]
         \draw[rounded corners = 2pt] (-0.5,-0.5) rectangle (2.5,1.5);

         \node[targetnode-old] (c1) at (0,1) {};
         \node[targetnode-old] (c2) at (1,1) {};
         \node[targetnode-old] (c3) at (2,1) {};
         \node[targetnode-old] (d1) at (0,0) {};
         \node[targetnode-old] (d2) at (1,0) {};
         \node[targetnode-old] (d3) at (2,0) {};

         \draw[targetedge-new] (c1) edge (d1);
         \draw[targetedge-new] (c2) edge (d1);
         \draw[targetedge-new] (c3) edge (d2);
         \draw[targetedge-new] (c3) edge (d3);
         \end{tikzpicture}
                           & \postoneg &
        \begin{tikzpicture}[baseline, scale=0.5]
         \node[hcdnode] (v1) at (0,1) {};
         \node[hcdnode] (v2) at (2,1) {};
         \node[hcdnode] (v3) at (1,0) {};

         \draw[hcdedge] (v1) edge (v2);
         \draw[hcdedge] (v2) edge (v3);
         \draw[hcdedge] (v3) edge (v1);
         \end{tikzpicture}
                           &
       \exampleAppliedSpace
        \begin{tikzpicture}[baseline, scale=0.5]
         \node[hcunode] (11) at (0,1) {};
         \node[hcunode] (12) at (1,1) {};
         \node[hcunode] (13) at (2,1) {};
         \node[hcunode] (21) at (4,1) {};
         \node[hcunode] (22) at (5,1) {};
         \node[hcunode] (23) at (6,1) {};
         \node[hcunode] (31) at (2,0) {};
         \node[hcunode] (32) at (3,0) {};
         \node[hcunode] (33) at (4,0) {};

         \draw[witnessedge] (12) edge[bend left] (21);
         \draw[witnessedge] (13) edge[bend left] (22);
         \draw[witnessedge] (22) edge (31);
         \draw[witnessedge] (23) edge (32);
         \draw[witnessedge] (32) edge (11);
         \draw[witnessedge] (33) edge (12);

         \draw[crossedge] (11) edge[bend left] (21);
         \draw[crossedge] (12) edge[bend left] (21);
         \draw[crossedge] (13) edge[bend left] (22);
         \draw[crossedge] (13) edge[bend left] (23);

         \draw[crossedge] (21) edge (31);
         \draw[crossedge] (22) edge (31);
         \draw[crossedge] (23) edge (32);
         \draw[crossedge] (23) edge (33);

         \draw[crossedge] (31) edge (11);
         \draw[crossedge] (32) edge (11);
         \draw[crossedge] (33) edge (12);
         \draw[crossedge] (33) edge (13);
         \end{tikzpicture}\\
        $ N_{(1_\triangleright,1_\triangleleft),(2_\triangleright,1_\triangleleft),(2_\triangleright,2_\triangleleft),(3_\triangleright,3_\triangleleft)} $ &
        \begin{tikzpicture}[baseline, scale=0.5]
         \draw[rounded corners = 2pt] (-0.5,-0.5) rectangle (2.5,1.5);

         \node[targetnode-old] (c1) at (0,1) {};
         \node[targetnode-old] (c2) at (1,1) {};
         \node[targetnode-old] (c3) at (2,1) {};
         \node[targetnode-old] (d1) at (0,0) {};
         \node[targetnode-old] (d2) at (1,0) {};
         \node[targetnode-old] (d3) at (2,0) {};
         \draw[targetedge-old] (c1) edge (c2);
         \draw[targetedge-old] (d1) edge (d2);

         \draw[targetedge-new] (c1) edge (d1);
         \draw[targetedge-new] (c2) edge (d1);
         \draw[targetedge-new] (c2) edge (d2);
         \draw[targetedge-new] (c3) edge (d3);
         \end{tikzpicture}
                           & \postoneg &
        \begin{tikzpicture}[baseline, scale=0.5]
         \node[hcdnode] (v1) at (0,1) {};
         \node[hcdnode] (v2) at (2,1) {};
         \node[hcdnode] (v3) at (1,0) {};

         \draw[hcdedge] (v1) edge (v2);
         \draw[hcdedge] (v2) edge (v3);
         \draw[hcdedge] (v3) edge (v1);
         \end{tikzpicture}
                           &
       \exampleAppliedSpace
        \begin{tikzpicture}[baseline, scale=0.5]
         \node[hcunode] (11) at (0,1) {};
         \node[hcunode] (12) at (1,1) {};
         \node[hcunode] (13) at (2,1) {};
         \node[hcunode] (21) at (4,1) {};
         \node[hcunode] (22) at (5,1) {};
         \node[hcunode] (23) at (6,1) {};
         \node[hcunode] (31) at (2,0) {};
         \node[hcunode] (32) at (3,0) {};
         \node[hcunode] (33) at (4,0) {};

         \draw[witnessedge] (13) edge[bend left] (23);
         \draw[witnessedge] (23) edge (33);
         \draw[witnessedge] (33) edge (13);

         \draw[hcuedge] (11) edge (12);
         \draw[hcuedge] (21) edge (22);
         \draw[hcuedge] (31) edge (32);

         \draw[crossedge] (11) edge[bend left] (21);
         \draw[crossedge] (12) edge[bend left] (21);
         \draw[crossedge] (12) edge[bend left] (22);
         \draw[crossedge] (13) edge[bend left] (23);

         \draw[crossedge] (21) edge (31);
         \draw[crossedge] (22) edge (31);
         \draw[crossedge] (22) edge (32);
         \draw[crossedge] (23) edge (33);

         \draw[crossedge] (31) edge (11);
         \draw[crossedge] (32) edge (11);
         \draw[crossedge] (32) edge (12);
         \draw[crossedge] (33) edge (13);
         \end{tikzpicture}\\
        $ N_{(1_\triangleright,1_\triangleleft),(2_\triangleright,1_\triangleleft),(3_\triangleright,2_\triangleleft),(3_\triangleright,3_\triangleleft)} $ &
        \begin{tikzpicture}[baseline, scale=0.5]
         \draw[rounded corners = 2pt] (-0.5,-0.5) rectangle (2.5,1.5);

         \node[targetnode-old] (c1) at (0,1) {};
         \node[targetnode-old] (c2) at (1,1) {};
         \node[targetnode-old] (c3) at (2,1) {};
         \node[targetnode-old] (d1) at (0,0) {};
         \node[targetnode-old] (d2) at (1,0) {};
         \node[targetnode-old] (d3) at (2,0) {};
         \draw[targetedge-old] (c1) edge (c2);
         \draw[targetedge-old] (d1) edge (d2);

         \draw[targetedge-new] (c1) edge (d1);
         \draw[targetedge-new] (c2) edge (d1);
         \draw[targetedge-new] (c3) edge (d2);
         \draw[targetedge-new] (c3) edge (d3);
         \end{tikzpicture}
                           & \negtopos &
        \begin{tikzpicture}[baseline, scale=0.5]
         \node[hcdnode] (v1) at (0,1) {};
         \node[hcdnode] (v2) at (0,0) {};
         \node[hcdnode] (v3) at (1.5,0) {};
         \node[hcdnode] (v4) at (1.5,1) {};

         \draw[hcdedge] (v1) edge (v2);
         \draw[hcdedge] (v2) edge (v3);
         \draw[hcdedge] (v4) edge (v3);
         \draw[hcdedge] (v1) edge (v4);
         \end{tikzpicture}
                           &
       \exampleAppliedSpace
        \begin{tikzpicture}[baseline, scale=0.5]
         \node[hcunode] (11) at (0,1) {};
         \node[hcunode] (12) at (1,1) {};
         \node[hcunode] (13) at (2,1) {};
         \node[hcunode] (21) at (0,0) {};
         \node[hcunode] (22) at (1,0) {};
         \node[hcunode] (23) at (2,0) {};
         \node[hcunode] (31) at (4,0) {};
         \node[hcunode] (32) at (5,0) {};
         \node[hcunode] (33) at (6,0) {};
         \node[hcunode] (41) at (4,1) {};
         \node[hcunode] (42) at (5,1) {};
         \node[hcunode] (43) at (6,1) {};

         \draw[witnessedge] (11) edge (12);
         \draw[witnessedge] (21) edge (22);
         \draw[witnessedge] (31) edge (32);
         \draw[witnessedge] (41) edge (42);
         \draw[witnessedge] (11) edge (21);
         \draw[witnessedge] (13) edge (22);
         \draw[witnessedge] (13) edge (23);
         \draw[witnessedge] (42) edge (31);
         \draw[witnessedge] (43) edge (32);
         \draw[witnessedge] (43) edge (33);
         \draw[witnessedge] (12) edge[bend left] (41);
         \draw[witnessedge] (23) edge[bend right] (33);

         \draw[hcuedge] (11) edge (12);
         \draw[hcuedge] (21) edge (22);
         \draw[hcuedge] (31) edge (32);
         \draw[hcuedge] (41) edge (42);

         \draw[crossedge] (11) edge (21);
         \draw[crossedge] (12) edge (21);
         \draw[crossedge] (13) edge (22);
         \draw[crossedge] (13) edge (23);

         \draw[crossedge] (41) edge (31);
         \draw[crossedge] (42) edge (31);
         \draw[crossedge] (43) edge (32);
         \draw[crossedge] (43) edge (33);

         \draw[crossedge] (11) edge[bend left] (41);
         \draw[crossedge] (12) edge[bend left] (41);
         \draw[crossedge] (13) edge[bend left] (42);
         \draw[crossedge] (13) edge[bend left] (43);

         \draw[crossedge] (21) edge[bend right] (31);
         \draw[crossedge] (22) edge[bend right] (31);
         \draw[crossedge] (23) edge[bend right] (32);
         \draw[crossedge] (23) edge[bend right] (33);
         \end{tikzpicture}\\
        $ N_{(1_\triangleright,1_\triangleleft),(2_\triangleright,2_\triangleleft),(3_\triangleright,2_\triangleleft),(3_\triangleright,3_\triangleleft)} $ &
        \begin{tikzpicture}[baseline, scale=0.5]
         \draw[rounded corners = 2pt] (-0.5,-0.5) rectangle (2.5,1.5);

         \node[targetnode-old] (c1) at (0,1) {};
         \node[targetnode-old] (c2) at (1,1) {};
         \node[targetnode-old] (c3) at (2,1) {};
         \node[targetnode-old] (d1) at (0,0) {};
         \node[targetnode-old] (d2) at (1,0) {};
         \node[targetnode-old] (d3) at (2,0) {};
         \draw[targetedge-old] (c1) edge (c2);
         \draw[targetedge-old] (d1) edge (d2);

         \draw[targetedge-new] (c1) edge (d1);
         \draw[targetedge-new] (c2) edge (d2);
         \draw[targetedge-new] (c3) edge (d2);
         \draw[targetedge-new] (c3) edge (d3);
         \end{tikzpicture}
                           & \negtopos &
        \begin{tikzpicture}[baseline, scale=0.5]
         \node[hcdnode] (v1) at (0,1) {};
         \node[hcdnode] (v2) at (0,0) {};
         \node[hcdnode] (v3) at (1.5,0) {};
         \node[hcdnode] (v4) at (1.5,1) {};

         \draw[hcdedge] (v1) edge (v2);
         \draw[hcdedge] (v2) edge (v3);
         \draw[hcdedge] (v4) edge (v3);
         \draw[hcdedge] (v1) edge (v4);
         \end{tikzpicture}
                           &
       \exampleAppliedSpace
        \begin{tikzpicture}[baseline, scale=0.5]
         \node[hcunode] (11) at (0,1) {};
         \node[hcunode] (12) at (1,1) {};
         \node[hcunode] (13) at (2,1) {};
         \node[hcunode] (21) at (0,0) {};
         \node[hcunode] (22) at (1,0) {};
         \node[hcunode] (23) at (2,0) {};
         \node[hcunode] (31) at (4,0) {};
         \node[hcunode] (32) at (5,0) {};
         \node[hcunode] (33) at (6,0) {};
         \node[hcunode] (41) at (4,1) {};
         \node[hcunode] (42) at (5,1) {};
         \node[hcunode] (43) at (6,1) {};

         \draw[witnessedge] (11) edge (12);
         \draw[witnessedge] (41) edge (42);
         \draw[witnessedge] (11) edge (21);
         \draw[witnessedge] (12) edge (22);
         \draw[witnessedge] (13) edge (22);
         \draw[witnessedge] (13) edge (23);
         \draw[witnessedge] (41) edge (31);
         \draw[witnessedge] (42) edge (32);
         \draw[witnessedge] (43) edge (32);
         \draw[witnessedge] (43) edge (33);
         \draw[witnessedge] (21) edge[bend right] (31);
         \draw[witnessedge] (23) edge[bend right] (33);

         \draw[hcuedge] (11) edge (12);
         \draw[hcuedge] (21) edge (22);
         \draw[hcuedge] (31) edge (32);
         \draw[hcuedge] (41) edge (42);

         \draw[crossedge] (11) edge (21);
         \draw[crossedge] (12) edge (22);
         \draw[crossedge] (13) edge (22);
         \draw[crossedge] (13) edge (23);

         \draw[crossedge] (41) edge (31);
         \draw[crossedge] (42) edge (32);
         \draw[crossedge] (43) edge (32);
         \draw[crossedge] (43) edge (33);

         \draw[crossedge] (11) edge[bend left] (41);
         \draw[crossedge] (12) edge[bend left] (42);
         \draw[crossedge] (13) edge[bend left] (42);
         \draw[crossedge] (13) edge[bend left] (43);

         \draw[crossedge] (21) edge[bend right] (31);
         \draw[crossedge] (22) edge[bend right] (32);
         \draw[crossedge] (23) edge[bend right] (32);
         \draw[crossedge] (23) edge[bend right] (33);
         \end{tikzpicture}\\
        $ N_{(1_\triangleright,1_\triangleleft),(3_\triangleright,2_\triangleleft),(3_\triangleright,3_\triangleleft)} $ &
        \begin{tikzpicture}[baseline, scale=0.5]
         \draw[rounded corners = 2pt] (-0.5,-0.5) rectangle (2.5,1.5);

         \node[targetnode-old] (c1) at (0,1) {};
         \node[targetnode-old] (c2) at (1,1) {};
         \node[targetnode-old] (c3) at (2,1) {};
         \node[targetnode-old] (d1) at (0,0) {};
         \node[targetnode-old] (d2) at (1,0) {};
         \node[targetnode-old] (d3) at (2,0) {};
         \draw[targetedge-old] (c1) edge (c2);
         \draw[targetedge-old] (d1) edge (d2);

         \draw[targetedge-new] (c1) edge (d1);
         \draw[targetedge-new] (c3) edge (d2);
         \draw[targetedge-new] (c3) edge (d3);
         \end{tikzpicture}
                           & \postoneg &
        \begin{tikzpicture}[baseline, scale=0.5]
         \node[hcdnode] (v1) at (0,1) {};
         \node[hcdnode] (v2) at (2,1) {};
         \node[hcdnode] (v3) at (1,0) {};

         \draw[hcdedge] (v1) edge (v2);
         \draw[hcdedge] (v2) edge (v3);
         \draw[hcdedge] (v3) edge (v1);
         \end{tikzpicture}
                           &
       \exampleAppliedSpace
        \begin{tikzpicture}[baseline, scale=0.5]
         \node[hcunode] (11) at (0,1) {};
         \node[hcunode] (12) at (1,1) {};
         \node[hcunode] (13) at (2,1) {};
         \node[hcunode] (21) at (4,1) {};
         \node[hcunode] (22) at (5,1) {};
         \node[hcunode] (23) at (6,1) {};
         \node[hcunode] (31) at (2,0) {};
         \node[hcunode] (32) at (3,0) {};
         \node[hcunode] (33) at (4,0) {};

         \draw[witnessedge] (11) edge (12);
         \draw[witnessedge] (21) edge (22);
         \draw[witnessedge] (31) edge (32);
         \draw[witnessedge] (13) edge[bend left] (22);
         \draw[witnessedge] (23) edge (32);
         \draw[witnessedge] (33) edge (12);

         \draw[vgedge] (11) edge (12);
         \draw[vgedge] (21) edge (22);
         \draw[vgedge] (31) edge (32);

         \draw[crossedge] (11) edge[bend left] (21);
         \draw[crossedge] (13) edge[bend left] (22);
         \draw[crossedge] (13) edge[bend left] (23);

         \draw[crossedge] (21) edge (31);
         \draw[crossedge] (23) edge (32);
         \draw[crossedge] (23) edge (33);

         \draw[crossedge] (31) edge (11);
         \draw[crossedge] (33) edge (12);
         \draw[crossedge] (33) edge (13);
         \end{tikzpicture}\\
    \end{tabular}

    \caption{Invalid gadgets for which the node graphs have three nodes but are not paths. As for the $ \negtopos $ gadgets in Figure \ref{fig:invalid-hc-gadgets}, a Hamiltonian cycle is indicated by edges highlighted in orange in the applied graphs of $ \negtopos $ gadgets. If there are such edges in the applied graph of a $ \postoneg $ gadget, these would have to be used on any Hamiltonian cycle which results in a contradiction to the existence of such a cycle.}
    \label{fig:invalid-nonpath-hc-gadgets}
\end{figure}
\begin{figure}[tp]
    \centering
    \begin{subfigure}[t]{0.3\textwidth}
        \centering
        \begin{tikzpicture}
         \draw[rounded corners = 2pt] (-0.5,-0.55) rectangle (2.5,1.55);
         \node[targetnode-old-big, label={[above=-0.5mm]:{\scriptsize $ 1_\triangleright $}}] (c1) at (0,1) {};
         \node[targetnode-old-big, label={[above=-0.5mm]:{\scriptsize $ 2_\triangleright $}}] (c2) at (1,1) {};
         \node[targetnode-old-big, label={[above=-0.5mm]:{\scriptsize $ 3_\triangleright $}}] (c3) at (2,1) {};
         \node[targetnode-old-big, label={[below=2.5mm]:{\scriptsize $ 1_\triangleleft $}}] (d1) at (0,0) {};
         \node[targetnode-old-big, label={[below=2.5mm]:{\scriptsize $ 2_\triangleleft $}}] (d2) at (1,0) {};
         \node[targetnode-old-big, label={[below=2.5mm]:{\scriptsize $ 3_\triangleleft $}}] (d3) at (2,0) {};
         \draw[targetedge-old] (c1) edge (c2);
         \draw[targetedge-old] (d1) edge (d2);

         \draw[targetedge-new] (c1) edge (d3);
         \draw[targetedge-new] (c3) edge (d2);
         \draw[targetedge-new] (c3) edge (d3);
         \end{tikzpicture}
         \caption{A $ \postoneg $ gadget.}
        \label{fig:diag-pos-to-neg_gadget}
    \end{subfigure}
    \hfill
    \begin{subfigure}[t]{0.3\textwidth}
        \centering
        \begin{tikzpicture}[scale=0.7]
         \node[hcdnode-big, label={[above=-0.5mm]:{\scriptsize $u$}}] (v1) at (0,1) {};
         \node[hcdnode-big, label={[above=-0.5mm]:{\scriptsize $ v $}}] (v2) at (2,1) {};
         \node[hcdnode-big, label={[below=2.5mm]:{\scriptsize $ w $}}] (v3) at (2,0) {};
         \node[hcdnode-big, label={[below=2.5mm]:{\scriptsize $ x $}}] (v4) at (0,0) {};

         \draw[hcdedge] (v1) edge (v2);
         \draw[hcdedge] (v2) edge (v3);
         \draw[hcdedge] (v3) edge (v4);
         \draw[hcdedge] (v4) edge (v1);

         \end{tikzpicture}
         \caption{A graph $ G $ with Hamiltonian cycle $ u,v,w,x,u $.}
         \label{fig:diag-pos-to-neg_cex}
    \end{subfigure}
    \hfill
    \begin{subfigure}[t]{0.3\textwidth}
        \centering
        \begin{tikzpicture}[scale=0.65]
         \node[targetnode-old-big] (11) at (0,1) {};
         \node[targetnode-old-big] (12) at (1,1) {};
         \node[targetnode-old-big] (13) at (2,1) {};
         \node[targetnode-old-big] (21) at (0,0) {};
         \node[targetnode-old-big] (22) at (1,0) {};
         \node[targetnode-old-big] (23) at (2,0) {};
         \node[targetnode-old-big] (31) at (4,0) {};
         \node[targetnode-old-big] (32) at (5,0) {};
         \node[targetnode-old-big] (33) at (6,0) {};
         \node[targetnode-old-big] (41) at (4,1) {};
         \node[targetnode-old-big] (42) at (5,1) {};
         \node[targetnode-old-big] (43) at (6,1) {};

         \draw[targetedge-old] (11) edge (12);
         \draw[targetedge-old] (21) edge (22);
         \draw[targetedge-old] (31) edge (32);
         \draw[targetedge-old] (41) edge (42);

         \draw[targetedge-new] (11) edge (23);
         \draw[targetedge-new] (13) edge (22);
         \draw[targetedge-new] (13) edge (23);

         \draw[targetedge-new] (21) edge[bend right] (33);
         \draw[targetedge-new] (23) edge[bend right] (32);
         \draw[targetedge-new] (23) edge[bend right] (33);

         \draw[targetedge-new] (31) edge (43);
         \draw[targetedge-new] (33) edge (42);
         \draw[targetedge-new] (33) edge (43);

         \draw[targetedge-new] (41) edge[bend right] (13);
         \draw[targetedge-new] (43) edge[bend right] (12);
         \draw[targetedge-new] (43) edge[bend right] (13);
         \end{tikzpicture}\\
         \caption{Graph $ G^\star $ without a Hamiltonian cycle resulting from applying the gadget to the example graph.}
         \label{fig:diag-pos-to-neg_applied}
    \end{subfigure}
    \caption{An example of a \postoneg gadget in which the node graphs are not paths and the edge $ (1_\triangleright,3_\triangleleft) $ is present.}
    \label{fig:diag-pos-to-neg}
\end{figure}

\FloatBarrier

\subsection{Proofs of Section \ref{section:algorithm-templates}: Decidable cases for classes of (fixed) algorithmic problems}

Let us start with a remark. If a class $\calR$ of reductions is finite, then \redgen[\calR][P][P^\star] is decidable for all algorithmic problems \problem and~\problembis. Thus, in particular, decidability follows for the class of quantifier-free first-order interpretations with dimension bounded by some $r > 0$, as for a fixed schema there are only finitely many different quantifier-free formulas, up to equivalence, and equivalence of quantifier-free formulas is decidable. 
 
Unfortunately, this reasoning fails for bounded-arity cookbook reductions, as unboundedly many new elements can be created for each type. So, there are infinitely many cookbook reductions with arity bounded by some $r > 0$.

\subsubsection{\texorpdfstring{Interpreting recipes with $\FO$-interpretations}{Interpreting recipes with FO-interpretations}}

  \lemmaRecipeFO*
  \begin{proofsketch}
    The interpretation $\interrecipe$ has dimension $r+2$. Sets of elements are encoded in the same way as in the proof of Theorem~\ref{th:cookbook_is_qf}, on the first $r+1$ variables. From there,
    \begin{itemize}
    \item the universe formula $\formule_U(x_1,\ldots,x_{r+1},y)$ is the disjunction, over every $\isotype\in\typesr$, of formulas stating that the set represented by $(x_1,\ldots,x_{r+1})$ has type $\isotype$ in $\struct$ (and in particular, that no $x_i$ belongs to any $\colt[\isotype']$), and that $\colt(y)$ holds,
    \item the equivalence formula $\formule_\sim(x_1,\ldots,x_{r+1},y;x'_1,\ldots,x'_{r+1},y')$ is the disjunction, over every $\isotype\in\typesr$ and every respective subsets $S, S'$ of the sets represented by $(x_1,\ldots,x_{r+1})$ and $(x'_1,\ldots,x'_{r+1})$, of the formulas stating that $S=S'$ have type $\isotype$ in $\struct$, and that there exists some element $z$ such that $\colt(z)$, for which $y=z$ or $y\inh z$, and $y'=z$ or $y'\inh z$, and
    \item for every relation $\relsymbbis\in\vocabbis$ of arity $k$, the formula $\formule_{\relsymbbis}(\bar x^{1},y^1;\ldots;\bar x^k,y^k)$ is the disjunction, over every $\isotype\in\typesr$, of formulas stating that $\bar x^1,\ldots,\bar x^k$ represent the same set, which has type $\isotype$ in $\struct$, and that $\redS(\isotype)\models \relsymbbis(y^1,\ldots,y^k)$.
  \end{itemize}

  We leave it to the reader to check that for every \vocab-structure $\struct$ and every \red of arity at most $r$, the following function $\bijrecipe$ from $\interrecipe(\struct\uplus\recipe)$ to $\red(\struct)$ is well defined, and indeed an isomorphism of \vocabbis-structures:

  Each element of $\interrecipe(\struct\uplus\recipe)$ is an equivalence class for $\sim$, and, by definition of $\formule_\sim$, contains an element $(x_1,\ldots,x_{r+1},y)$ such that the set $S$ represented by $(x_1,\ldots,x_{r+1})$ is included in every set represented by $(x'_1,\ldots,x'_{r+1})$ where $(x'_1,\ldots,x'_{r+1},y')$ in the class.
  Then $\bijrecipe$ maps this class to the element $(S,y)$ of $\red(\struct)$.
\end{proofsketch}

\subsubsection{Proof that \FO-interpretations do not preserve MSO-similarity}
Here is an example that \FO-interpretations (even quantifier-free ones) 
of dimension greater than one do not necessarily preserve \MSO-similarity.

\newcommand{\structn}{\struct_n}
\newcommand{\structbisn}{\structbis_n}

\begin{proposition}
  Let $\vocab:=\{\succrel^{(2)},\colone^{(1)},\coltwo^{(1)}\}$.

  There exists a quantifier-free \FO-interpretation $\inter$ of $\vocab$-structures, an integer $k$ and families of $\vocab$-structures $(\structn)_{n\in\N},(\structbisn)_{n\in\N}$ such that:

  \begin{itemize}
  \item $\forall n\in\N,\quad \structn\msoeq{n}\structbisn\,,$
  \item $\forall n\in\N,\quad \inter(\structn)\not\msoeq{k}\inter(\structbisn)\,.$
  \end{itemize}
\end{proposition}

\begin{proof}

  Consider the non-regular language $\lang \df \{a^pb^p:p\in\N\}$ and the regular language $L' \df L(a^*b^*)$. Recall that a language is regular iff it is definable by an \MSO-formula which may use unary relations corresponding to the letter of the alphabet and a binary predicate $\succrel$ encoding successive positions~\cite{buchi1960weak}. Therefore, for every $n\in\N$, there are $\sigma$-structures $\calA_n \in L$ and $\calA'_n \in L(a^*b^*)\setminus\lang$ such that $\structn\msoeq{n}\structbisn$. Here and in the following we identify $\sigma$-structures and words over $\Sigma = \{a, b\}$ represented by them.
  
  Now consider the two-dimensional \FO-interpretation $\inter$ defined as follows:
  \begin{itemize}
  \item its universe formula is $\formule_\text{U}(x,y)\df\top\,,$
  \item the formula defining $\succrel$ is \[\formule_\succrel(x,y;x',y')\df(x=x'\land S(y,y'))\lor(S(x',x)\land y=y')\,,\]
  \item the formula defining $\colone$ is $\formule_{\colone}(x,y)\df(\colone(x)\land\coltwo(y))\lor(\coltwo(x)\land\colone(y))\,,$
  \item the formula defining $\coltwo$ is $\formule_{\coltwo}(x,y)\df(\colone(x)\land\colone(y))\lor(\coltwo(x)\land\coltwo(y))\,.$
  \end{itemize}

    \begin{figure}[ht!]
    \centering
    \begin{tabular}{c|cccccc}
      &a&$\cdots$& a&b&$\cdots$&b\\
      \hline
      b&a&$\cdots$& a&b&$\cdots$&b\\
      $\vdots$&$\vdots$&$\ddots$&$\vdots$&$\vdots$&$\ddots$&$\vdots$\\
      b&a&$\cdots$& a&b&$\cdots$&b\\
      a&b&$\cdots$& b&a&$\cdots$&a\\
      $\vdots$&$\vdots$&$\ddots$&$\vdots$&$\vdots$&$\ddots$&$\vdots$\\
      a&b&$\cdots$& b&a&$\cdots$&a\\
    \end{tabular}
    \caption{Illustration of $\inter(\struct)$ (in the bottom-right corner of the tabular), for $\struct\in a^*b^*$. In $\inter(\struct)$, $\succrel$ links an element of the grid to its right and bottom neighbors.}
    \label{fig:inter}
  \end{figure}

  The result of $\inter$ for a word $\calA$ of the form $a^*b^*$ is illustrated in Figure~\ref{fig:inter}. Note that $\inter(\struct)$ is an $\succrel$-square-grid, and the $\succrel$-diagonal contains an element colored in $\coltwo$ iff there is a different number of $a$'s and $b$'s in $\struct$.

  \newcommand{\diagform}{\Phi_\text{diag}}

  Let us construct an \MSO-sentence $\Phi$ stating precisely that the $\succrel$-diagonal contains only elements in $\colone$.

  First, consider the following formula $\diagform(X)$, stating, in a $\succrel$-square-grid, that $X$ contains the $\succrel$-diagonal

  \begin{tabular}{l}
    $\diagform\df\forall x\Big( \neg(\exists y \succrel(y,x))\to X(x) \Big)$\\
    $\quad\quad\quad\quad\quad\land \;\forall x \forall y \Big(X(x)\land \exists u \exists v\big(u\neq v\land\succrel(x,u)\land\succrel(x,v)\land\succrel(u,y)\land\succrel(v,y)\big) \to X(y)\Big).$
  \end{tabular}
  
  The first line states that the top-left element of the diagonal is contained in $x$. The second line states that $X$ is closed under going one step right and one down.

  Consider now the formula \[\Phi\df\exists D\Big(\ \diagform(D)\ \land\ \forall x\big(D(x)\to\colone(x)\big)\Big)\]
  Given $\struct\in L(a^*b^*)$, we have that $\inter(\struct)\models\Phi$ iff its $\succrel$-diagonal contains only $a$'s, i.e. iff $\struct\in\lang$.
  It only remains to pick $k$ as the quantifier rank of $\Phi$, and our families $(\structn)_n$ and $(\structbisn)_n$ are suitable witnesses.
\end{proof}

\subsubsection{Algorithmic templates for \MSO, beyond edge gadget reductions}\label{app:mso-fixed}

\theoremAlgorithmTemplatesMSO*

\begin{proof}
  
  Since we are only considering edge gadget reductions, we adapt the way reductions are specified as input. Instead of the full \recipe, edge gadget reductions are only specified via their edge gadget graph $\gadget_\red=\redS(\inlineEdgeType)$, where the two endpoints are marked with the unary symbol $R_1$ and $R_2$.

  Our proof approach is the same as for Theorem \ref{theorem:algorthmic-templates}. We show that whether an edge gadget reduction $\rho$ is a reduction from $P$ to $P^\star$ solely depends on the $\qdMSO[m]$-type of $\gadget_\rho$, for some $m$ large enough and depending only on $P$ and $P^\star$. As there are only finitely many such $\qdMSO[m]$-types and because the type of $\gadget_\rho$ can be determined, the statement follows.

  \newcommand{\source}{\text{Inc}^\text{source}}
  \newcommand{\target}{\text{Inc}^\text{target}}
  \newcommand{\selem}[1]{\langle #1 \rangle}

  Let $G$ be a graph. We start by assigning an arbitrary direction to each edge of $G$, and we consider the structure $G^\text{inc}$, representing the incidence graph of $G$, over the vocabulary $\{\source,\target,\text{Edge}\}$, where $\source$ and $\target$ are binary and $\text{Edge}$ is unary, defined as follows:
  \begin{itemize}
  \item its universe has one element for each node of $G$, and one element for each (oriented) edge of $G$,
  \item the unary relation $\text{Edge}$ marks all elements corresponding to edges of $G$,
  \item $\source(e,v)$ (resp. $\target(e,v)$) holds if the node $v$ is the source (resp. target) of the edge $e$ in $G$.
  \end{itemize}
  
We can express the result of replacing every edge of $G$ with $\gadget_\rho$ as a so-called generalized sum.
  
\begin{definition}[\cite{shelah1975monadic}, formulation following \cite{blumensath2008logical}]\label{def:generalsum}
Let $\calI = (I,S_1,\ldots, S_r)$ be a structure and $(\calD_i)_{i \in I}$ a sequence of structures $\calD_i = (D_i, R^i_1, \ldots, R^i_t)$ indexed by elements $i$ of $\calI$.
The \emph{generalized sum} of $(\calD_i)_{i \in I}$ is the structure \[ \sum_{i \in I} \calD_i \df (U, \sim, R'_1, \ldots, R'_t, S'_1, \ldots, S'_r)\]
with universe $U \df \{\selem{i,a} \mid i \in I, a \in D_i\}$ and relations
\begin{itemize}
 \item $\selem{i,a} \sim \selem{i',a'}$ if and only if $i = i'$
 \item $R'_j \df \{( \selem{i,a_1},\ldots, \selem{i,a_\ell} ) \mid i \in I, (a_1, \ldots, a_\ell) \in R^i_j\}$
 \item $S'_j \df \{( \selem{i_1,a_1},\ldots,\selem{i_\ell,a_\ell} ) \mid (i_1, \ldots, i_\ell) \in S_j, a_k \in D_{i_k}$ for all $k \in \{1, \ldots, \ell\} \}$.
\end{itemize}
\end{definition}

The structures $\calI$ and $\calD_i$ in this definition are also referred to as \emph{index structure} and \emph{component structures}, respectively. 
We now consider the generalized sum $\sum_{e\in G^{\text{inc}}}\gadget_\rho$ of copies of $\gadget_\rho$ indexed by $G^{\text{inc}}$:
  \begin{itemize}
  \item elements of $\sum_{e\in G^{\text{inc}}}\gadget_\rho$ are of the form $\langle a,b\rangle$ with $a\in G^{\text{inc}}$ and $b\in \gadget_\rho$,
  \item $\sum_{e\in G^{\text{inc}}}\gadget_\rho$ inherits the relation $\text{Edge}$ from $G^{\text{inc}}$ on the first coordinate, as well as $\source$ and $\target$: for $R\in\{\source,\target\}, R(\langle a,b\rangle,\langle a',b'\rangle)\text{ if and only if }R(a,a')$,
  \item $\sum_{e\in G^{\text{inc}}}\gadget_\rho$ inherits on the second coordinate relations $R_1$ and $R_2$ from $\gadget_\rho$, as well as $E$ when the first coordinate is fixed: $E(\langle a,b\rangle,\langle a',b'\rangle)\text{ if and only if }E(b,b')$ and $a=a'$.
  \end{itemize}

  There exists a $1$-dimensional \FO-interpretation of quantifier depth $1$ that yields $\rho(G)$ on $\sum_{e\in G^{\text{inc}}}\gadget_\rho$, for every $G$ and $\rho$:
  \begin{itemize}
  \item its universe consists of all $\langle a,b\rangle$ such that (i) $a$ is an edge and $b$ is any element of $\gadget_\rho$, or (ii) $a$ is an isolated node and $b$ is the element such that $R_1(b)$,
  \item the edge relation is taken from $\gadget_\rho$,
  \item the endpoints of gadgets that correspond to the same node in $\rho(G)$ are identified via the formula
  \[
  \begin{aligned}
    \varphi_\sim& (\langle x,y\rangle,\langle x',y'\rangle)\df\exists \langle x'',y''\rangle \\
    & \big[\source(\langle x,y\rangle,\langle x'',y''\rangle)\land R_1(\langle x,y\rangle) \lor \target(\langle x,y\rangle,\langle x'',y''\rangle)\land R_2(\langle x,y\rangle)\big]\\
    \land&\big[\source(\langle x',y'\rangle,\langle x'',y''\rangle)\land R_1(\langle x',y'\rangle) \lor \target(\langle x',y'\rangle,\langle x'',y''\rangle)\land R_2(\langle x',y'\rangle)\big]\,.\\
  \end{aligned}
  \]
  \end{itemize}
  
Generalized sums are \MSO-compatible, as witnessed by the following fact.
  
\begin{lemma}[{\cite{shelah1975monadic}, formulation following \cite[Theorem 3.16]{blumensath2008logical}}] \label{theorem:shelah}
From every \MSO sentence~$\varphi$, a finite sequence $\chi_0, \ldots, \chi_{s-1}$ of \MSO formulas and an \MSO formula $\psi$ can be constructed such that for every graph $G$ and every edge gadget reduction $\rho$,
\[\sum_{e\in G^{\text{inc}}}\gadget_\rho\models\varphi\quad\text{if and only if}\quad(G^{\text{inc}},B_0,\ldots,B_{s-1})\models\psi\,,\]
  where each $B_i$ is a propositional variable (i.e., a $0$-ary relation) which is true if and only if $\gadget_\rho\models\chi_i$. 
\end{lemma}

  Let us denote by $\nu(\varphi)$ the maximum among the quantifier ranks of formulas $\psi, \chi_0,\ldots,\chi_{s-1}$ that Lemma~\ref{theorem:shelah} yields for formula $\varphi$. Let $k$ be the quantifier rank of an \MSO formula describing $\problembis$. Consider the formulas $\varphi_0,\ldots,\varphi_{p-1}$, each characterizing one of the finitely many $\qdMSO[k+1]$-types of structures over the schema of the generalized sums, and let $m$ be the maximum of the $\nu(\varphi_i)$, for $0\leq i<p$.

  Assume that $\gadget_\rho\msoeq{m}\gadget_{\rho'}$. By definition of $m$, $\gadget_\rho$ and $\gadget_{\rho'}$ agree on all the formulas $\chi_j$ given by Lemma~\ref{theorem:shelah} for $\varphi_i$, for all $0\leq i<p$. Thus, for every graph $G$, $\sum_{e\in G^{\text{inc}}}\gadget_\rho\models\varphi_i$ if and only if $\sum_{e\in G^{\text{inc}}}\gadget_{\rho'}\models\varphi_i$ for every $i$, meaning that $\sum_{e\in G^{\text{inc}}}\gadget_\rho$ and $\sum_{e\in G^{\text{inc}}}\gadget_{\rho'}$ have the same $\qdMSO[k+1]$-type. It follows that $\rho(G)$ and $\rho'(G)$ have the same $\qdMSO[k]$-type and that $\rho(G)\in\problembis$ if and only if $\rho'(G)\in\problembis$. In other words, whether \red is a reduction from \problem to \problembis only depends on the $\qdMSO[m]$-type of $\gadget_\rho$. Given that only finitely many such types exist, one can compute the $\qdMSO[m]$-type of $\gadget_\rho$ and match it against the list of types of valid reductions from \problem to \problembis.
\end{proof}

As indicated in the main part, the proof idea can be generalized beyond edge gadget reductions. The class of \indepgadg reductions is obtained from the class of cookbook reductions by forbidding inheritance, except from types of arity 1. Formally, a cookbook reduction $\red$ is a \emph{\indepgadg reduction} if
\begin{itemize}
\item $\red$ has no global elements (i.e. $\fresh{\isotype}=0$ for the type $\isotype$ of arity $0$),
\item for every type $\isotype$ of arity $1$, $\fresh{\isotype}=1$,
\item for every type $\isotype$ such that $\fresh{\isotype}>0$, and for every subtype $\isotype'$ of $\isotype$ of arity other than $1$, $\fresh{\isotype'}=0$.
\end{itemize}
Note that edge gadget reductions are a particular case of \indepgadg reductions.

\begin{theorem}
  Let $\calR$ be the class of \indepgadg reductions of arity at most $r$, for some $r>0$, let $P$ be an arbitrary algorithmic problem and $P^\star$ an algorithmic problem definable in monadic second-order logic. Then \redgen[\calR][\problem][\problembis] is decidable.
\end{theorem}

\begin{proof}
  Let $k$ be the depth of an \MSO-sentence defining \problembis.
  We prove the following: there exists $m\in\N$ (depending on \vocab, \vocabbis, $r$ and $k$) such that whether a \indepgadg reduction $\red$ is a valid instance of \redgen[\calR][\problem][\problembis] only depends on the $\qdMSO[m]$-type of $\recipe$.

  In order to break symmetries, we consider ordered types here: we add $r$ constants $c_1,\ldots,c_r$ to the vocabulary, which are used to order the elements of a set. We will only consider the types of $r$ elements or less. We pick an arbitrary ordered representative $\dot \isotype$ for every $\isotype\in\typesr$, and we denote by \ordtypesr the set of all these representatives.

  Recall that a cookbook reduction $\rho$ is specified as its recipe structure \recipe. We slightly modify the previous definition of recipes to take into account ordered types instead of types: the recipe \recipe of an \indepgadg reduction $\rho$ is a structure over the schema $\vocabbis\cup\{R_1,\ldots,R_r\}\cup\{C_{\dot \isotype} \mid \dot \isotype\in\ordtypesr\}$, where the $R_i$ and the $C_{\dot \isotype}$ are unary. 

  The restriction of \recipe to the schema $\vocabbis\cup\{C_{\dot \isotype} \mid \dot \isotype\in\ordtypesr\}$ is the disjoint union $\biguplus_{\isotype\in\typesr}\redS(\isotype)$, where each $C_{\dot \isotype}$ is interpreted as the universe of $\redS(\isotype)$.   
  The definition of \indepgadg ensures that we can do the following: for every $\dot \isotype\in\ordtypesr$ of arity $l$, we add one element of $\redS(\isotype)$ to each $R_i$, for $1\leq i\leq l$, namely the element which is inherited by the $i$-th element of $\dot \isotype$.

 Let $\struct$ be a $\vocab$-structure $\struct$. We consider the structure $\extstruct$ over the schema $\{P_1,\ldots,P_r\}\cup\{C'_{\dot \isotype} \mid \dot \isotype\in\ordtypesr\}$ (where the $P_i$ are binary relation symbols, and the $C'_{\dot \isotype}$ are unary) which collects the sets of size at most $r$ of elements of $\struct$ together with their type in $\struct$: for every set $\{a_1,\ldots,a_l\}\subseteq\structdom$, with $l\leq r$, add to the universe of $\extstruct$ one arbitrary tuple $(a_{\xi(1)},\ldots,a_{\xi(n)})$ of ordered type $\dot \isotype$. For any $1\leq i\leq r$, $(x,y)$ belongs to the interpretation of $P_i$ if $x$ is a $1$-tuple, whose element is the $i$-th element of tuple $y$.
 Finally, $C'_{\dot \isotype}$ is interpreted as the set of tuples $(a_1,\ldots,a_l)$ such that $\mathfrak{tp}_\struct(a_1,\ldots,a_l)=\dot \isotype$.

 We consider the generalized sum $\sum_{e\in\extstruct}\recipe$ of copies of \recipe, where \extstruct plays the role of the index structure. For convenience, and since the schemas of \extstruct and \recipe are disjoint, we will consider $\sum_{e\in\extstruct}\recipe$ as a structure on the union of both schemas. Recall that
 \begin{itemize}
 \item elements of $\sum_{e\in\extstruct}\recipe$ are of the form $\langle a,b\rangle$ with $a\in\extstruct$ and $b\in\recipe$,
 \item $\sum_{e\in\extstruct}\recipe$ inherits relations from \extstruct on the first coordinate: for every relation $R$ on the schema of \extstruct, \[R(\langle a_1,b_1\rangle,\ldots,\langle a_l,b_l\rangle)\text{ if and only if }R(a_1,\ldots,a_l)\]
 \item $\sum_{e\in\extstruct}\recipe$ inherits relations from \extstruct on the second coordinate, when the first is fixed: for every relation $R$ on the schema of \recipe, \[R(\langle a_1,b_1\rangle,\ldots,\langle a_l,b_l\rangle)\text{ if and only if }R(b_1,\ldots,b_l)\land a_1=\cdots=a_l\,.\]
 \end{itemize}
 
  Let us explain how $\red(\struct)$ can be interpreted in $\sum_{e\in\extstruct}\recipe$, via an \FO-interpretation of dimension $1$ and quantifier depth $1$, depending only on schemas $\vocab,\vocabbis$ and $r$: 
  \begin{itemize}
  \item its universe formula is $\formule_U(\langle x,y\rangle)\df\bigvee_{\dot \isotype\in\ordtypesr}C'_{\dot \isotype}(x)\land C_{\dot \isotype}(y)$,
  \item its equivalence formula is
    \[\formule_\sim(\langle x,y\rangle,\langle x',y'\rangle):=\exists\langle x'',y''\rangle, \bigvee_{1\leq i,j\leq r}P_i(x'',x)\land P_j(x'',x')\land R_i(y) \land R_j(y')\]
  \item all relations are inherited without modification.
  \end{itemize}

We use again Lemma~\ref{theorem:shelah}, with $\extstruct$ replacing $G^{\text{inc}}$ and $\recipe$ replacing $\gadget_\rho$.
  Let us denote by $\nu(\varphi)$ the maximum among the quantifier ranks of formulas $\psi, \chi_0,\ldots,\chi_{s-1}$ that Lemma~\ref{theorem:shelah} yields for formula $\varphi$. Consider formulas $\varphi_0,\ldots,\varphi_{p-1}$, each of which characterizing one of the $p$ different $\qdMSO[k+1]$-types of structures over the schema $\vocabbis$, and let $m$ be the maximum of the $\nu(\varphi_i)$, for $0\leq i<p$.

  Assume that $\recipe\msoeq{m}\recipebis$: by definition of $m$, \recipe and \recipebis agree on all the sentences $\chi_j$ given by Lemma~\ref{theorem:shelah} for all the $\varphi_i$, $0\leq i<p$. Thus, for every \vocab-structure $\struct$, $\sum_{e\in\extstruct}\recipe\models\varphi_i$ if and only if $\sum_{e\in\extstruct}\recipebis\models\varphi_i$ for every $i$, meaning that $\sum_{e\in\extstruct}\recipe$ and $\sum_{e\in\extstruct}\recipebis$ have the same $\qdMSO[k+1]$-type, which in turn entails that $\red(\struct)\in\problembis$ if and only if $\redbis(\struct)\in\problembis$. In other words, whether \red is a reduction from \problem to \problembis only depends on the $\qdMSO[m]$-type of \recipe: given that only finitely many such types exist, one can compute the $\qdMSO[m]$-type of \recipe and match it against the list of types of valid reductions from \problem to \problembis.
\end{proof}

\subsection{Proofs of Section \ref{section:formulas-as-inputs}: Algorithmic problems as input: decidable cases}

\theoremProblemsAsInputsDecidability*
\begin{proof}
	Recall that whether a quantifier-free interpretation $\calI$ is a reduction from the algorithmic problem defined by $\varphi \in \calL$ to the one defined by $\varphi^\star \in \calL^\star$ is equivalent to the question whether $\form\leftrightarrow\interinv(\formbis)$ is a tautology. 

  Towards proving part (1), let $\form,\formbis\in\EFO$ and $\inter$ be a quantifier-free interpretation. Then $\psi^\star \df \interinv(\formbis)$ is in \EFO, and therefore $\form\leftrightarrow\interinv(\formbis) \equiv \big(\neg \form \wedge \neg \psi^\star\big) \vee \big(\form \wedge \psi^\star\big)$ can be translated into a $\forall^*\exists^*\FO$-formula. As tautologies can be tested for this first-order fragment (see, e.g.,\cite{BorgerGG2001}), the statement follows.

  Towards proving part (2), we distinguish whether $\problem$ is definable in $\EFO$, or not. For each $\problem$, we happen to know which case holds, since $\problem$ is not part of the input. 

 If $\problem\notin\EFO$, then there is no positive instance of \redqf[\problem][\EFO]. Indeed, if there was such a positive instance with $\problembis \in \EFO$ and $\inter \in \QF$, then $\problem$ could be described by the $\EFO$-formula $\interinv(\formbis)$, where $\formbis$ describes $\problembis$. Thus, in this case, an algorithm can answer ``no'' for every input.
 
 If $\problem \in \EFO$, let $\form$ be an $\EFO$-formula describing \problem and apply part (a)

  Towards proving part (3), fix a formula $\formbis\in\MSO$ describing a problem $\problembis$ and let $k$ be the quantifier rank $\formbis$. Now, as in the proof of Theorem~\ref{th:fSO_fMSO}, given a graph $G$ and a an edge gadget reduction $\rho$ with edge gadget $\gadget$ as input, whether $\rho(G) \models\formbis$ only depends on the $\qdMSO[k]$-type of $\gadget$. 

  \newcommand{\gset}{A_\type}

  Consider the following set of graphs for every $\qdMSO[k]$-type $\type$:
  \[\gset:=\{G \mid \rho(G) \models\formbis\text{ for any (thus, every)  reduction $\rho$ with edge gadget }\gadget\text{ of type }\type \}\,.\]
  Since $\formbis$ is fixed and there are only finitely many such types $\type$, we can have access to a function $\calF$ which maps every $\type$ to the answer to the question ``is $\gset$ \EFO-definable?'', together with a witness \EFO-sentence $\psi_\tau$ when this is the case.

  An algorithm for \redgen[\calR][\EFO][\problembis] can now work as follows given a formula $\form$ and edge gadget reduction $\rho$ with edge gadget $\gadget$ as input. It starts by computing the $\qdMSO[k]$-type $\type$ of $\gadget$. Then it looks at $\calF$ and does the following:
  \begin{enumerate}
  \item If $\calF$ says that $\gset$ is not \EFO-definable, the algorithm answers ``no''. 
  \item If $\calF$ says that $\gset$ is defined by some $\psi_\tau\in\EFO$, then the algorithm calls the algorithm of part (a) as subroutine with $\form$, $\psi_\tau$  and $\rho$  as input

  \end{enumerate}
  
  The answer is correct in case (1), because if $(\form,\gadget)$ was a yes-instance, then we would have by definition $\forall\graphe,\graphe\models\form$ iff $\gexp\models\formbis$, and $\form$ would define $\gset$. It is correct in case (2), because the algorithm for part (a) is correct.%
\end{proof}
  
\end{document}